\documentclass{lmcs} 


\usepackage{caption}
\usepackage{subcaption}
\usepackage{amsthm}
\usepackage{etex}
\usepackage{verbatim}
\usepackage{fancyhdr}
\usepackage{fancybox}
\usepackage{stmaryrd}
\usepackage[reqno]{amsmath}
\usepackage{mathtools}

\usepackage{amssymb}
\usepackage{amsfonts}
\usepackage{amsbsy}
\usepackage{xypic}
\usepackage{graphicx}
\usepackage{proof}
\usepackage{bbding}
\usepackage{pgf}
\usepackage{cancel}
\usepackage{semantic}
\usepackage{extarrows}
\usepackage[noend]{algpseudocode}
\usepackage{algorithmicx,algorithm}
\usepackage[normalem]{ulem}
\usepackage{mathrsfs}
\usepackage{hyperref}
\usepackage{wasysym}
\usepackage{multirow}
\usepackage{array}
\usepackage{diagbox}
\usepackage[algo2e,ruled, lined, linesnumbered, commentsnumbered, longend]{algorithm2e}

\hypersetup
{
	pdfauthor = UnnamedOrange,
	colorlinks = true,
	linkcolor = blue,
	anchorcolor = blue,
	citecolor = blue,
	urlcolor = blue
}

\usepackage{tikz}
\usetikzlibrary{positioning,arrows,automata}
\usetikzlibrary{arrows}
\tikzstyle{block}=[draw opacity=0.7,line width=1.4cm]
\usepackage{amssymb}
\usetikzlibrary {fit,shapes.geometric}
\usetikzlibrary {arrows.meta}
\usetikzlibrary {quotes}
\usetikzlibrary {positioning}
\usetikzlibrary {calc}
\usetikzlibrary{backgrounds}

\newcommand{\defn}{\stackrel{\text{def}}{=}}
\def\PRCCS{\mathcal{P}_{\mathrm{RCCS}_{fs}}}

\newtheorem{theorem}{Theorem}[section]
\newtheorem{proposition}[theorem]{Proposition}
\newtheorem{lemma}[theorem]{Lemma}
\newtheorem{corollary}[theorem]{Corollary}

\theoremstyle{definition}
\newtheorem{definition}[theorem]{Definition}
\newtheorem{example}[theorem]{Example}

\theoremstyle{remark}
\newtheorem{remark}[theorem]{Remark}

\begin{document}
\title[Analyzing Divergence for Nondeterministic Probabilistic Models]{Analyzing Divergence for Nondeterministic Probabilistic Models}

\author[H.~Wu\lmcsorcid{0000-0002-8955-8216}]{Hao Wu}[a]
\author[Y.~Fu]{Yuxi Fu}[b]
\author[H.~Long\lmcsorcid{0000-0002-1328-6197}]{Huan long}[b]
\author[X.~Xu\lmcsorcid{0000-0001-9713-9751}]{Xian Xu}[c]
\author[W.~Zhang\lmcsorcid{0000−0003−1375−0081}]{Wenbo Zhang}[d]

\address{Shanghai Maritime University,  Shanghai, China}	
\email{wuhao@shmtu.edu.cn}  

\address{BASICS, Shanghai Jiao Tong University, Shanghai, China}	
\email{fu-yx@cs.sjtu.edu.cn (Corresponding author), longhuan@sjtu.edu.cn}  

\address{East China University of Science and Technology,  Shanghai, China}	
\email{xuxian@ecust.edu.cn}  

\address{Shanghai Ocean University, Shanghai, China}	
\email{wbzhang@shou.edu.cn}  





\begin{abstract}
	Branching and weak probabilistic bisimilarities are two well-known notions capturing behavioral equivalence between nondeterministic probabilistic systems.
	For probabilistic systems, divergence is of major concern.
	Recently several divergence-sensitive refinements of branching and weak probabilistic bisimilarities have been proposed in the literature.
	Both the definitions of these equivalences and the techniques to investigate them differ significantly.
	This paper presents a comprehensive comparative study on divergence-sensitive behavioral equivalence relations that refine branching and weak probabilistic bisimilarity.
	Additionally, these equivalence relations are shown to have efficient checking algorithms.
	The techniques of this paper may be of independent interest in a more general setting.
\end{abstract}

\maketitle		
	
	\section{Introduction}\label{sec-introduction}
	\subsubsection*{Background and Motivation}
	In the area of program analysis, probability and nondeterminism have received significant attention in recent years \cite{fu_TerminationNondeterministicProbabilistic_2019,chatterjee_TerminationNondeterministicRecursive_2017,etessami_RecursiveMarkovDecision_2015}.
	Many different nondeterministic probabilistic models have been studied from both theoretical and practical perspectives, such as 
	Markov decision processes (MDP) \cite{baier_PrinciplesModelChecking_2008,etessami_RecursiveMarkovDecision_2015,brazdil_ReachabilityRecursiveMarkov_2008}, 
	Probabilistic automata (PA) 
	\cite{segala_ModelingVerificationRandomized_1995, cattani_DecisionAlgorithmsProbabilistic_2002, turrini_PolynomialTimeDecision_2015}, 
	Randomized CCS ($\mathrm{RCCS}$)
	\cite{fu_ModelIndependentApproach_2021,zhang_UniformRandomProcess_2019, wu_ProbabilisticWeakBisimulation_2023}, etc..
	For these models, a fundamental question is how to define behavioral equivalence between probabilistic systems.
	Variants of equivalence for these nondeterministic probabilistic models have already been studied over the years, including strong bisimulation \cite{larsen_BisimulationProbabilisticTesting_1989,vanglabbeek_ReactiveGenerativeStratified_1995,castiglioni_ModalDecompositionNondeterministic_2016}, weak bisimulation \cite{baier_WeakBisimulationFully_1997, deng_AxiomatizationsProbabilisticFinitestate_2007,wu_ProbabilisticWeakBisimulation_2023,segala_ModelingVerificationRandomized_1995}, branching bisimulation \cite{segala_ProbabilisticSimulationsProbabilistic_1994,castiglioni_RaidersLostEquivalence_2020,castiglioni_ProbabilisticDivideCongruence_2020,fu_ModelIndependentApproach_2021}, trace equivalence \cite{jou_EquivalencesCongruencesComplete_1990a} and testing equivalence \cite{larsen_BisimulationProbabilisticTesting_1989,yi_TestingProbabilisticNondeterministic_1992,deng_TestingFinitaryProbabilistic_2009}.
	Among them probabilistic branching and weak bisimulations are of great importance.
	Their non-probabilistic versions have been intensively studied in the linear-time branching-time spectrum by van Glabbeek \cite{vanglabbeek_LinearTimeBranching_1993}. 
	Traditional branching and weak bisimulations  ignore the role of divergence, i.e., 
	infinite sequences of internal computation steps need not be bisimulated. However, divergence is crucial in practice as a non-terminating computation could be unintended in many applications. As it turns out, system behaviors become far more complicated when divergence is an issue.
	Liu et al. \cite{liu_AnalyzingDivergenceBisimulation_2017} have demonstrated the importance of divergence for non-probabilistic processes in system verification. They put forward divergence-sensitive branching and weak bisimilarities in the non-probabilistic setting, and give equivalent characterizations for them.
	
	There have been mainly two ways to capture divergence in the nondeterministic probabilistic models.
	The first one is defined by the existence of a \emph{divergent $\epsilon$-tree} (roughly, the probabilistic version of state-preserving internal action sequences)~\cite{fu_ModelIndependentApproach_2021}. The second one is defined by the reachability to a \emph{$\tau$-EC} (roughly, the probabilistic version of the internal action cycle)~\cite{he_DivergencesensitiveWeakProbabilistic_2023}.
	Although the two concepts are defined in the context of probabilistic branching and weak bisimulations respectively, they are actually independent of specific bisimulation semantics.
	
	We give an example to explain the motivation  of our work.
	In Figure \ref{fig-motivation}, $S$ is the specification of a probabilistic system, and $P_1, P_2$ are two implementation candidates. 
	We would like to tell whether $P_1$ and $P_2$ implement $S$ faithfully. In probabilistic program analysis, \emph{almost-sure termination} \cite{chatterjee_TerminationNondeterministicRecursive_2017,mciver_NewProofRule_2017,fu_TerminationNondeterministicProbabilistic_2019} is a standard criterion, which requires that a given probabilistic program terminates with probability $1$.
	In this example, if we ignore divergence, one can argue that $P_1$, $P_2$ and $S$ are pairwise branching (also weak) bisimilar to each other.
	However, only $P_1$ and $S$ are almost-surely terminating, whereas $P_2$ is not almost-surely terminating (as $P_2$ can reach a state $Q_2$ that can loop forever).
	Thus from the point of view of almost-sure termination, $P_1$ and $P_2$ are not equivalent, and it is reasonable to say that only $P_1$ implements $S$ faithfully.
	Since $P_2$ can reach a silent cycle whereas $S$ and $P_1$ cannot, the \emph{exhaustive weak probabilistic bisimilarity} proposed by He et al. \cite{he_DivergencesensitiveWeakProbabilistic_2023}
	distinguishes $P_2$ from $P_1$ (and $S$ as well).
	
	Let us take an even closer look at $P_2$, and consider the pair of states $(P_2,Q_2)$.
	Neither $P_2$ nor $Q_2$ is almost-surely terminating, and both can reach the cycle of $Q_2$. So they cannot be separated by the exhaustive weak probabilistic bisimilarity of He et al. \cite{he_DivergencesensitiveWeakProbabilistic_2023}.
	However, their behaviors might appear very different to environments, and from the perspective of observation they ought to be distinguished.
	Consider the two nondeterministic transitions from $Q_2$, one has $\mathsf{tr}_1 = Q_2 \xrightarrow{\tau} Q_2$ and $\mathsf{tr}_2 = Q_2 \xrightarrow{\tau} P_2$.
	By our understanding of nondeterminism, there is the possibility that $\mathsf{tr}_1$ is repeatedly executed ad infinitum, due to hardware malfunction for instance.
	An external observer $O$ can tell $P_2$ and $Q_2$ apart by interacting with them.
	There is a non-zero {\em probability} that $O$ communicates with $P_2$ through channel $a$ or $b$.
	On the other hand there is a {\em possibility} that $O$ may never communicate with $Q_2$.
	The distinction between probability and possibility must be maintained in probabilistic nondeterministic models.
	The subtle difference between $P_2$ and $Q_2$ cannot be detected by the $\tau$-EC approach.
	It can be recognized by the \emph{divergence-sensitive branching bisimilarity} of  Fu \cite{fu_ModelIndependentApproach_2021}.
	
	\begin{figure}[htb]
		\vspace{-1.15mm}
		\captionsetup{justification=centerlast}
		\hspace{-10mm}
		\begin{subfigure}[t]{0.25\textwidth}
			\vspace{0pt}
			\centering
			\begin{tikzpicture}[on grid,scale=0.8]
				\begin{scope}[ every node/.style={draw,circle,transform shape,inner sep=3pt}]
					\foreach \pos/\name/\text in {
						{(0,1.5)/a/{$a$}},
						{(1,3)/S/{$S$}},
						{(2,1.5)/b/{$b$}},
						{(1,0.75)/z/{$\mathbf{0}$}}}
					\node (\name) at \pos {\text};
				\end{scope}
				
				\begin{scope}[every node/.style={transform shape, auto=left,inner sep=1pt}, >={Latex[width'=0pt .5, length=5pt]}]
					\foreach \source/ \dest / \text in {
						b/z/{$b$}, S/b/{$\frac{1}{2}\tau$}}
					\path[->] (\source) edge node {\text} (\dest);
				\end{scope}	
				
				\begin{scope}[every node/.style={transform shape, auto=right,inner sep=1pt}, >={Latex[width'=0pt .5, length=5pt]}]
					\foreach \source/ \dest / \text in {
						a/z/{$a$}, S/a/{$\frac{1}{2}\tau$}}
					\path[->] (\source) edge node {\text} (\dest);
					
				\end{scope}	
			\end{tikzpicture}
			\caption{Specification.}\label{fig-motivation-a}
		\end{subfigure}
	    \hspace{5pt}
		\begin{subfigure}[t]{0.3\textwidth}
			\vspace{0pt}
			\centering
			\begin{tikzpicture}[on grid,scale=0.8]
				\begin{scope}[every node/.style={draw,circle,transform shape,inner sep=3pt}]
					\foreach \pos/\name/\text in {
						{(1,3)/P1/{$P_1$}},
						{(1,1.5)/a/{$a$}},
						{(2+0.5,1.5)/b/{$b$}},
						{(1.75,0.75)/z/{$\mathbf{0}$}}}
					\node (\name) at \pos {\text};
				\end{scope}
				
				\begin{scope}[ every node/.style={transform shape, auto=left,inner sep=1pt}, >={Latex[width'=0pt .5, length=5pt]}]
					\foreach \source/ \dest / \text in {
						b/z/{$b$}, P1/a/{$\frac{1}{3}\tau$}, P1/b/{$\frac{1}{3}\tau$}}
					\path[->] (\source) edge node {\text} (\dest);
				\end{scope}	
				
				\begin{scope}[ every node/.style={transform shape, auto=right,inner sep=1pt}, >={Latex[width'=0pt .5, length=5pt]}]
					\foreach \source/ \dest / \text in {
						a/z/{$a$}}
					\path[->] (\source) edge node {\text} (\dest);
				\end{scope}	
				
				\begin{scope}[every node/.style={transform shape, auto=right,inner sep=1pt}, >={Latex[width'=0pt .5, length=5pt]}]
					\path[->] (P1) edge [out=-115, in=-155, looseness=20] node {$\frac{1}{3}\tau$} (P1);
				\end{scope}	
			\end{tikzpicture}
			\vspace{-2pt}
			\caption{The first implementation.}\label{fig-motivation-b}
		\end{subfigure}
		\hfil
		\begin{subfigure}[t]{0.35\textwidth}
			\vspace{0pt}
			\centering
			\begin{tikzpicture}[on grid,scale=0.8]
				\begin{scope}[every node/.style={draw,circle,transform shape,inner sep=3pt}]
					\foreach \pos/\name/\text in {
						{(0-0.5,1.5)/B/{$Q_2$}},
						{(1,3)/A/{$P_2$}},
						{(1,1.5)/a/{$a$}},
						{(2+0.5,1.5)/b/{$b$}},
						{(1.75,0.75)/z/{$\mathbf{0}$}}}
					\node (\name) at \pos {\text};
				\end{scope}
				
				\begin{scope}[every node/.style={transform shape, auto=left,inner sep=1pt}, >={Latex[width'=0pt .5, length=5pt]}]
					\foreach \source/ \dest / \text in {
						b/z/{$b$}, A/a/{$\frac{1}{3}\tau$}, A/b/{$\frac{1}{3}\tau$}}
					\path[->] (\source) edge node {\text} (\dest);
					
					\path[->] (B) edge [out=-135, in=-45, looseness=6] node {$\tau$} (B);
					\path[->] (B) edge [bend left=45] node {$\tau$} (A);
				\end{scope}	
				
				\begin{scope}[every node/.style={transform shape, auto=right,inner sep=1pt}, >={Latex[width'=0pt .5, length=5pt]}]
					\foreach \source/ \dest / \text in {
						a/z/{$a$}, A/B/{$\frac{1}{3}\tau$}}
					\path[->] (\source) edge node {\text} (\dest);
					
				\end{scope}	
			\end{tikzpicture}
			\vspace{-8pt}
			\caption{The second implementation.}\label{fig-motivation-c}
		\end{subfigure}
		\caption{Examples of systems with different divergence behaviors.}\label{fig-motivation}
	\end{figure}
	
	\subsubsection*{Related Work}
	In \cite{milner_CommunicationConcurrency_1989}, \emph{strong bisimulation} and \emph{weak bisimulation} are introduced for the $\mathrm{CCS}$ model. These two bisimulation semantics differ in the way to treat internal computations: the former requires that for each pair of bisimilar processes, every action immediately enabled by one process must be matched by the same action immediately enabled by the other process; the latter allows the matching steps using additional internal actions.
	Then in \cite{vanglabbeek_BranchingTimeAbstraction_1996}, van Glabbeek and Weijland propose a refined alternative to weak bisimulation, namely \emph{branching bisimulation}. Branching bisimulation is finer than weak bisimulation as it requires that the additional internal actions used in the matching steps need to be \emph{state-preserving} (i.e., the intermediate states passed in the steps all belong to the same equivalence class).
	Traditional weak and branching bisimulation ignore the role of divergence, some of their divergence-sensitive invariants then are considered in \cite{vanglabbeek_BranchingBisimilarityExplicit_2009, liu_AnalyzingDivergenceBisimulation_2017}.
	In \cite{vanglabbeek_BranchingBisimilarityExplicit_2009}, van Glabbeek et al.\ propose the notion of \emph{branching bisimilarity with explicit divergence} and prove that it is an equivalence.
	Liu et al. \cite{liu_AnalyzingDivergenceBisimulation_2017} show that it is much more difficult  to prove the equivalence property of the \emph{weak bisimilarity with explicit divergence}.
	Instead of giving a direct proof, they get around the difficulty by constructing a new equivalence called \emph{complete weak bisimilarity} and showing that it is the largest weak bisimulation with explicit divergence.
	Recently, the notion of \emph{rooted divergence-preserving branching bisimilarity} has been proposed in \cite{sun_RootedDivergencePreservingBranching_2023} and has been proved to be a congruence for $\mathrm{CCS}$ with guarded recursion.

	When generalized to probabilistic process model, several probabilistic bisimilarities has been proposed and investigated, including strong probabilistic bisimilarity \cite{segala_ModelingVerificationRandomized_1995}, weak probabilistic bisimilarity \cite{segala_ModelingVerificationRandomized_1995, turrini_PolynomialTimeDecision_2015} and branching probabilistic bisimilarity \cite{fu_ModelIndependentApproach_2021, castiglioni_RaidersLostEquivalence_2020}.
	Two representative works are the \emph{distribution-based weak probabilistic bisimilarity} \cite{segala_ModelingVerificationRandomized_1995} and \emph{$\epsilon$-tree based branching probabilistic bisimilarity} \cite{fu_ModelIndependentApproach_2021}.
	Weak probabilistic bisimilarity has been introduced in \cite{segala_ModelingVerificationRandomized_1995} for the PA model and has been investigated  extensively over the past 30 years.
	The branching probabilistic bisimilarity proposed by Fu \cite{fu_ModelIndependentApproach_2021}  is a conservative generalization of the classical branching bisimilarity \cite{vanglabbeek_BranchingTimeAbstraction_1996} and has been shown to be a congruence for the RCCS model.
	It is  well-known that branching bisimilarities are strictly finer than weak bisimilarities~\cite{vanglabbeek_BranchingTimeAbstraction_1996}.
	However, exploring such relationship in probabilistic setting turns out to be a challenge. 
 The technical reason is that Fu \cite{fu_ModelIndependentApproach_2021} takes a tree-based characterization for probabilistic transitions whereas Segala
	\cite{segala_ModelingVerificationRandomized_1995} takes the distribution-based characterization. 
	Divergence issue in probabilistic setting has been considered in \cite{fu_ModelIndependentApproach_2021} and \cite{he_DivergencesensitiveWeakProbabilistic_2023}.
	In \cite{fu_ModelIndependentApproach_2021}, based on the notion of \emph{divergent $\epsilon$-tree}, Fu introduces a divergence-sensitive refinement of branching bisimilarity, which can be seen as a probabilistic generalization of branching bisimilarity with explicit divergence in classical CCS model.
	In \cite{he_DivergencesensitiveWeakProbabilistic_2023}, He et al.\ propose \emph{exhaustive weak probabilistic bisimilarity} as a divergence-sensitive refinement of weak probabilistic bisimilarity \cite{segala_ModelingVerificationRandomized_1995}, where the divergence property is based on the notion of \emph{$\tau$-EC}.
	The exhaustive weak probabilistic bisimilarity is actually a probabilistic version of the complete weak bisimilarity of Liu et al. \cite{liu_AnalyzingDivergenceBisimulation_2017}.	

	Quite a few equivalence checking algorithms for the above mentioned equivalences have been studied in the literature.
	Kanellakis and Smolka  \cite{kanellakis_CCSExpressionsFinite_1990} propose polynomial-time decision algorithms for strong bisimilarity and weak bisimilarity in CCS model. A key technique used in their algorithms is the \emph{partition-refinement approach} \cite{paige_ThreePartitionRefinement_1987}: given a set $S$ of processes, they start with the coarsest partition of $S$ and then keep refining it until the resulting partition satisfies the requirement of strong (or weak) bisimulation.
	An efficient decision algorithm for branching bisimilarity is given in \cite{groote_EfficientAlgorithmBranching_1990}.
	Later in \cite{baier_DecidingBisimilaritySimilarity_2000a}, Baier et al.\ generalizing the partition-refinement approach to probabilistic models and present efficient decision algorithm for strong probabilistic bisimilarity introduced in \cite{segala_ModelingVerificationRandomized_1995}. 
	Recently, Jacobs and Wi{\ss}mann \cite{jacobs_FastCoalgebraicBisimilarity_2023} present a generic algorithm for deciding a class of behavioural equivalences whose underlying transition structure is specified by a functor in the category of sets, subsuming strong bisimilarity \cite{milner_CommunicationConcurrency_1989} and strong probabilistic bisimilarity \cite{segala_ModelingVerificationRandomized_1995}.
	Turrini and Hermanns \cite{turrini_PolynomialTimeDecision_2015} give a delicate polynomial time algorithm for deciding weak probabilistic bisimilarity \cite{segala_ModelingVerificationRandomized_1995} for PA model, significantly improving the previous exponential complexity in \cite{cattani_DecisionAlgorithmsProbabilistic_2002}.  The key technique  in \cite{turrini_PolynomialTimeDecision_2015} is  a novel
	characterization of the weak combined transitions as a linear programming problem.
	Zhang et al. \cite{zhang_UniformRandomProcess_2019} introduce a novel notion of \emph{$\epsilon$-graph} and use it to give a polynomial algorithm for checking branching probabilistic bisimilarity proposed in \cite{fu_ModelIndependentApproach_2021}.
	Neither Turrini and Hermanns \cite{turrini_PolynomialTimeDecision_2015} nor Zhang et al. \cite{zhang_UniformRandomProcess_2019} give consideration to the divergence issue.
	To the best of our knowledge, algorithmic treatments to divergence-sensitive bisimilarity for probabilistic models appears in \cite{he_DivergencesensitiveWeakProbabilistic_2023} for the first time.
	By combining the classical partition-refinement framework with the inductive verification approach proposed in~\cite{liu_AnalyzingDivergenceBisimulation_2017}, He et al. \cite{he_DivergencesensitiveWeakProbabilistic_2023} present a polynomial verification algorithm for exhaustive weak probabilistic bisimilarity.
	
	
	The picture of the divergence-sensitive probabilistic bisimulation equivalences is far from complete.
	In this paper we focus on the divergence issue in this picture.
	We shall prove a number of separation results regarding the equivalence relations mentioned above, and carry out algorithmic studies on these equivalences. 
	
	\subsubsection*{Contribution}
	The main contributions of this paper are stated as follows.
	\begin{enumerate}
		\item \label{2023-07-10} We give a comprehensive comparison between variants of (divergence-sensitive) branching and weak bisimulation semantics for probabilistic processes (Theorem \ref{thm-spectrum}). Particularly, we show that the $\epsilon$-tree based branching bisimilarity is finer than the distribution-based weak bisimilarity (Theorem \ref{thm-BbInWb}).
		We also show that the divergent $\epsilon$-tree property is stronger than the $\tau$-EC property in the branching semantics (Theorem \ref{thm-dsbbVSebb}).
		
		\item 
		We give efficient verification algorithms for these divergence-sensitive bisimilarities. Particularly, for the exhaustive weak bisimilarity ($\approx_e$), rather than using the inductive verification method, we propose a new polynomial-time verification algorithm by making use of the so-called maximal $\tau$-EC.
		
		\item We also present some novel techniques that could be of independent interest.
		In establishing Theorem \ref{thm-BbInWb}, we come up with a way to relate distribution-based semantics and $\epsilon$-tree based semantics for probabilistic models.
		When proving Theorem \ref{thm-dsbbVSebb}, we apply a technical lemma (Lemma \ref{lem-EcAndTree}) that builds the connection between $\tau$-EC and divergent $\epsilon$-tree.
	\end{enumerate}
	
	\subsubsection*{Organization}
	The paper is organized as follows. Section \ref{sec-preliminary} summarizes the necessary knowledge about the finite state probabilistic model and two notions of bisimulations. The relationship between the branching and weak bisimilarities for such model is studied.
	Section~\ref{sec-branching} defines two divergence-sensitive branching bisimulation semantics, the branching bisimilarity with explicit divergence and the exhaustive branching bisimilarity, along with the discussion of their relationship.
	Section~\ref{sec-spectrum} builds up a lattice for the variants of the probabilistic bisimilarities. Section~\ref{sec-algorithm} gives the equivalence checking algorithms for the divergence-sensitive bisimilarities studied in the paper, all with polynomial time complexity.
	Section~\ref{sec-conclusion} concludes.

	\section{Preliminaries}\label{sec-preliminary}
	We begin by fixing the probabilistic process model of this paper.
	We then introduce the branching and weak bisimilarities without any consideration of divergence.
	The technical contribution of this section is a proof of the fact that the branching bisimilarity indeed implies the weak bisimilarity in the randomized $\mathrm{CCS}$ model.
	This is not a routine exercise since it calls for a comparison of the $\epsilon$-tree based semantics against the distribution-based semantics.
	
	\subsection{Background knowledge}
	\subsubsection{Finite state randomized \textrm{CCS} model}
	Let $\mathcal{A}$ be the set of external actions, ranged over by lowercase letters $a,b,c$. We use a special symbol $\tau \notin \mathcal{A}$ to represent the internal action. The set of actions is $Act = \mathcal{A} \cup \{\tau\}$, ranged over by $\alpha, \beta, \gamma,\ell$.
	Let $Act_p$ be the set $Act \cup \{p\tau \mid 0 <p <1\}$, ranged over by $\lambda$.
	The grammar of finite state randomized CCS model, $\mathrm{RCCS}_{fs}$, is defined as:
	\begin{equation}\label{eq-grammer}
		T :=  \mathbf{0} ~\Big{|}~ X ~\Big{|}~ \sum_{i \in I} \alpha_i.T_i ~\Big{|}~ \bigoplus_{i \in I}p_i\tau.T_i ~\Big{|}~ \mu X.T,\tag{$*$}
	\end{equation}
	where the non-empty index set $I$ is finite.
	In (\ref{eq-grammer}), $\mathbf{0}$ is the nil term,
	$X$ is a process variable, $\sum_{i \in I} \alpha_i.T_i$ is a nondeterministic choice term, $\bigoplus_{i \in I}p_i\tau.T_i$ is a probabilistic choice term, and $\mu X.T$ is a fixpoint term.
	A trailing $\mathbf{0}$ which appears at the end of a term is often omitted, e.g., $\tau.a$ represents $\tau.a.\mathbf{0}$.
	Sometimes we will use the infix notation of $\sum$ to specify particular summands in the nondeterministic choice term, writing for example $\sum_{i \in I'} \alpha_i. T_i + \beta.T' + \gamma.T''$.
	In the probabilistic choice term $\bigoplus_{i \in I}p_i\tau.T_i$, $I$ is a finite set with $|I| \ge 2$, each $p_i \in (0,1)$ and $\sum_{i \in I} p_i = 1$.
	A process variable $X$ that appears in $\sum_{i \in I} \alpha_i.T_i$ (or $\bigoplus_{i \in I}p_i\tau.T_i$) is \emph{guarded}. We shall assume that in the fixpoint term $\mu X.T$ the \emph{bound} variable $X$ is guarded in $T$. A variable in a term is \emph{free} if it is not bound. A term is a \emph{process} if it contains no free variables.
	We write $A,B,C,P,Q$ for processes. The set of all $\mathrm{RCCS}_{fs}$ processes is denoted by $\PRCCS$.
	The operational semantics of $\mathrm{RCCS}_{fs}$ is given by the \emph{labeled transition system} (LTS for short) in Figure \ref{fig-rccs}, where $\lambda \in Act_p$ and the transition relation $\longrightarrow\;\subseteq \PRCCS\times Act_p \times \PRCCS$.
	
	\begin{figure}[htb]	
		\vspace{-5mm}
		\begin{center}
			\[	\frac{}{~\sum_{i \in I} \alpha_i.T_i \xrightarrow{\alpha_i} T_i~} \qquad
			\frac{}{~\bigoplus_{i \in I}p_i\tau.T_i \xrightarrow{p_i\tau} T_i~} \qquad
			\frac{T\{\mu X.T/X\} \xrightarrow{\lambda}T'}{\mu X.T \xrightarrow{\lambda} T'}
			\]
		\end{center}
		\vspace{-1mm}
		\caption{LTS for $\mathrm{RCCS}_{fs}$.}\label{fig-rccs}
		\vspace{-2mm}
	\end{figure}
	
	For any $A\in \PRCCS$, there could be only a finite number of processes reachable from $A$.
	The \emph{induced transition graph of $A$}, denoted by $G_A = (V_A, E_A)$, is a directed labeled graph satisfying that $V_A$ contains all the processes reachable from $A$, $E_A$ contains all the transitions on $V_A$ and each edge $e_A = (A', A'') \in E_A$ with label $\lambda \in Act_p$ stands for the transition $A' \xrightarrow{\lambda} A''$ in LTS.
	
	\begin{example}
		The three probabilistic systems in Figure \ref{fig-motivation} can be defined as the following $\mathrm{RCCS}_{fs}$ processes: $S = \frac{1}{2}\tau. a \oplus \frac{1}{2}\tau.b$, $P_1 = \mu X. (\frac{1}{3}\tau.X \oplus \frac{1}{3}\tau.a \oplus  \frac{1}{3}\tau.b)$, $Q_2 = \mu X.(\tau.X + \tau.(\frac{1}{3}\tau.X \oplus\frac{1}{3}\tau.a \oplus \frac{1}{3}\tau.b))$ and $P_2 =\frac{1}{3}\tau.Q_2 \oplus \frac{1}{3}\tau.a \oplus \frac{1}{3}\tau.b$.
		Figures \ref{fig-motivation-a}, \ref{fig-motivation-b} and \ref{fig-motivation-c} then give the induced transition graph for the $\mathrm{RCCS}_{fs}$ process $S$, $P_1$ and $P_2$, respectively.
	\end{example}
	
	Following~\cite{fu_ModelIndependentApproach_2021}, a collection of probabilistic transitions $\left\{\bigoplus_{i \in I}p_i\tau.T_i \xrightarrow{p_i \tau}T_i\right\}_{i\in I}$ can be treated as a \emph{collective silent transition}, in notation $\bigoplus_{i \in I}p_i\tau.T_i \xrightarrow{\coprod_{i \in I}p_i \tau} \coprod_{i \in I} T_i$, where the  auxiliary notation $\coprod$ is used to indicate a collection of things.
	We extend the notation $\xrightarrow{\coprod_{i \in I}p_i \tau}$ to fixpoint terms as follows: if $T\{\mu X.T / X\} \xrightarrow{\coprod_{i \in I}p_i \tau} \coprod_{i \in I} T_i$, then we define $\mu X.T \xrightarrow{\coprod_{i \in I}p_i \tau} \coprod_{i \in I} T_i$. 
	We give an example as follows, where the notation $[k]$ stands for the set $\{1, \cdots, k\}$.
	\begin{example}
		Let $T = \frac{1}{3}\tau. X \oplus \frac{2}{3}\tau. \mathbf{0}$ and consider the fixpoint process $P = \mu X.T = \mu X.(\frac{1}{3}\tau. X \oplus \frac{2}{3}\tau. \mathbf{0})$. Let $p_1 = \frac{1}{3}$, $p_2 = \frac{2}{3}$, $T_1 = P$ and $T_2 = \mathbf{0}$, then $T\{P / X\} = \bigoplus_{i \in [2]}p_i\tau.T_i $.
		Since $T\{P / X\}$ can perform the collective silent transition $T\{P / X\} \xrightarrow{\coprod_{i \in [2]}p_i \tau} \coprod_{i \in [2]} T_i$, one has $P \xrightarrow{\coprod_{i \in [2]}p_i \tau} \coprod_{i \in [2]} T_i$.
	\end{example}

	An \emph{immediate silent transition of $A$}, denoted by $\mathsf{itr}_A$, is either a non-probabilistic silent transition $A \xrightarrow{\tau} A'$  or a collective silent transition $A \xrightarrow{\coprod_{j \in J}q_j \tau} \coprod_{j \in J} A_j$ (where $\sum_{j \in J} q_j = 1$).
	We use $tgt(\mathsf{itr}_A)$ to denote the \emph{target set} of $\mathsf{itr}_A$, which is defined as $tgt(A \xrightarrow{\tau} A') = \{A'\}$ and $tgt(A \xrightarrow{\coprod_{j \in J}q_j \tau} \coprod_{j \in J} A_j) = \{A_j~|~{j\in J}\}$.
	
	We will use $\mathcal{E}$ to denote an equivalence and $\mathcal{R}$ to denote a binary relation.
	We write $A ~\mathcal{E}~ B$ for $(A,B) \in \mathcal{E}$ and use $[A]_{\mathcal{E}}$ to denote the equivalence class containing $A$.
	For an equivalence $\mathcal{E}$ on $\PRCCS$,
	the notation $\PRCCS / \mathcal{E}$ stands for the set of equivalence classes defined by $\mathcal{E}$.
	Given an equivalence $\mathcal{E}$ on $\PRCCS$, we say that an immediate silent transition $\mathsf{itr} = A \xrightarrow{\tau} A'$ is \emph{state-preserving} if
	$A' ~\mathcal{E}~ A$ and $\mathsf{itr} = A \xrightarrow{\coprod_{j \in J}q_j \tau} \coprod_{j \in J} A_j$ is state-preserving if $A_j ~\mathcal{E}~ A$ for all $j \in J$.
	An immediate silent transition $\mathsf{itr}$ is called \emph{state-changing} if it is not state-preserving.
	
	\subsubsection{Branching bisimilarity}\label{subsec-Bb}
	
	Branching bisimilarity for $\mathrm{RCCS}_{fs}$ model was proposed by Fu \cite{fu_ModelIndependentApproach_2021}. It is a behavioral equivalence compatible with the classical branching bisimilarity \cite{vanglabbeek_BranchingTimeAbstraction_1996}.
	We start with the definition of  \emph{$\epsilon$-tree} \cite{fu_ModelIndependentApproach_2021}.
	Intuitively, $\epsilon$-tree is a probabilistic version of $\Longrightarrow_{\mathcal{E}}$ (a sequence of state-preserving internal actions with regard to the equivalence $\mathcal{E}$ in non-probabilistic setting).

	\begin{definition}[$\epsilon$-tree \cite{fu_ModelIndependentApproach_2021}]\label{defn-epsilontree}
		Let $\mathcal{E}$ be an equivalence on $\PRCCS$ and $A \in \PRCCS$ be a process. An \emph{$\epsilon$-tree $t_{\mathcal{E}}^{A}$ of $A$ with regard to $\mathcal{E}$} is a labeled tree such that the following statements are valid.
		\begin{itemize}
			\item Each node of $t_{\mathcal{E}}^{A}$ is labeled by an element of $[A]_{\mathcal{E}}$, and each edge is labeled by an element of $(0,1]$. The root of $t_{\mathcal{E}}^{A}$ is labeled by $A$.
			\item If a node labeled $B$ has only one child $B'$, then $B \xrightarrow{\tau} B'$ and the edge from $B$ to $B'$ is labeled $1$.
			\item If a node labeled $B$ has $k$ children $B_1, \cdots, B_k$ and each edge from $B$ to $B_i$ is labeled $p_i$, then $B \xrightarrow{\coprod_{i \in [k]}p_i \tau} \coprod_{i \in [k]} B_i$.
		\end{itemize}
	\end{definition}

	An $\epsilon$-tree $t_{\mathcal{E}}^{A}$ of $A$ with regard to $\mathcal{E}$ is called \emph{maximal} if there is no other $\epsilon$-tree $(t')_{\mathcal{E}}^{A}$ such that $t_{\mathcal{E}}^{A}$ is a proper subtree of $(t')_{\mathcal{E}}^{A}$.
	For a tree $t$, a \emph{branch}  is either a finite path from its root to a leaf or an infinite path.
	For a finite path $\pi$, we will use $\pi(i)$ to denote the label of the $i$-th edge in $\pi$ and use $|\pi|$ to denote the length of $\pi$. The probability $\mathsf{P} (\pi)$ of a finite path $\pi$ is $\prod_{i\leq |\pi|} \pi(i)$.
	The \emph{convergence probability} of $t_{\mathcal{E}}^{A}$ is then defined by $\mathsf{P}^c(t_{\mathcal{E}}^{A}) = \lim_{k\to \infty}\mathsf{P}_{k}(t_{\mathcal{E}}^{A})$, where
	$$
	\mathsf{P}_{k}(t_{\mathcal{E}}^{A}) ~\defn~ \sum\{\mathsf{P}
	(\pi) \mid \pi \text{ is a finite branch in $t_{\mathcal{E}}^{A}$ such that $|\pi| \le k$}\}.
	$$
	
	\begin{definition}[Regular and divergent $\epsilon$-tree \cite{fu_ModelIndependentApproach_2021}]\label{defn-regdiv}
		An $\epsilon$-tree $t_{\mathcal{E}}^{A}$ is \emph{regular} if $\mathsf{P}^c(t_{\mathcal{E}}^{A})=1$;
		it is \emph{divergent} if $\mathsf{P}^c(t_{\mathcal{E}}^{A})=0$.
	\end{definition}

	\begin{example}\label{example-etree}
		Let $P_2 = \mu X.(a. \mu Y.(\frac{1}{5}\tau. X \oplus \frac{3}{5}\tau.Y \oplus \frac{1}{5}\tau.\mathbf{0}))$, 
		$P_4 = \mu Y. (\frac{1}{5}\tau. P_2 \oplus \frac{3}{5}\tau.Y \oplus \frac{1}{5}\tau.\mathbf{0})$,
		$P_1 = \mu Z. (\frac{1}{3}\tau.P_2 \oplus \frac{2}{3}\tau. a. (\frac{1}{2}\tau. Z \oplus \frac{1}{2}\tau. \mathbf{0} ))$,
		$P_3 = a. (\frac{1}{2}\tau. P_1 \oplus \frac{1}{2}\tau. \mathbf{0} )$,
		$P_3 = \frac{1}{2}\tau. P_1 \oplus \frac{1}{2}\tau. \mathbf{0}$,
		and $P = \mu W.(\tau. W + \tau.P_1)$.
		The induced transition graph of process $P$ is depicted in Figure~\ref{fig-example-a}.
		Now consider the equivalence $\mathcal{E} = \{\{P, P_1, P_2, P_3\}, \{P_4,P_5\}, \{\mathbf{0}\}\}$. Figure \ref{fig-example-b} and \ref{fig-example-c} then give regular and divergent $\epsilon$-trees of $P$ with regard to $\mathcal{E}$, respectively.
	\end{example}

	\begin{figure}[htb]
		\begin{subfigure}[t]{0.38\textwidth}
			\centering
			\begin{tikzpicture}[on grid,scale=0.8]
				\begin{scope}[every node/.style={transform shape,inner sep=1.5pt}]
					\foreach \pos/\name/\text in {
						{(1,7)/P/{$P$}},
						{(1,5.5)/P1/{$P_1$}},
						{(0,4.5)/P2/{$P_2$}},
						{(2,4.5)/P3/{$P_3$}},
						{(0,3)/P4/{$P_4$}},
						{(2,3)/P5/{$P_5$}},
						{(0,1.5)/Q/{$\mathbf{0}$}}}
					\node (\name) at \pos {\text};
				\end{scope}
				
				\begin{scope}[every node/.style={transform shape, auto=left,inner sep=1pt}, >={Latex[width'=0pt .5, length=5pt]}]
					\foreach \source/ \dest / \text in {
						P1/P3/{\footnotesize$\frac{2}{3}\tau$},
						P3/P5/{\footnotesize$a$}
					},
					\path[->] (\source) edge node {\text} (\dest);
					
					\path[->] (P4) edge [bend left=45] node {\footnotesize$\frac{1}{5}\tau$} (P2);
					\path[->] (P5) edge [bend left=30] node {\footnotesize$\frac{1}{2}\tau$} (Q);
					
					\path[->] (P) edge [out=135, in=45, looseness=6] node {\footnotesize$\tau$} (P);
				\end{scope}	
				
				\begin{scope}[every node/.style={transform shape, auto=right,inner sep=1pt}, >={Latex[width'=0pt .5, length=5pt]}]
					\foreach \source/ \dest / \text in {
						P/P1/{\footnotesize$\tau$},
						P1/P2/{\footnotesize$\frac{1}{3}\tau$},
						P2/P4/{\footnotesize$a$},
						P4/Q/{\footnotesize$\frac{1}{5}\tau$}
					},
					\path[->] (\source) edge node {\text} (\dest);
					
					\path[->] (P5) edge [bend right=75] node {\footnotesize$\frac{1}{2}\tau$} (P1);			
					\path[->] (P4) edge [out=150, in=-150, looseness=7] node {\footnotesize$\frac{3}{5}\tau$} (P4);
				\end{scope}			
			\end{tikzpicture}
			\caption{The induced transition graph of $P$.}\label{fig-example-a}
		\end{subfigure}
		\hspace{1mm}
		\begin{subfigure}[t]{0.25\textwidth}
			\centering
			\vspace{-46mm}
			\begin{tikzpicture}[on grid,scale=0.8]
				\begin{scope}[ every node/.style={transform shape,inner sep=1.5pt}]
					\foreach \pos/\name/\text in {
						{(1,7)/P/{$P$}},
						{(1,5.5)/P1/{$P_1$}},
						{(0,4.5)/P2/{$P_2$}},
						{(2,4.5)/P3/{$P_3$}}}
					\node (\name) at \pos {\text};
				\end{scope}
				
				\begin{scope}[every node/.style={transform shape, auto=left,inner sep=1pt}, >={Latex[width'=0pt .5, length=5pt]}]
					\foreach \source/ \dest / \text in {
						P1/P3/{\footnotesize$\frac{2}{3}$}
					},
					\path[->] (\source) edge node {\text} (\dest);
				\end{scope}	
				
				\begin{scope}[every node/.style={transform shape, auto=right,inner sep=1pt}, >={Latex[width'=0pt .5, length=5pt]}]
					\foreach \source/ \dest / \text in {
						P1/P2/{\footnotesize$\frac{1}{3}$},
						P/P1/{\footnotesize$1$}
					},
					\path[->] (\source) edge node {\text} (\dest);
				\end{scope}	
			\end{tikzpicture}
			\vspace{23mm}
			\caption{A regular $\epsilon$-tree of $P$.}\label{fig-example-b}
		\end{subfigure}
		\begin{subfigure}[t]{0.3\textwidth}  
			\centering
			\vspace{-45.5mm}
			\begin{tikzpicture}[on grid,scale=0.8]
				\begin{scope}[every node/.style={transform shape,inner sep=1.5pt}]
					\foreach \pos/\name/\text in {
						{(1,5.5)/P/{$P$}},
						{(1,4)/P'/{$P$}},
						{(1,2.5)/P''/{$P$}}}
					\node (\name) at \pos {\text};
				\end{scope}
				
				\begin{scope}[every node/.style={transform shape, auto=right,inner sep=1pt}, >={Latex[width'=0pt .5, length=5pt]}]
					\foreach \source/ \dest / \text in {
						P/P'/{\footnotesize$1$},
						P'/P''/{\footnotesize$1$}
					},
					\path[->] (\source) edge node {\text} (\dest);
				\end{scope}	
				
				\begin{scope}[inner sep=0pt,dot/.style={fill,circle,minimum size=1pt}]
					\foreach \pos in {0.3, 0.5, 0.7}
					\node [dot] [below = (\pos) of P''] {};						
				\end{scope}					
			\end{tikzpicture}
			\vspace{13mm}
			\caption{A divergent $\epsilon$-tree of $P$.}\label{fig-example-c}
		\end{subfigure}
		\caption{Example of $\epsilon$-trees.}\label{fig-example}	
	\end{figure}
	
	We then give the definition of $\ell$-transition.
	Intuitively, $\ell$-transition can be seen as a probabilistic generalization of the transition $\Longrightarrow_{\mathcal{E}}\xrightarrow{\ell}$ in classical CCS model, where the state-preserving internal actions sequence $\Longrightarrow_{\mathcal{E}}$ is now replaced by a regular $\epsilon$-tree.
	
	\begin{definition}[$\ell$-transition \cite{fu_ModelIndependentApproach_2021}] \label{defn-LT}
		Suppose $\mathcal{B} \in \PRCCS/\mathcal{E}$ and $(\ell \in \mathcal{A}) \vee \left(\ell = \tau \wedge \mathcal{B} \ne [A]_{\mathcal{E}} \right)$. 
		We say that there is an $\ell$-transition from $A$ to $\mathcal{B}$ with regard to $\mathcal{E}$, written $A \rightsquigarrow_{\mathcal{E}}\xrightarrow{\ell} \mathcal{B}$, if there exists a regular $\epsilon$-tree $t_{\mathcal{E}}^{A}$ satisfying that for every leaf $L$ of $t_{\mathcal{E}}^{A}$, there exists a transition $L \xrightarrow{\ell} L'$ such that $L' \in \mathcal{B}$.
	\end{definition}


	\begin{example}
		Consider the process $P$ and equivalence $\mathcal{E}$ in Example \ref{example-etree}. Since $P_2 \xrightarrow{a} P_4$, $P_3  \xrightarrow{a} P_5$ and $[P_4]_{\mathcal{E}} = [P_5]_{\mathcal{E}} \ne [P]_{\mathcal{E}}$, the regular $\epsilon$-tree in Figure \ref{fig-example-b} then induces the $\ell$-transition $P \rightsquigarrow_{\mathcal{E}}\xrightarrow{a} [P_4]_{\mathcal{E}}$.
	\end{example}
	
	State-changing probabilistic silent actions are characterized by $q$-transitions in \cite{fu_ModelIndependentApproach_2021}.
	Intuitively $q$-transitions capture the idea that after some state-preserving internal actions, every derived process can evolve into some new equivalence
	class with the same conditional probability $q$.
	

	Given a collective silent transition $L \xrightarrow{\coprod_{i \in [k]} \, p_i\tau} \coprod_{i \in [k]} \, L_i$ and an equivalence class $\mathcal{B} \in \PRCCS/\mathcal{E}$, the probability of $L$ arrives at $\mathcal{B}$ is defined by $\mathsf{P}  ( L \xrightarrow{\coprod_{i \in [k]} \, p_i\tau} \mathcal{B} ) =  \sum_{\substack{i \in [k], L_i \in \mathcal{B}}} \,p_i$. 

	Suppose $L \xrightarrow{\coprod_{i \in [k]} \, p_i\tau} \coprod_{i \in [k]} \, L_i$ and $\mathsf{P}  ( L \xrightarrow{\coprod_{i \in [k]} \, p_i\tau} [L]_{\mathcal{E}} ) < 1$, the \emph{normalized probability} is defined as the conditional probability of $L$ arrives at $\mathcal{B}$ given that $L$ leaves $[L]_\mathcal{E}$, i.e., 
	\begin{flalign*}
		\mathsf{P}_{\mathcal{E}} ( L \xrightarrow{\coprod_{i \in [k]} \, p_i\tau} \mathcal{B}  )
		~\defn~  
		\mathsf{P}  ( L \xrightarrow{\coprod_{i \in [k]} \, p_i\tau} \mathcal{B} )/
		(1 - \mathsf{P} ( L \xrightarrow{\coprod_{i \in [k]} \, p_i\tau} [L]_{\mathcal{E}} )).
	\end{flalign*}
	
	\begin{definition}[$q$-transition \cite{fu_ModelIndependentApproach_2021}]\label{defn-qtrans}
		Suppose $\mathcal{B} \in \PRCCS/\mathcal{E}$ and $\mathcal{B} \ne [A]_{\mathcal{E}}$. 
		We say that there is a $q$-transition from $A$ to $\mathcal{B}$ with regard to $\mathcal{E}$, written $A \rightsquigarrow_{\mathcal{E}}\xrightarrow{q} \mathcal{B}$, if there exists a regular $\epsilon$-tree $t_{\mathcal{E}}^{A}$ satisfying that  for every leaf $L$ of $t_{\mathcal{E}}^{A}$, there exists a collective silent transition $L \xrightarrow{\coprod_{i \in [k]}p_i \tau} \coprod_{i \in [k]} L_i$ such that the normalized probability $\mathsf{P}_{\mathcal{E}}(L\xrightarrow{\coprod_{i \in [k]}p_i \tau}\mathcal{B}) = q$.
	\end{definition}
	

	\begin{example}
		Consider the process $P_4$ and equivalence $\mathcal{E}$ in Example \ref{example-etree}. Since the conditional probability of $P_4$ arrives at $[P_2]_{\mathcal{E}}$ given that it leaves $[P_4]_{\mathcal{E}}$ is $\frac{1}{5}/(1-\frac{3}{5}) = \frac{1}{2}$, the $\epsilon$-tree containing only one single node $P_4$ induces the $q$-transition $P_4 \rightsquigarrow_{\mathcal{E}}\xrightarrow{1/2} [P_2]_{\mathcal{E}}$. Similarly, we can show that there exists the $q$-transition $P_4 \rightsquigarrow_{\mathcal{E}}\xrightarrow{1/2} [\mathbf{0}]_{\mathcal{E}}$.
	\end{example}

	Now we present the definition of branching bisimulation for $\mathrm{RCCS}_{fs}$.
	
	\begin{definition}[Branching bisimulation \cite{fu_ModelIndependentApproach_2021}]\label{defn-BB}
		An equivalence $\mathcal{E}$ on $\PRCCS$ is a \emph{branching bisimulation} if, whenever $(A, B) \in \mathcal{E}$, then for all $\mathcal{C} \in \PRCCS/\mathcal{E}$ it holds that:
		\begin{enumerate}
			\item If $A \rightsquigarrow_{\mathcal{E}}\xrightarrow{\ell} \mathcal{C}$ and $(\ell \in \mathcal{A}) \vee \left(\ell = \tau \wedge \mathcal{C} \ne [A]_{\mathcal{E}} \right)$,
			then $B \rightsquigarrow_{\mathcal{E}}\xrightarrow{\ell} \mathcal{C}$.
			\item If $A  \rightsquigarrow_{\mathcal{E}}\xrightarrow{q} \mathcal{C}$ such that $\mathcal{C} \ne [A]_{\mathcal{E}}$, then $B \rightsquigarrow_{\mathcal{E}}\xrightarrow{q} \mathcal{C}$.
		\end{enumerate}
		We write $A \simeq B$ if there is a branching bisimulation $\mathcal{E}$ such that $(A,B) \in \mathcal{E}$.
	\end{definition}

	\begin{example}\label{example-bb}
		Consider the equivalence $\mathcal{E}$ in Example \ref{example-etree}. 
		It is not hard to verify that $\mathcal{E}$ is a branching bisimulation.
		Here the $\ell$-transition $P \rightsquigarrow_{\mathcal{E}}\xrightarrow{a} [P_4]_{\mathcal{E}}$ can be bisimulated by $P_1$, $P_2$ and $P_3$. 
		We also see that the $q$-transition $P_4 \rightsquigarrow_{\mathcal{E}}\xrightarrow{1/2} [P_2]_{\mathcal{E}}$ and $P_4 \rightsquigarrow_{\mathcal{E}}\xrightarrow{1/2} [\mathbf{0}]_{\mathcal{E}}$ can be bisimulated by $P_5$.
	\end{example}
	
	For a relation $\mathcal{E}$ on $\PRCCS$, we write $\mathcal{E}^{*}$ for its equivalence closure.
	
	\begin{lemma}[\cite{fu_ModelIndependentApproach_2021}]
		\label{lem-bb-equivalence}
		If $\{\mathcal{E}_i\}_{i \in I}$ is a collection of branching bisimulations, then $\mathcal{E} = (\bigcup_{i \in I} \mathcal{E}_i)^{*}$ is also a branching bisimulation.
	\end{lemma}

	\begin{theorem}\label{thm-bb-equivalence}
	The relation $\simeq$ is the largest branching bisimulation, and it is an equivalence relation.
	\end{theorem}
	
    \subsubsection{Weak bisimilarity}\label{subsec-Wb}
    \vspace*{-3pt}
    We start by recalling the necessary notions for defining the weak bisimilarity for $\mathrm{RCCS}_{fs}$.
    A probabilistic distribution over a countable set $S$ is a function $\rho: S \to [0,1]$ such that $\sum_{A \in S} \rho(A) = 1$.
    We denote by $\mathsf{Distr}(S)$ the set of probabilistic distributions over $S$.
    For $S' \subseteq S$, we define $\rho(S') = \sum_{A \in S'} \rho(A)$.
    We use $\delta_A$ to denote the \emph{Dirac} distributions, defined by $\delta_A(A)=1$ and $\delta_A(A')=0$ for all $A'\ne A$.
    The \emph{support} of a probabilistic distribution $\rho$, denoted by $\mathsf{Supp}(\rho)$, is the set $\{A\mid \rho(A)>0\}$.
    For a distribution with finite support, we also write $\rho = \{(A: \rho(A))\mid A \in \mathsf{Supp}(\rho)\}$ to enumerate the probability associated with each element of $\mathsf{Supp}(\rho)$.
    Given a countable set of distributions $\{\rho_i \in \mathsf{Distr}(S)\}_{i \in I}$ and a countable set of real numbers $\{c_i \in [0,1]\}_{i \in I}$ such that $\sum_{i \in I} c_i = 1$, we say that $\rho$ is the \emph{convex combination} of $\{\rho_i\}_{i\in I}$ according to $\{c_i\}_{i\in I}$, denoted by $\sum_{i\in I} c_i\cdot \rho_i$, if for each $A \in S$, $\rho(A) = \sum_{i\in I} c_i\cdot \rho_i(A)$.
    
    To define the weak bisimilarity in $\mathrm{RCCS}_{fs}$, we need to introduce \emph{a probabilistic labeled transition system} (pLTS for short).
    The system is defined in Figure \ref{fig-pLTS}, where $\beta \in Act$ and the probabilistic transition relation $\longrightarrow\;\subseteq \PRCCS\times Act \times \mathsf{Distr}(\PRCCS)$. Although we use the same symbol $\longrightarrow$ for the LTS and pLTS rules, its meaning should be clear from the context.
    \begin{figure}[htb]	
    	\vspace{-5mm}
    	\begin{center}
    		\[	\frac{}{~\sum_{i \in I} \alpha_i.T_i \xrightarrow{\alpha_i} \delta_{T_i}~} \qquad
    		\frac{}{~\bigoplus_{i \in I}p_i\tau.T_i \xrightarrow{\tau} \{T_i: p_i\}_{i\in I}~} \qquad
    		\frac{T\{\mu X.T/X\} \xrightarrow{\beta} \rho}{\mu X.T \xrightarrow{\beta} \rho}
    		\]
    	\end{center}
    	\vspace{-2mm}
    	\caption{pLTS for $\mathrm{RCCS}_{fs}$.}\label{fig-pLTS}
    	\vspace{-2mm}
    \end{figure}
    
    Let $\mathsf{Tr} = \{(A,\alpha,\rho)\mid \text{$A \xrightarrow{\alpha} \rho$ can be derived in the pLTS}\}$ be the set of transitions.
    For a transition $\mathsf{tr} = (A,\alpha,\rho)$, we denote by $src(\mathsf{tr})$ the source process $A$, by $act(\mathsf{tr})$ the action $\alpha$, and by $\rho_{\mathsf{tr}}$ the evolved distribution $\rho$.
    Let $\mathsf{Tr}(\alpha) = \{\mathsf{tr}\in \mathsf{Tr} \mid act(\mathsf{tr}) = \alpha\}$.
    An \emph{execution fragment} of some process $A_0$ is a finite or infinite sequence of alternating states and actions $\omega = A_0 \alpha_0 A_1 \alpha_1 A_2 \alpha_2 \cdots$ such that $A_i \xrightarrow{\alpha_i} \rho_i$ and $\rho_i(A_{i+1}) >0$.
    If $\omega$ is finite, we denote by $\mathsf{last}(\omega)$ the last state of $\omega$.
    We denote by $frags^{*}(A)$ and $frags(A)$ the set of finite and all execution fragments of $A$, respectively.
    Given $\alpha\in Act$, we define $\widehat{\alpha} = \alpha$ if $\alpha \in \mathcal{A}$, and $\widehat{\alpha} = \epsilon$ (the empty string) if $\alpha =\tau$.
    The \emph{trace} of an execution fragment $\omega$ is the sub-sequence of external actions of $\omega$, i.e., $trace(\omega) = \widehat{\alpha_0} \widehat{\alpha_1}\widehat{\alpha_2}\cdots$.
    
    In \cite{turrini_PolynomialTimeDecision_2015}, the notion of \emph{scheduler} is used  to resolve non-determinism.
    To a process $A$, a scheduler is a function $\sigma:frags^{*}(A) \to \mathsf{Distr}(\mathsf{Tr}\cup \{\bot\})$ such that for each $\omega \in frags^{*}(A), \sigma(\omega) \in \mathsf{Distr}(\{\mathsf{tr}\in\mathsf{Tr}\mid src(\mathsf{tr})= last(\omega)\}\cup \{\bot\})$. Intuitively a scheduler specifies a distribution over possible next transitions starting from state $last(\omega)$. 
    If a scheduler takes the special value $\bot$, it chooses no further transition and terminates.
    We call a scheduler (of $A$) $Dirac$ if for each $\omega\in frags^{*}(A)$, $\sigma(\omega) = \delta_{\mathsf{tr}}$ for some $\mathsf{tr} \in \mathsf{Tr}$ or $\sigma(\omega) = \delta_{\bot}$.
    A scheduler $\sigma$ and a process $A$ induce a probability distribution $\rho_{\sigma,A}$ over finite execution fragments as follows.
    The basic measurable events are the \emph{cones} of finite execution fragments, where the cone of $\omega$ is defined by $C_{\omega} = \{\omega'\in frags(A)\mid \text{$\omega$ is a prefix of $\omega'$}\}$.
    The probability $\rho_{\sigma,A}$ of a cone $C_{\omega}$ is defined recursively as follows:
    $$
    \rho_{\sigma,A}(C_{\omega})= \begin{cases}
    	1, & \text{if $\omega = A$,} \\
    	0, & \text{if $\omega = B$ for a process $B\ne A$,} \\
    	\rho_{\sigma,A}(C_{\omega'})\cdot\sum_{\mathsf{tr}\in \mathsf{Tr}(\alpha)}\sigma(\omega')(\mathsf{tr})\cdot \rho_{\mathsf{tr}}(B), & \text{if $\omega = \omega'\alpha B$.}
    \end{cases}
    $$
    Finally, for any $\omega \in frags^{*}(A)$, $\rho_{\sigma,A}(\omega)$
    is defined as $\rho_{\sigma,A}(\omega) = \rho_{\sigma,A}(C_{\omega}) \cdot \sigma(\omega)(\bot)$, where $\sigma(\omega)(\bot)$ is the probability of choosing no further transition  (i.e., terminating) after $\omega$.
    
    The next definition of weak combined transition is standard \cite{segala_ModelingVerificationRandomized_1995, turrini_PolynomialTimeDecision_2015}.
    The fact that state $A$ can weakly transfer to  distribution $\rho$ by executing an observable action $\alpha$ is defined as follows: if there exists a scheduler $\sigma$, from $A$ by doing $\alpha$ and a finite number of silent actions following $\sigma$, the probability of the final state being $B$ equals $\rho(B)$.
    \begin{definition}[Weak combined transition]
    	Given a process $A \in \PRCCS$, an action $\alpha \in Act$ and a distribution $\rho \in \mathsf{Distr}(\PRCCS)$, we say that there is a \emph{weak combined transition} from $A$ to $\rho$ labeled by $\alpha$, denoted by $A \stackrel{\alpha}{\Longrightarrow}_{c} \rho$, if there exists a scheduler $\sigma$ such that the following holds for the induced distribution $\rho_{\sigma,A}$:
    	\begin{enumerate}
    		\item $\rho_{\sigma,A}(frags^{*}(A)) = 1$.
    		\item For each $\omega \in frags^{*}(A)$, if $\rho_{\sigma,A}(\omega)>0$ then $trace(\omega) = \widehat{\alpha}$.
    		\item For each process $B$, $\rho_{\sigma,A}\{\omega\in frags^{*}(A)\mid last(\omega)=B\} = \rho(B)$.
    	\end{enumerate}
    \end{definition}
    
%

     \begin{definition}[Relation lifting \cite{turrini_PolynomialTimeDecision_2015}]
    	Given a binary relation $\mathcal{R}\subseteq X \times Y$. The lifting of $\mathcal{R}$ is the relation $\mathcal{R}^{\dagger} \subseteq \mathsf{Distr}(X) \times\mathsf{Distr}(Y)$ satisfying that $(\rho_X,\rho_Y) \in \mathcal{R}^{\dagger}$ iff there exists a \emph{weighting function} $\mathsf{w}: X \times Y \to [0,1]$ such that
    	\begin{itemize}
    		\item $\mathsf{w}(x,y) > 0$ implies $(x,y) \in \mathcal{R}$,
    		\item $\sum_{y \in Y} \mathsf{w}(x,y) = \rho_X(x)$, and 
    		\item $\sum_{x \in X} \mathsf{w}(x,y) = \rho_Y(y)$.
    	\end{itemize}
    \end{definition}
    The following equivalent definition of relation lifting comes from \cite{deng_LocalAlgorithmChecking_2009}.
    \begin{proposition}[\cite{deng_LocalAlgorithmChecking_2009}, Proposition $2.3$ $(1)$]\label{prop-lifting}
    	Given a binary relation $\mathcal{R}\subseteq X \times Y$ and two distributions $\rho_X \in \mathsf{Distr}(X)$, $\rho_Y \in \mathsf{Distr}(Y)$.
    	Then $(\rho_X,\rho_Y) \in \mathcal{R}^{\dagger}$ iff there exists an index set $I$ and a set of weights $p_i \in (0,1]$ with $\sum_{i \in I} p_i = 1$ such that 
    	\begin{itemize}
    		\item $\rho_X= \sum_{i \in I}p_i \delta_{x_i}$,
    		\item $\rho_Y = \sum_{i \in I}p_i \delta_{y_i}$, and 
    		\item $(x_i,y_i) \in \mathcal{R}$ for all $i \in I$.
    	\end{itemize}
    \end{proposition}
	An immediate result of Proposition \ref{prop-lifting} is the following theorem.
	\begin{theorem}[\cite{deng_LocalAlgorithmChecking_2009}, Proposition $2.3$ $(2)$]
		Given an equivalence $\mathcal{E}$ on a set $X$ and two distributions $\rho_1, \rho_2 \in  \mathsf{Distr}(X)$. Then $(\rho_1,\rho_2) \in \mathcal{E}^{\dagger}$ iff for each $\mathcal{C} \in X \slash \mathcal{E}$, $\rho_1(\mathcal{C}) = \rho_2(\mathcal{C})$.
	\end{theorem}
	Now for an equivalence $\mathcal{E}$, we often use $\rho_1=_{\mathcal{E}} \rho_2$ to denote that $(\rho_1,\rho_2) \in \mathcal{E}^{\dagger}$.
    The next definition resembles the traditional conception for probabilistic automata \cite{he_DivergencesensitiveWeakProbabilistic_2023, turrini_PolynomialTimeDecision_2015}.
    \begin{definition}[Weak bisimulation]\label{defn-WB}
    	An equivalence $\mathcal{E}$ on $\PRCCS$ is a \emph{weak bisimulation} if, for all $(A,B) \in \mathcal{E}$, if $A \xrightarrow{\alpha} \rho_A$, then there exists $\rho_B$ such that $B \stackrel{\alpha}{\Longrightarrow}_{c} \rho_B$ and $\rho_A\; \mathcal{E}^{\dagger} \;\rho_B$.
    	
    	We write $A \approx B$ if there is a weak bisimulation $\mathcal{E}$ such that $(A,B) \in \mathcal{E}$.
    \end{definition}
    
    \begin{theorem}[\cite{turrini_PolynomialTimeDecision_2015}]\label{thm-wb-equivalence}
	The relation $\approx$ is the largest weak bisimulation, and it is an equivalence relation.
    \end{theorem}

	\subsection{The comparison of branching and weak bisimulation semantics}\label{subsec-BbInWb}
	The comparison between branching and weak semantics for probabilistic models is not trivial because their definitions are quite different.
	To establish the containment relationship between them, we need to find a way to relate $\epsilon$-trees with probability distributions. To this end, we start by proving some technical lemmas.
	Lemma \ref{lem-distribution} states that given a set of distributions with the same conditional probability of leaving some $[A]_{\mathcal{E}}$ for any other equivalence class $\mathcal{C}$, the convex combination of these distributions will not change the corresponding conditional probability.
	
	\begin{lemma}\label{lem-distribution}
		Given a process $A$ and an equivalence $\mathcal{E}$ on $\PRCCS$.
		Let $\{\rho_i\}_{i \in I}$ be a countable set of distributions satisfying the following for all $i \in I$:
		\begin{itemize}
			\item $\rho_i([A]_{\mathcal{E}}) = p_i <1$.
			\item For each equivalence class $\mathcal{C} \in \PRCCS/\mathcal{E}$ and $\mathcal{C} \ne [A]_{\mathcal{E}}$, the conditional probability ${\rho_i}_{\mid !A}(\mathcal{C})$ is a constant $q_{\mathcal{C}}$, where ${\rho_i}_{\mid !A}(\mathcal{C}) = \rho_i(\mathcal{C})/(1 - \rho_i([A]_{\mathcal{E}}))$.
		\end{itemize}
		Then for any convex combination $\rho = \sum_{i\in I} c_i \rho_i$ of $\{\rho_i\}_{i \in I}$ according to $\{c_i\}_{i\in I}$, we have $\rho([A]_{\mathcal{E}}) <1$ and ${\rho}_{\mid !A}(\mathcal{C}) = q_{\mathcal{C}}$ for all $\mathcal{C} \in \PRCCS/\mathcal{E}$ and $\mathcal{C} \ne [A]_{\mathcal{E}}$.
	\end{lemma}
	
	\begin{proof}
		By putting all processes in the same equivalence classes together, we can represent each distribution $\rho_i$ as $p_i\cdot [A]_{\mathcal{E}} + \sum_{j \in J_i} (p_{ij} \cdot [A_j]_{\mathcal{E}})$, where $J_i$ is a finite index set and $[A]_{\mathcal{E}}$, $\{[A_j]_{\mathcal{E}}\}_{j \in J_i}$ are all different equivalence classes (we reuse the addition sign to stand for the combination of distributions over equivalent classes with respect to $\mathcal{E}$).
		Now, by the assumption of the lemma, we have
		\begin{enumerate}[(1)]
			\item $\forall i \in I: J_i = J$ for some constant index set $J$;
			\item $\forall i\in I: p_i < 1 \text{ and } p_i + \sum_{j \in J}p_{ij} = 1$; \label{distribution-fact2}
			\item $\forall i\in I, j\in J: p_{ij}/(1 - p_i) = q_j$, where $q_j$ is a constant for each fixed $j$;\label{distribution-fact3}
			\item $\sum_{j \in J} q_j =1$.\label{distribution-fact4}
		\end{enumerate}
		Then each convex combination $\rho = \sum_{i\in I} c_i \rho_i$ can be represented by $\sum_{i\in I} c_i (p_i\cdot [A]_{\mathcal{E}} + \sum_{j \in J} p_{ij} \cdot [A_j]_{\mathcal{E}}) = (\sum_{i\in I} c_i p_i)\cdot [A]_{\mathcal{E}} + \sum_{j \in J}(\sum_{i \in I}c_i p_{ij}) \cdot [A_j]_{\mathcal{E}}$.
		
		Let $r = \sum_{i\in I} c_i p_i$ be the probability $\rho([A]_{\mathcal{E}})$, then $r <1$ follows from (\ref{distribution-fact2}) and $\sum_{i \in I} c_i = 1$.
		Since the conditional probability ${\rho}_{\mid !A}([A_j]_{\mathcal{E}}) = \sum_{i \in I}c_i p_{ij}/(1-r) $, we only need to show that $\sum_{i \in I}c_i p_{ij}/(1-r) = q_j$ holds for all $j \in J$. In fact, we have
		\begin{flalign*}
			\sum_{i \in I}c_i p_{ij} & ~~=~~ \sum_{i \in I}c_i (q_j(1-p_i)) \quad(\text{by (\ref{distribution-fact3})})\\
			& ~~=~~ q_j(\sum_{i \in I}c_i (1-p_i)) = q_j(\sum_{i \in I}c_i - \sum_{i \in I}c_i p_i) = q_j(1-r) \quad(\text{by the definition of $r$})
		\end{flalign*}
		which completes the proof.
	\end{proof}
	
	Given two distributions $\rho, \rho'$ with the same conditional probability of leaving $[A]_{\mathcal{E}}$ to any other equivalence class $\mathcal{C}$, the following lemma shows that if a process $B$ enables a weak combined transition $B \stackrel{\alpha}{\Longrightarrow}_{c} \rho$, then it also enables another weak combined transition $B \stackrel{\alpha}{\Longrightarrow}_{c} \rho''$ for some $\rho''$ such that $\rho''$ is related to $\rho'$ via lifting.
	\begin{lemma}\label{lem-combinedTransition}
		Given a process $A$ and an equivalence $\mathcal{E}$ on $\PRCCS$.
		Let  $\rho,\rho' \in \mathsf{Distr}(\PRCCS)$ be two distributions rendering true the followings:
		\begin{enumerate}
			\item $\rho([A]_{\mathcal{E}}) <1$ and $\rho'([A]_{\mathcal{E}}) <1$;
			\item For all equivalence class $\mathcal{C} \in \PRCCS/\mathcal{E}$ and $\mathcal{C} \ne [A]_{\mathcal{E}}$, the conditional probability ${\rho}_{\mid !A}(\mathcal{C}) = {\rho'}_{\mid !A}(\mathcal{C})$, where ${\rho}_{\mid !A}(\mathcal{C}) = \rho(\mathcal{C})/(1 - \rho([A]_{\mathcal{E}}))$.
		\end{enumerate}
		For any process $B \in [A]_{\mathcal{E}}$,
		if $B \stackrel{\alpha}{\Longrightarrow}_{c} \rho$, then $B \stackrel{\alpha}{\Longrightarrow}_{c} \rho''$ for some $\rho''$ such that $\rho' \;\mathcal{E}^{\dagger}\; \rho''$.
	\end{lemma}
	
	\begin{proof}
		By a similar argument as in the proof of Lemma \ref{lem-distribution}, we can assume that $\rho =_{\mathcal{E}} p\cdot[A]_{\mathcal{E}} + \sum_{i \in I} (p_i \cdot [A_i]_{\mathcal{E}}) $ and $\rho' =_{\mathcal{E}} q\cdot[A]_{\mathcal{E}} + \sum_{i \in I} (q_i \cdot [A_i]_{\mathcal{E}}) $,  where $I$ is a finite index set and $[A]_{\mathcal{E}}$, $\{[A_i]_{\mathcal{E}}\}_{i \in I}$ are pairwise different equivalence classes.
		Moreover, the conditional probability $p_i /(1-p) = q_i/(1-q)$ holds for all $i \in I$.
		
		Now suppose the weak combined transition $B \stackrel{\alpha}{\Longrightarrow}_{c} \rho$ is induced by scheduler $\sigma$, we then construct a scheduler $\sigma''$ that induces $B \stackrel{\alpha}{\Longrightarrow}_{c} \rho''$ and $\rho' \;\mathcal{E}^{\dagger}\; \rho''$ as follows:
		Let $c = (1-q)/(1-p)$ and $\sigma_{\bot}$ be the scheduler choosing no transitions, i.e., $\sigma_{\bot}(B) = \bot$;
		scheduler $\sigma''$ will behave as $\sigma$ with probability $c$ and behave as $\sigma_{\bot}$ with probability $1-c$.
		Since $\sigma_{\bot}$ induces the distribution $\delta_{B}$ and $B \in [A]_{\mathcal{E}}$, the induced distribution by $\sigma''$ would be $\rho'' = c \rho + (1-c)\delta_{B} =_{\mathcal{E}} c(p\cdot[A]_{\mathcal{E}} + \sum_{i \in I} p_i \cdot [A_i]_{\mathcal{E}}) + (1-c) (1\cdot [A]_{\mathcal{E}})$.
		By simple calculation, one can find that the last one equals $ q\cdot[A]_{\mathcal{E}} + \sum_{i \in I} q_i \cdot [A_i]_{\mathcal{E}} =_{\mathcal{E}} \rho'$. \end{proof}
	
	With the above preparation, we can prove that branching bisimilarity implies weak bisimilarity (Theorem \ref{thm-BbInWb}). 
	As far as we know, it is the first time that the $\epsilon$-tree based branching bisimilarity and the distribution-based weak bisimilarity are compared in the setting of probabilistic models.

	\begin{theorem}
		[$\simeq \;\subseteq\; \approx$]\label{thm-BbInWb}
		If $\mathcal{E}$ is a branching bisimulation, then $\mathcal{E}$ is a weak bisimulation.
	\end{theorem}
	
	\begin{proof}
		Let $\mathcal{E}$ be a branching bisimulation. Suppose $(A,B) \in \mathcal{E}$ and $A \xrightarrow{\alpha} \rho_A$, according to Definition \ref{defn-WB}, we need to show that there exists $\rho_B$ such that $B \stackrel{\alpha}{\Longrightarrow}_{c} \rho_B$ and $\rho_A\; \mathcal{E}^{\dagger} \;\rho_B$.
		We focus on the most difficult case, i.e., $\alpha =\tau$ and there exists $A' \in \mathsf{Supp}(\rho_A)$ such that $(A,A') \notin  \mathcal{E}$.
		
		Now assume that $\rho_A = p\cdot[A]_{\mathcal{E}} + \sum_{i \in I} (p_i \cdot [A_i]_{\mathcal{E}}) $,  where $p < 1$, $\sum_{i \in I} p_i = 1 -p$, and $[A]_{\mathcal{E}}$, $\{[A_i]_{\mathcal{E}}\}_{i \in I}$ are pairwise different equivalence classes.
		Since $(A,B) \in \mathcal{E}$ and $\mathcal{E}$ is a branching bisimulation, $B \rightsquigarrow_{\mathcal{E}}\xrightarrow{q_i} [A_i]_{\mathcal{E}}$ for all $i \in I$, where $q_i = p_i/(1-p)$.
		According to Definition \ref{defn-qtrans}, there exists a regular $\epsilon$-tree $t_{\mathcal{E}}^{B}$ satisfying that for every leaf $L$ of $t_{\mathcal{E}}^{B}$, $L \xrightarrow{\coprod_{j \in J}r_j \tau} \coprod_{j \in J} M_j$ and $\mathsf{P}_{\mathcal{E}}\left(L\xrightarrow{\coprod_{j \in J}r_j \tau} [A_i]_{\mathcal{E}}\right) = q_i$ for all $i \in I$.
		Let $\{L_k\}_{k \in K}$ be the countable set of leaves in the tree $t_{\mathcal{E}}^{B}$ and $c_k$ be the probability of the path from $B$ to $L_k$ in $t_{\mathcal{E}}^{B}$.
		Since $t_{\mathcal{E}}^{B}$ is a regular tree, $\sum_{k \in K} c_k = 1$.
		Let $\rho_{k}$ be the induced probability distribution of the transition $L_k \xrightarrow{\coprod_{j \in J}r_j \tau} \coprod_{j \in J} M_j$.
		Then we can see that the regular $\epsilon$-tree $t_{\mathcal{E}}^{B}$ induces a Dirac scheduler $\sigma$ and a distribution $\rho = \sum_{k \in K} c_k \rho_{k}$ such that $B \stackrel{\tau}{\Longrightarrow}_{c} \rho$.
		Since $\rho_{k}([A]_{\mathcal{E}}) < 1$ and ${\rho_k}_{\mid !A}([A_i]_{\mathcal{E}}) = q_i$ hold for all $k \in K$ and $i \in I$, according to Lemma \ref{lem-distribution}, we have $\rho([A]_{\mathcal{E}}) < 1$ and ${\rho}_{\mid !A}([A_i]_{\mathcal{E}}) = q_i$ for all $i \in I$.
		Now since $B \stackrel{\tau}{\Longrightarrow}_{c} \rho$, by Lemma \ref{lem-combinedTransition}, there exists a distribution $\rho_B$ such that $B \stackrel{\tau}{\Longrightarrow}_{c} \rho_B$ and $\rho_A\; \mathcal{E}^{\dagger} \;\rho_B$. So we are done.
	\end{proof}
	
	At the end of this part, we want to mention that there are some other bisimilarities based on a similar notion of $\epsilon$-tree or probabilistic distribution in  the literature.
	For example, the \emph{probabilistic weak bisimilarity} proposed in \cite{wu_ProbabilisticWeakBisimulation_2023} relies on the notion of \emph{weak $\epsilon$-tree} while the \emph{branching probabilistic bisimilarity} defined in \cite{turrini_PolynomialTimeDecision_2015} is based on probabilistic distribution.
	The techniques developed in this subsection should be helpful in establishing the relationship among them and the branching (weak) bisimilarities defined in this paper.
	

	\section{Divergence in probabilistic branching bisimulation semantics}\label{sec-branching}
    In this section we turn to the issue of divergence.
    We propose two bisimulation relations that interpret divergence with varying strength.
    The first one, branching bisimulation with explicit divergence (Definition \ref{def-dsbb}), is a probabilistic extension of the equivalence studied in~\cite{liu_AnalyzingDivergenceBisimulation_2017}.
    The second one, exhaustive branching bisimulation (Definition \ref{def-ebb}), is an instructive graph-based equivalence that extends the complete branching bisimulation in \cite{liu_AnalyzingDivergenceBisimulation_2017}.
    A similar equivalence is also used in \cite{he_DivergencesensitiveWeakProbabilistic_2023} to characterize divergence-sensitive weak bisimulation.
    The relationship among these equivalence relations will be discussed in Section \ref{sec-spectrum}.

    \subsection{Branching bisimilarity with explicit divergence}
    Following Definition \ref{defn-regdiv}, for an equivalence relation $\mathcal{E}$ on $\PRCCS$,  a process $A$ is divergent with respect to $\mathcal{E}$, denoted by $A\Uparrow_{\mathcal{E}}$, if there exists a divergent $\epsilon$-tree $t_{\mathcal{E}}^A$ of $A$ with regard to $\mathcal{E}$. We use $A \not\Uparrow_{\mathcal{E}}$ to denote that $A$ is not divergent with respect to $\mathcal{E}$.
    
    \begin{definition}[Divergent $\epsilon$-tree preserving]\label{defn-tree-inv}
    	Let $\mathcal{E}$ be an equivalence on $\PRCCS$. $\mathcal{E}$ is \emph{divergent $\epsilon$-tree preserving} if for all $(A,B)\in\mathcal{E}$ the following holds: $A\Uparrow_{\mathcal{E}}$ if and only if $B\Uparrow_{\mathcal{E}}$.
    \end{definition}
    
    The following definition is an extension of the corresponding notion proposed in \cite{vanglabbeek_BranchingBisimilarityExplicit_2009, liu_AnalyzingDivergenceBisimulation_2017} for $\PRCCS$.
    
    \begin{definition}[Branching bisimulation with explicit divergence]\label{def-dsbb}
    	Let $\mathcal{E}$ be an equivalence on $\PRCCS$.
    	$\mathcal{E}$ is called a \emph{branching bisimulation with explicit divergence} if $\mathcal{E}$ is a branching bisimulation and is divergent $\epsilon$-tree preserving.
    	
    	We write $A\simeq^{\Delta}B$ if there is a branching bisimulation with explicit divergence $\mathcal{E}$ such that $(A,B) \in \mathcal{E}$.
    \end{definition}

    Similar to the non-probabilistic situation~\cite{liu_AnalyzingDivergenceBisimulation_2017,vanglabbeek_BranchingBisimilarityExplicit_2009}, the requirement of being divergent $\epsilon$-tree preserving makes it non-trivial to prove that $\simeq^{\Delta}$ is indeed the largest branching bisimulation with explicit divergence.
    In \cite{fu_ModelIndependentApproach_2021}, the divergence-sensitivity of $\simeq^{\Delta}$ is given as Lemma 4.1 without  proof. However, it should be pointed out that the original statement of this lemma is not correct in our setting. More specifically, the lemma can be rephrased in our language as `If $\{\mathcal{E}_i\}_{i \in I}$ is a collection of divergent $\epsilon$-tree preserving equivalences, then $\mathcal{E} = (\bigcup_{i \in I} \mathcal{E}_i)^{*}$ is also a divergent $\epsilon$-tree preserving equivalence'. 
    There is a simple counterexample to this statement. 
    Let $A_1 = \tau.\mathbf{0}$, $A_2 = \mathbf{0}$, $B_1 = \mu X.(\tau.\tau.\tau.X)$, $B_2 = \tau.\tau.B_1$ and $B_3 = \tau.B_1$. Consider the equivalence $\mathcal{E}_1 = \{\{A_1,B_1,B_2\}, \{B_3\},\{A_2\}\}$ and $\mathcal{E}_2 = \{\{B_1,B_2,B_3\}, \{A_1\},\{A_2\}\}$. It is not hard to check that both $\mathcal{E}_1$ and $\mathcal{E}_2$ are divergent $\epsilon$-tree preserving. Yet the equivalence $\mathcal{E} = (\mathcal{E}_1 \cup \mathcal{E}_2)^{*} = \{\{A_1,B_1,B_2,B_3\}, \{A_2\}\}$ is not divergent $\epsilon$-tree preserving, as $A_1 \not\Uparrow_{\mathcal{E}}$ and $B_1 \Uparrow_{\mathcal{E}}$ hold for the pair $(A_1,B_1) \in \mathcal{E}$.
    The key issue here is that $\mathcal{E}_1$ is not a \emph{branching} bisimulation, as the $\ell$-transition $A_1 \rightsquigarrow_{\mathcal{E}_1}\xrightarrow{\tau} [A_2]_{\mathcal{E}_1}$ cannot be bisimulated by $B_1$.
    We find that the requirement of branching bisimulation is necessary for achieving the divergent $\epsilon$-tree preserving property.
    Then we need the following technical lemma, a probabilistic generalization of the corresponding result in \cite{vanglabbeek_BranchingBisimilarityExplicit_2009}.

    \begin{lemma}\label{lem-dsbb-equivalence}
    	If $\{\mathcal{E}_i\}_{i \in I}$ is a collection of branching bisimulation with explicit divergence, then $\mathcal{E} = (\bigcup_{i \in I} \mathcal{E}_i)^{*}$ is also a branching bisimulation with explicit divergence.
    \end{lemma}
    \begin{proof}
    	Since each $\mathcal{E}_i$ is a branching bisimulation, by Lemma \ref{lem-bb-equivalence}, $\mathcal{E}$ is also a branching bisimulation. It suffices to show that for each $i \in I$ and $(A,B) \in \mathcal{E}_i\subseteq\mathcal{E}$, the following divergence-sensitive property holds: if $A \Uparrow_{\mathcal{E}}$ then $B \Uparrow_{\mathcal{E}}$.
    	It should be emphasized that the  pair $(A,B)$ is in $\mathcal{E}_i$, while the divergence property is required with regard to $\mathcal{E}$.
    	Suppose $A \Uparrow_{\mathcal{E}}$, then there exists a divergent $\epsilon$-tree $t_{\mathcal{E}}^A$ of $A$ with regard to $\mathcal{E}$.
    	We will construct a divergent $\epsilon$-tree $t_{\mathcal{E}}^B$ of $B$ with regard to $\mathcal{E}$ by induction on the structure of $t_{\mathcal{E}}^A$. There are two cases.
    	\begin{itemize}
    		\item $t_{\mathcal{E}}^A$ is an $\epsilon$-tree of $A$ with regard to $\mathcal{E}_i$.
    		In this case, all nodes in tree $t_{\mathcal{E}}^A$ belong to $[A]_{\mathcal{E}_i}$, therefore $A \Uparrow_{\mathcal{E}_i}$. As $\mathcal{E}_i$ is a branching bisimulation with explicit divergence, we have $B \Uparrow_{\mathcal{E}_i}$,
    		which implies a divergent $\epsilon$-tree $t_{\mathcal{E}}^B$ of $B$ with regard to $\mathcal{E}$.
    		
    		\item $t_{\mathcal{E}}^A$ is not an $\epsilon$-tree of $A$ with regard to $\mathcal{E}_i$. In this case, there exist some nodes in tree $t_{\mathcal{E}}^A$ which do not belong to $[A]_{\mathcal{E}_i}$. We only consider the case that $A$ has multiple children (the case of having only one child is similar and easier).
    		Assume that $A$ has $k$ children $A^1, \cdots, A^k$ with the corresponding edges labeled by $p_1, \cdots, p_k$ respectively. In other words, $A \xrightarrow{\coprod_{j \in [k]} \, p_j} \coprod_{j \in [k]} A^j$.
    		There are two sub-cases:
    		\begin{enumerate}[i)]
    			\item  $A^j \mathcal{E}_i A$ for all $j \in [k]$, i.e., $A \xrightarrow{\coprod_{j \in [k]} \, p_j} \coprod_{j \in [k]} A^j$ is state-preserving.
    			Since $\mathcal{E}_i$ is an equivalence and $A \mathcal{E}_i B$, we have $A^j \mathcal{E}_i B$ for all $j \in [k]$. Thus we can continue to construct $t_{\mathcal{E}}^{B}$ by structural induction on the divergent $\epsilon$-tree of $A^1$.
    			\item \label{lem-dsbb-case2}
    			$A^{j}\not\in [A]_{\mathcal{E}_i}$ for some $j\in [k]$, i.e., $A \xrightarrow{\coprod_{j \in [k]} \, p_j} \coprod_{j \in [k]} A^j$ is state-changing.
    			Suppose without loss of generality that $A^1 \notin [A]_{\mathcal{E}_i}$.
    			Let $q = \mathsf{P}_{\mathcal{E}_i} \left ( A \xrightarrow{\coprod_{j \in [k]} \, p_k} [A^1]_{\mathcal{E}_i} \right )$. Then $B \rightsquigarrow_{\mathcal{E}_i} \xrightarrow{q} [A^1]_{\mathcal{E}_i}$ follows from the fact that $\mathcal{E}_i$ is a branching bisimulation.
    			The $q$-transition consists of a regular $\epsilon$-tree ${t'}_{\mathcal{E}_i}^{B}$ of $B$ with regard to $\mathcal{E}_i$ and, for each leaf $B''$ of ${t'}_{\mathcal{E}_i}^{B}$, there exists a collective transition $B'' \xrightarrow{\coprod_{j' \in [k']} \, r_{j'}} \coprod_{j' \in [k']} B^{j'}$ such that the normalized probability $\mathsf{P}_{\mathcal{E}_i} \left ( B'' \xrightarrow{\coprod_{j' \in [k']} \,  r_{j'}} [A^1]_{\mathcal{E}_i} \right ) = q$. For every process $B^{j'} \in [A^1]_{\mathcal{E}_i}$, we continue to construct an $\epsilon$-tree of $B^{j'}$ by structural induction on the divergent $\epsilon$-tree of $A^1$. According to our construction, we have $B'' \mathcal{E}_i B$ and $B^{j'} \mathcal{E}_i A^1$. By assumption, we have $A \mathcal{E}_i B$ and $A \mathcal{E} A^1$. For $\mathcal{E}_i$ and $\mathcal{E}$ are equivalence, we can get $B''\mathcal{E}_i B\mathcal{E}_i A\mathcal{E}A^{1}\mathcal{E}_i B^{j'}$. As $\mathcal{E}_i \subseteq \mathcal{E}$, we can further get that $B \mathcal{E} B'' \mathcal{E} B^{j'}$. In other words, $B^{j'}$ is indeed a node in $t^{B}_{\mathcal{E}}$.
    		\end{enumerate}
    	\end{itemize}
    	
    	We now prove that the above constructed $t_{\mathcal{E}}^B$ is a divergent $\epsilon$-tree. Two cases are possible.
    	In the first case, $A$ can go through infinite state-changing transitions with regard to $\mathcal{E}_i$ and never reach any $\epsilon$-tree with regard to $\mathcal{E}_i$. Since each state-changing transition cannot be bisimulated vacuously in \ref{lem-dsbb-case2}), we see that the constructed $t_{\mathcal{E}}^B$ is a divergent tree in this case.
    	In the second case, $A$ will reach an $\epsilon$-tree $(t')_{\mathcal{E}_{i}}^{A'}$ of some $A'$ with regard to $\mathcal{E}_i$ (where $A'$ is a node in $t_{\mathcal{E}}^A$) after finite state-preserving transitions with regard to $\mathcal{E}_i$.
    	In this case, since $(t')_{\mathcal{E}_{i}}^{A'}$ is a divergent $\epsilon$-tree, we have $A' \Uparrow_{\mathcal{E}_i}$.
    	According to the above construction procedure, there exists a process $B'$ in $t_{\mathcal{E}}^B$ satisfying that $A' \mathcal{E}_i B'$. Since $\mathcal{E}_i$ is a branching bisimulation with explicit divergence, we have $B' \Uparrow_{\mathcal{E}_i}$.
    	Therefore the constructed $t_{\mathcal{E}}^B$ is also a divergent tree in this case.
    \end{proof}
    
    Theorem \ref{thm-dsbb-equivalence} then follows directly from Lemma \ref{lem-dsbb-equivalence}.
    
    \begin{theorem}\label{thm-dsbb-equivalence}
	The relation $\simeq^{\Delta}$ is the largest branching bisimulation with explicit divergence, and it is an equivalence.
    \end{theorem}
	
	\subsection{Exhaustive branching bisimilarity}
	In \cite{liu_AnalyzingDivergenceBisimulation_2017},  the concept of \emph{divergence set} is proposed to define the so-called complete weak bisimilarity, which gives rise to an alternative characterization of weak bisimilarity with explicit divergence in the non-probabilistic scenario.
	In \cite{he_DivergencesensitiveWeakProbabilistic_2023}, a similar concept, \emph{$\tau$-end component} ($\tau$-EC), is introduced for probabilistic automata, based on which the authors defined exhaustive weak probabilistic bisimulations.  The basic idea behinds these  definitions is that they consider a process to be divergent if it can \emph{reach} a silent circle in finitely many silent steps.
	Next we extend the concept to probabilistic branching bisimulation. 
	
	We start with the following reformulation of $\tau$-EC \cite{he_DivergencesensitiveWeakProbabilistic_2023} in $\mathrm{RCCS}_{fs}$ model.
	
	\begin{definition}[$\tau$-EC]
	Given a process $B\in \PRCCS$.
	Let $G_B = (V_B, E_B)$ be the induced transition graph of $B$.
	A $\tau$-EC (of $B$), denoted by $\mathsf{ec}=(V_{\mathsf{ec}}, E_{\mathsf{ec}})$, is a subgraph of $G_B$ satisfying:
	\begin{enumerate}
		\item $\mathsf{ec}=(V_{\mathsf{ec}}, E_{\mathsf{ec}})$ is strongly connected;
		\item All edges in $E_{\mathsf{ec}}$ are restricted to be labeled with $\tau$ or $p\tau$ (where $p \in (0,1)$);
		\item If there is an edge $e' = (C,C') \in E_{\mathsf{ec}}$ with label $q\tau$, then there must exist some collective silent transition $C \xrightarrow{\coprod_{i \in I}p_i \tau} \coprod_{i \in I} C_i$ such that $q = p_k$ and $C' = C_k$ for some $k\in I$. Moreover, for all $i \in I$, the edge $e_i = (C, C_i) \in E_{\mathsf{ec}}$ and labeled by $p_i\tau$.
	\end{enumerate}
	
	Given a process $B'\in \PRCCS$, we write $B' \circlearrowleft_{\mathsf{ec}}$ to denote that $B'$ is in the $\tau$-EC labeled as  $\mathsf{ec}$.
	\end{definition}
	
	Intuitively speaking, $\tau$-EC is a strongly connected graph which contains only silent transitions and is closed under  probabilistic silent transitions.
	For two given  $\tau$-ECs, we will need to relate them under some binary relation $\mathcal{R}$. Definition \ref{defn-rtauEC} promotes the relation $\mathcal{R}$ between nodes (in $\tau$-EC) to a relation between $\tau$-ECs.
	
	\begin{definition}[Related $\tau$-EC]\label{defn-rtauEC}
	Given a binary relation $\mathcal{R}$, and two $\tau$-ECs $\mathsf{ec}_1 = (V_{\mathsf{ec}_1}, E_{\mathsf{ec}_1})$ and $\mathsf{ec}_2 = (V_{\mathsf{ec}_2}, E_{\mathsf{ec}_2})$.
	We say $\mathsf{ec}_1$ is related to $\mathsf{ec}_2$ with regard to $\mathcal{R}$, denoted by  $\mathsf{ec}_1 \;\mathcal{R}^{\ddagger}\; \mathsf{ec}_2$, iff for all $B \in V_{\mathsf{ec}_2}$ there exists $A \in V_{\mathsf{ec}_1}$ with $(A,B) \in \mathcal{R}$.
	\end{definition}
	
	\begin{remark}
	The notion of \emph{related $\tau$-EC} is actually a generalization of the corresponding requirement used in the definition of \emph{complete weak bisimulation} (Definition 2.8, \cite{liu_AnalyzingDivergenceBisimulation_2017}). The asymmetric requirement in Definition \ref{defn-rtauEC} is necessary for the correctness of Lemma \ref{lem-tauEC-equivalence}.
	\end{remark}
	
	Given two processes $A, B$ and let $G_A$ be the induced transition graph of $A$, we use $A \xRightarrow{} B$ to stand for that $B$  can be reached from $A$ in $G_A$ through a sequence of edges labeled with $\tau$ or $p\tau$ (where $p\in(0,1)$).
	We use $A \xRightarrow{} \circlearrowleft_{\mathsf{ec}}$ to denote that there exists $A'$ such that $A \xRightarrow{} A'$ and $A'\circlearrowleft_{\mathsf{ec}}$.
	
	Given a binary relation $\mathcal{R}$, its \emph{reverse relation} is denoted by $\mathcal{R}^{-1} = \{(B,A) \mid (A,B) \in \mathcal{R}\}$. The \emph{composition} of two relations $\mathcal{R}_1$ and $\mathcal{R}_2$ is denoted by $\mathcal{R}_1 \circ \mathcal{R}_2 = \{(A,C) \mid \exists B. (A,B) \in \mathcal{R}_1 \land (B,C) \in \mathcal{R}_2\}$.
	Before giving the definition of exhaustive branching bisimilarity for $\PRCCS$, we need the following definition that characterizes the divergence-sensitive property with regard to $\tau$-EC.
	

	\begin{definition}[$\tau$-EC invariant]\label{defn-rtauEC-inv}
		Given a binary relation $\mathcal{R}$ on $\PRCCS$, we say $\mathcal{R}$ is \emph{$\tau$-EC invariant} if for all $(A,B)\in\mathcal{R}$ the following hold: 
		\begin{enumerate}
			\item whenever $A \xRightarrow{} \circlearrowleft_{\mathsf{ec}_1}$, then $B \xRightarrow{} \circlearrowleft_{\mathsf{ec}_2}$ for some  $\tau$-EC $\mathsf{ec}_2$ such that  $\mathsf{ec}_1 \;\mathcal{R}^{\ddagger}\; \mathsf{ec}_2$;
			\item whenever $B \xRightarrow{} \circlearrowleft_{\mathsf{ec}_2}$, then $A \xRightarrow{} \circlearrowleft_{\mathsf{ec}_1}$ for some  $\tau$-EC $\mathsf{ec}_1$ such that  $\mathsf{ec}_2 \;{(\mathcal{R}^{-1})}^{\ddagger}\; \mathsf{ec}_1$.
		\end{enumerate}
	\end{definition}

\begin{lemma}\label{lem-tauEC-equivalence}
	If $\{\mathcal{E}_i\}_{i \in I}$ is a collection of $\tau$-EC invariant equivalences, then $\mathcal{E} = (\bigcup_{i \in I} \mathcal{E}_i)^{*}$ is also a $\tau$-EC invariant equivalence.
\end{lemma}
\begin{proof}
	It is easy to see that the $\tau$-EC invariant property is closed under relation union. Then it suffices to show that the $\tau$-EC invariant property is closed under relation composition. Now suppose both $\mathcal{E}_1$ and $\mathcal{E}_2$ are $\tau$-EC invariant, we will prove that their composition $\mathcal{E}_1 \circ \mathcal{E}_2$ is also $\tau$-EC invariant.
	Consider any pair $(A,C) \in \mathcal{E}_1 \circ \mathcal{E}_2$ with $A \xRightarrow{} \circlearrowleft_{\mathsf{ec}_1}$. According to definition, there exists process $B$ such that $(A,B) \in \mathcal{E}_1$ and $(B,C) \in \mathcal{E}_2$. 
	Since $ \mathcal{E}_1$ is $\tau$-EC invariant and $A \xRightarrow{} \circlearrowleft_{\mathsf{ec}_1}$, we have $B \xRightarrow{} \circlearrowleft_{\mathsf{ec}_2}$ for some  $\tau$-EC $\mathsf{ec}_2$ such that  $\mathsf{ec}_1 \;{\mathcal{E}_1}^{\ddagger}\; \mathsf{ec}_2$.
	Similarly, as $ \mathcal{E}_2$ is $\tau$-EC invariant and $B \xRightarrow{} \circlearrowleft_{\mathsf{ec}_2}$, we have $C \xRightarrow{} \circlearrowleft_{\mathsf{ec}_3}$ for some  $\tau$-EC $\mathsf{ec}_3$ such that  $\mathsf{ec}_2 \;{\mathcal{E}_2}^{\ddagger}\; \mathsf{ec}_3$.
	Since $\mathsf{ec}_2 \;{\mathcal{E}_2}^{\ddagger}\; \mathsf{ec}_3$, according to Definition \ref{defn-rtauEC}, for all $C' \in V_{\mathsf{ec}_3}$ there exists $B' \in V_{\mathsf{ec}_2}$ with $(B', C') \in \mathcal{E}_2$.
	Since $\mathsf{ec}_1 \;{\mathcal{E}_1}^{\ddagger}\; \mathsf{ec}_2$ and $B' \in V_{\mathsf{ec}_2}$, by Definition \ref{defn-rtauEC} again, there exists $A' \in V_{\mathsf{ec}_1}$ with $(A', B') \in \mathcal{E}_1$. Then $(A', C') \in \mathcal{E}_1 \circ \mathcal{E}_2$ follows from $(A', B') \in \mathcal{E}_1$ and $(B', C') \in \mathcal{E}_2$. Since for all $C' \in V_{\mathsf{ec}_3}$ there exists $A' \in V_{\mathsf{ec}_1}$ with $(A', C') \in \mathcal{E}_1 \circ \mathcal{E}_2$, we obtain that $\mathsf{ec}_1 \;(\mathcal{E}_1 \circ \mathcal{E}_2)^{\ddagger}\; \mathsf{ec}_3$.
	Now we have proved that there exists some  $\tau$-EC $\mathsf{ec}_3$ such that $C \xRightarrow{} \circlearrowleft_{\mathsf{ec}_3}$ and $\mathsf{ec}_1 \;(\mathcal{E}_1 \circ \mathcal{E}_2)^{\ddagger}\; \mathsf{ec}_3$. Therefore, $\mathcal{E}_1 \circ \mathcal{E}_2$ is  $\tau$-EC invariant.
\end{proof}
	
	\begin{definition}[Exhaustive branching bisimulation]\label{def-ebb}
	Let $\mathcal{E}$ be an equivalence on $\PRCCS$.
	$\mathcal{E}$ is called an \emph{exhaustive branching bisimulation} if $\mathcal{E}$ is a branching bisimulation and is $\tau$-EC invariant.
	
	We write $A\simeq_{e} B$ if $(A,B) \in \mathcal{E}$ for some exhaustive branching bisimulation $\mathcal{E}$.
	\end{definition}

\begin{lemma}
	If $\{\mathcal{E}_i\}_{i \in I}$ is a collection of exhaustive branching bisimulation, then so is $\mathcal{E} = (\bigcup_{i \in I} \mathcal{E}_i)^{*}$.
\end{lemma}
\begin{proof}
	Immediate by Lemma \ref{lem-bb-equivalence} and Lemma \ref{lem-tauEC-equivalence}.
\end{proof}
		
	\begin{theorem}\label{thm-ebb-equivalence}
	The relation $\simeq_{e}$ is the largest exhaustive branching bisimulation, and it is an equivalence.
	\end{theorem}


	We note that $\simeq^{\Delta}$ and $\simeq_{e}$ treat divergence in different ways. In $\simeq^{\Delta}$, a process $A$ is divergent  if and only if it can diverge with probability 1 (i.e., there exists a divergent $\epsilon$-tree of $A$). In contrast, in $\simeq_{e}$, a process $B$ is divergent if it can diverge with some non-zero probability (i.e., $B$ can reach some $\tau$-EC).
	Here we prove that $\simeq^{\Delta}$ implies $\simeq_{e}$ (Theorem \ref{thm-dsbbVSebb}). The strictness of the implication will be shown in Example \ref{exam-StrictInclusion}.
	
	\begin{remark}
	One may consider defining a finer notion of bisimilarity by requiring that the probability of reaching two related $\tau$-ECs to be the same for two bisimilar processes. Actually a  similar method has been taken to define \emph{probabilistic applicative bisimulation} for the probabilistic $\lambda$-calculus \cite{dallago_CoinductiveEquivalencesHigherorder_2014}. However, this idea does not work directly in our setting. Consider the following $\mathrm{RCCS}_{fs}$ processes: $P=\tau. P_1+ \tau.P_2$, where $P_1=\frac{1}{2}\tau.a\oplus\frac{1}{2}\tau.\Omega$, $P_2=\frac{1}{3}\tau.b\oplus\frac{2}{3}\tau.\Omega$, and $\Omega=\mu X. (\tau. X)$  is an always divergent process. As the first step from $P$ is a nondeterministic choice between $P_1$ and $P_2$, we cannot simply say that $P$ diverge with probability $\frac{1}{2}$ or $\frac{2}{3}$. The reason why such definition does work in the probabilistic $\lambda$-calculus model is the absence of nondeterminism. That is to say, given a probabilistic $\lambda$-term, the induced probabilistic distribution by applying some specific reduction strategy (such as call-by-value or call-by-name strategy) is unique.
	This is also the reason why \emph{schedulers} \cite{turrini_PolynomialTimeDecision_2015} have been used to resolve non-determinism for nondeterministic probabilistic models.
	\end{remark}
	
	The following lemma states that for any two nodes $A,B$ from a $\tau$-EC, we can construct a regular  $\epsilon$-tree whose root is $A$ and every leaf is $B$.
	
	\begin{lemma}\label{lem-EcAndTree}
	Given a $\tau$-EC $\mathsf{ec} = (V_{\mathsf{ec}}, E_{\mathsf{ec}})$, let $\mathcal{E} = V_{\mathsf{ec}} \times V_{\mathsf{ec}}$.
	If $A, B \in V_{\mathsf{ec}}$, then there exists a regular $\epsilon$-tree $t_{\mathcal{E}}^A$ of $A$ with regard to $\mathcal{E}$ such that all leaves of $t_{\mathcal{E}}^A$ are labeled by $B$.
	\end{lemma}
	
	\begin{proof}
	As $V_{\mathsf{ec}}$ is a finite set, for any two nodes $A, B \in V_{\mathsf{ec}}$, let $V_{\mathsf{ec}} \backslash \{B\} = \{A_i~|~{i \in I}\}$ where $I$ is a finite index set.
	Surely $A\in \{A_i~|~{i \in I}\}$.
	As $\mathsf{ec}$ is strongly connected, we can choose a shortest path $\pi_i$ from $A_i$ to $B$ for all $i \in I$.
	We use $(t_0)_{\mathcal{E}}^{A_i}$ to stand for the minimal (finite) $\epsilon$-tree induced by $\pi_i$.
	We will inductively (starting from $(t_0)_{\mathcal{E}}^{A}$) build a regular tree whose leaves are all $B$.
	
	For any tree $t$, we use $\mathsf{P}^{\ne B}(t)$ to denote the probability of all finite paths in $t$ that do not end with $B$. 
	Again by strong connectivity, $\forall i\in I$, $p_i = \mathsf{P}^{\ne B}((t_0)_{\mathcal{E}}^{A_i}) < 1$. Let $p = \max_{i\in I}\{p_i\}$, $p<1$.
	We then inductively build a sequence of $\epsilon$-trees $\{(t_n)_{\mathcal{E}}^{A}\}_{n \in \mathsf{N}}$ of $A$ with regard to $\mathcal{E}$ as follows:
	\begin{center}
		\emph{For each $n\ge 0$, $(t_{n+1})_{\mathcal{E}}^{A}$ is obtained from $(t_{n})_{\mathcal{E}}^{A}$ by replacing every leaf $A_i\neq B$ by $(t_0)_{\mathcal{E}}^{A_i}$.}
	\end{center}
	\noindent   We next show that $\mathsf{P}^{\ne B}((t_{n})_{\mathcal{E}}^{A}) \le p^{n}$ holds for all $n \ge 0$ by induction on $n$.
	\begin{itemize}
		\item The base case $n=0$ holds trivially.
		\item For the induction step, suppose $\mathsf{P}^{\ne B}((t_{n})_{\mathcal{E}}^{A}) \le p^{n}$:\\ For any leaf $A_i\neq B$ in $(t_{n})_{\mathcal{E}}^{A}$, let $(\pi_n)_i$ be the path from $A$ to $A_i$ in $(t_{n})_{\mathcal{E}}^{A}$. Then we have $
		\mathsf{P}^{\ne B}((t_{n+1})_{\mathcal{E}}^{A})  = \sum_{i \in I} \left(\mathsf{P}((\pi_n)_i) \cdot \mathsf{P}^{\ne B}((t_{0})_{\mathcal{E}}^{A_i}) \right)\le \sum_{i \in I} \mathsf{P}((\pi_n)_i) \cdot p = p \cdot\left( \mathsf{P}^{\ne B}((t_{n})_{\mathcal{E}}^{A})\right) \le p \cdot p^n = p^{n+1}
		$, as desired.
	\end{itemize}
	
	To the set of all branches of $(t_{n})_{\mathcal{E}}^{A}$, either they end with leaf $B$, or the probability of the rest (i.e., $\mathsf{P}^{\ne B}((t_{n})_{\mathcal{E}}^{A})$ )  is upper-bounded by $ p^{n}$, and can be replaced further.
	Thus when $n$ approaches infinity, the convergence probability of the $\epsilon$-tree $(t_{\infty})_{\mathcal{E}}^{A}$ is $0$. By Definition \ref{defn-regdiv}, $(t_{\infty})_{\mathcal{E}}^{A}$ is a regular $\epsilon$-tree of $A$ with regard to $\mathcal{E}$ whose all leaves are $B$.
	\end{proof}
	
	The following lemma shows that the nodes in a $\tau$-EC are all branching bisimilar with explicit divergence.
	\begin{lemma}\label{lem-tauEC}
	Given a $\tau$-EC $\mathsf{ec} = (V_{\mathsf{ec}}, E_{\mathsf{ec}})$, and suppose $A, B \in V_{\mathsf{ec}}$.
	Then $A \simeq^{\Delta} B$.
	\end{lemma}

	\begin{proof}
	Let $\mathcal{E} = V_{\mathsf{ec}} \times V_{\mathsf{ec}}$ and $\mathcal{E}' = (\mathcal{E} ~\cup \simeq^{\Delta})^{*}$.
	We only need to show that $\mathcal{E}'$ is a branching bisimulation with explicit divergence. Since each node $A \in V_{\mathsf{ec}}$ satisfies that $A \Uparrow_{\mathcal{E}}$, it is not hard to see that $\mathcal{E}'$ is divergent $\epsilon$-tree preserving.
	Next we prove that $\mathcal{E}'$ is a  branching bisimulation.
	
	Consider any pair $(A,B) \in \mathcal{E}$.
	The $\ell$-transition $B \rightsquigarrow_{\mathcal{E}}\xrightarrow{\ell} \mathcal{B}$ of $B$ consists of a regular $\epsilon$-tree $t_{\mathcal{E}}^{B}$ of $B$ satisfying that $L \xrightarrow{\ell} L' \in \mathcal{B}$ for every leaf $L$ of $t_{\mathcal{E}}^{B}$. According to Lemma \ref{lem-EcAndTree}, there exists a regular $\epsilon$-tree $t_{\mathcal{E}}^{A}$ of $A$ with regard to $\mathcal{E}$ whose leaves are all $B$.  By replacing all leaves $B$ with $t_{\mathcal{E}}^{B}$ in $t_{\mathcal{E}}^{A}$, we obtain a new $\epsilon$-tree $(t')_{\mathcal{E}}^{A}$ of $A$ satisfying that $L \xrightarrow{\ell} L' \in \mathcal{B}$ for every leaf $L$ of $(t')_{\mathcal{E}}^{A}$.
	We then verify the regularity of $(t')_{\mathcal{E}}^{A}$.
	Given any $\delta \in (0,1)$, since $t_{\mathcal{E}}^{A}$ and $t_{\mathcal{E}}^{B}$ are two regular $\epsilon$-trees, there exists two numbers $M_{\delta}$ and $N_{\delta}$ such that $1 - \mathsf{P}_{M_{\delta}}(t_{\mathcal{E}}^{A}) < \delta/2$ and $1 - \mathsf{P}_{N_{\delta}}(t_{\mathcal{E}}^{A}) < \delta/2$.
	Since the $\epsilon$-tree $(t')_{\mathcal{E}}^{A}$ is obtained from $t_{\mathcal{E}}^{A}$ by replacing all leaves $B$ in $t_{\mathcal{E}}^{A}$ with $t_{\mathcal{E}}^{B}$, we have
	$$
	1 - \mathsf{P}_{M_{\delta}+N_{\delta}}((t')_{\mathcal{E}}^{A})
	< (1 - \mathsf{P}_{M_{\delta}}(t_{\mathcal{E}}^{A})) + \mathsf{P}_{M_{\delta}}(t_{\mathcal{E}}^{A}) \cdot (1 - \mathsf{P}_{N_{\delta}}(t_{\mathcal{E}}^{B}))
	< \delta/2 + 1 \cdot \delta/2 = \delta.
	$$
	Therefore $(t')_{\mathcal{E}}^{A}$ is a regular $\epsilon$-tree.
	Now we see that the $\ell$-transition $B \rightsquigarrow_{\mathcal{E}}\xrightarrow{\ell} \mathcal{B}$ is bisimulated by $A \rightsquigarrow_{\mathcal{E}}\xrightarrow{\ell} \mathcal{B}$. The case for $q$-transition is similar and thus omitted.
	\end{proof}
	
	For any node $A$ in a $\tau$-EC, a Dirac scheduler $\sigma_A$ of $A$ induces a divergent tree of $A$, Lemma \ref{lem-tauEC} further ensures every node in the tree is in the same equivalence class with respect to $\simeq^{\Delta}$, then we have the following.
	\begin{corollary}\label{coro-tauEC}
	Given a $\tau$-EC $\mathsf{ec} = (V_{\mathsf{ec}}, E_{\mathsf{ec}})$, and suppose $A \in V_{\mathsf{ec}}$.
	Then $A \Uparrow_{\simeq^{\Delta}}$.
	\end{corollary}

	\begin{theorem}[$\simeq^{\Delta} \;\subseteq\; \simeq_{e}$]\label{thm-dsbbVSebb}
	The equivalence $\simeq^{\Delta}$ is an exhaustive branching bisimulation.	
	\end{theorem}

	\begin{proof}
	Suppose $(A,B) \in \;\simeq^{\Delta}$ and $A \xRightarrow{} \circlearrowleft_{\mathsf{ec}_1}$. We need to show that there exists some $\mathsf{ec}_2$ such that $B \xRightarrow{} \circlearrowleft_{\mathsf{ec}_2}$ and  $\mathsf{ec}_1 \;(\simeq^{\Delta})^{\ddagger}\; \mathsf{ec}_2$.
	By assumption there exists $A'$ such that $A\xRightarrow{} A'$ and $A'\circlearrowleft_{\mathsf{ec}_1}$. Since $(A,B) \in \;\simeq^{\Delta}$ and $\simeq^{\Delta}$ is a branching bisimulation, we can show that there exists $B'$ such that $B\xRightarrow{} B'$ and $(A',B') \in \;\simeq^{\Delta}$ by induction on the path from $A$ to $A'$.
	Since $A'\circlearrowleft_{\mathsf{ec}_1}$, by Corollary \ref{coro-tauEC}, there exists a divergent $\epsilon$-tree of $A'$ with respect to $\simeq^{\Delta}$.
	Since $(A',B') \in \;\simeq^{\Delta}$, there exists a divergent $\epsilon$-tree $t^{B'}_{\simeq^{\Delta}}$ of $B'$ with respect to $\simeq^{\Delta}$.
	Due to the second property of ECs presented in Theorem 3.2 in \cite{dealfaro_FormalVerificationProbabilistic_1998}, from $B'$ any path in the tree $t^{B'}_{\simeq^{\Delta}}$ will end up with probability one in a $\tau$-EC.
	Arbitrarily choose one of these $\tau$-ECs $\mathsf{ec}_2$,  then there exists $B''$ such that $B'\xRightarrow{} B''$ and $B''\circlearrowleft_{\mathsf{ec}_2}$.
	Now for any $B''' \in S_{\mathsf{ec}_2}$, since $B'''$ is in the $\epsilon$-tree $t^{B'}_{\simeq^{\Delta}}$, we have $B'''\simeq^{\Delta}B'\simeq^{\Delta}A'$.
	Therefore $\mathsf{ec}_1 \;(\simeq^{\Delta})^{\ddagger}\; \mathsf{ec}_2$.
	Putting together the above analysis, we have $B \xRightarrow{} B' \xRightarrow{} B''$ such that $B''\circlearrowleft_{\mathsf{ec}_2}$ and $\mathsf{ec}_1 \;(\simeq^{\Delta})^{\ddagger}\; \mathsf{ec}_2$.
	\end{proof}

	\section{Variations on divergence-sensitive bisimulations}\label{sec-spectrum}
	In this section, we study the relationship between several divergence-sensitive branching and weak bisimilarities. Previously, two such bisimilarities were proposed in \cite{fu_ModelIndependentApproach_2021} and \cite{he_DivergencesensitiveWeakProbabilistic_2023}, respectively.
	A comparative study of these two equivalences has not been carried out so far.
	We will show that probability plus divergence bring extra separation power to the model.
	
	The definition of \emph{exhaustive weak bisimulation} \cite{he_DivergencesensitiveWeakProbabilistic_2023} can be reformulated as follows in our setting.
	\begin{definition}[Exhaustive weak bisimulation]\label{def-ewb}
		Let $\mathcal{E}$ be an equivalence on $\PRCCS$.
		The relation $\mathcal{E}$ is called an \emph{exhaustive weak bisimulation} if $\mathcal{E}$ is a weak bisimulation and is $\tau$-EC invariant.	
		
		We write $A\approx_{e} B$ if $(A,B) \in \mathcal{E}$ for some exhaustive weak bisimulation $\mathcal{E}$.
	\end{definition}

	\begin{theorem}[\cite{he_DivergencesensitiveWeakProbabilistic_2023}]\label{thm-ewb-equivalence}
	The relation $\approx_{e}$ is the largest exhaustive weak bisimulation, and it is an equivalence.
	\end{theorem}

	\begin{remark} \label{remark-approxDelta}
		Refer to Definition \ref{def-dsbb}, a natural definition of \emph{weak bisimulation with explicit divergence} could be:
		Let $\mathcal{E}$ be an equivalence on $\PRCCS$,
		$\mathcal{E}$ is called a \emph{weak bisimulation with explicit divergence} if $\mathcal{E}$ is a weak bisimulation and is divergent $\epsilon$-tree preserving.
		We write $A\approx^{\Delta}B$ if there is a weak bisimulation with explicit divergence $\mathcal{E}$ such that $(A,B) \in \mathcal{E}$.
		
		Although the above definition is straightforward, it is challenging to justify that $\approx^{\Delta}$ is the largest weak bisimulation with explicit divergence.
		So far we are not able to prove that it is an equivalence. The proof for Lemma \ref{lem-dsbb-equivalence} does not apply here, for probabilistic weak bisimulation does not have the stuttering property, which prevents us from constructing $\epsilon$-trees. 
		Neither can we take the strategy in \cite{liu_AnalyzingDivergenceBisimulation_2017} for non-probabilistic weak bisimulation with explicit divergence, as $\approx^{\Delta} \;=\; \approx_{e}$ no longer holds in probabilistic setting (a counterexample will be given in Example \ref{exam-StrictInclusion}).
		We leave the justification of $\approx^{\Delta}$ as an open problem.
	\end{remark}
	
	We give two representative examples to highlight the differences between these bisimilarities.
	For any process $A \in \PRCCS$ and equivalence $\mathcal{E}$ on $ \PRCCS$, let $G_A = (V_A, E_A)$ be the induced transition graph of $A$.
	From now on, we will often abbreviate $V_A \slash \mathcal{E}$  as $\mathcal{E}$ for clarity.
	
	\begin{example}\label{exam-StrictInclusion}
		Let $B_1 = \mu X. (\tau.X + \tau. (\frac{1}{3}\tau.X \oplus \frac{1}{3}\tau.a \oplus \frac{1}{3}\tau.b))$ and $A_1 = \frac{1}{3}\tau.B_1 \oplus \frac{1}{3}\tau.a \oplus \frac{1}{3}\tau.b$.
		The induced transition graph of $A_1$ is depicted in Figure \ref{fig-inclusion-a}.
		Now consider the equivalence $\mathcal{E} = \{\{A_{1},B_{1}\}, \{a\}, \{b\}, \{\mathbf{0}\}\}$. The following facts can be easily checked.
		\begin{enumerate}
			\item $\mathcal{E}$ is the largest weak bisimulation and the largest branching bisimulation. That is $\approx \;=\; \simeq \;= \mathcal{E}$. To see that $\mathcal{E}$ is a branching bisimulation, only note that the $q$-transition $A_{1} \rightsquigarrow_{\mathcal{E}}\xrightarrow{1/2} [a]_{\mathcal{E}}$ of $A_{1}$ (where $\frac{1}{2}=\frac{1}{3}/({1-\frac{1}{3}})$ is the conditional probability of leaving $[A_1]_{\mathcal{E}}$ to $[a]_{\mathcal{E}}$) can be bisimulated by the transition $B_{1} \xrightarrow{\tau} A_{1} \rightsquigarrow_{\mathcal{E}}\xrightarrow{1/2} [a]_{\mathcal{E}}$ of $B_{1}$.
			Here since $(A_1, B_1) \in \mathcal{E}$, by sticking the regular $\epsilon$-tree of $A_1$ to the edge $B_{1} \xrightarrow{\tau} A_{1}$, it then forms a regular $\epsilon$-tree of $B_1$ and then induces a $q$-transition of $B_{1}$.
			\item \label{exam-EbbNotinDsbb-item2}$\mathcal{E}$ is an exhaustive branching bisimulation. We see that both $A_{1}$ and $B_{1}$ can only reach the $\tau$-EC $\mathsf{ec}_{B_{1}} = (B_{1}, \{B_{1} \xrightarrow{\tau} B_{1}\})$ and thus satisfy the divergence requirement.
			\item $\mathcal{E}$ is not a weak bisimulation with explicit divergence, since for the pair $(A_{1},B_{1}) \in \mathcal{E}$, $A_{1} \not\Uparrow_{\mathcal{E}}$ whereas $B_{1} \Uparrow_{\mathcal{E}}$.
		\end{enumerate}
		Since $\simeq_{e} \;\subseteq\; \simeq \;= \mathcal{E}$ and $\mathcal{E}$ is an exhaustive branching bisimulation, we have $\simeq_{e} \;= \mathcal{E}$.
		Together with the fact $\simeq_{e} \;\subseteq\; \approx_{e} \;\subseteq\; \approx$, we derive that $\approx_{e} \;= \mathcal{E}$.
		Since $\approx^{\Delta} \;\subseteq\; \approx \;= \mathcal{E}$ and $\mathcal{E}$ is not a weak bisimulation with explicit divergence, we have $\approx^{\Delta} = \{\{A_{1}\}, \{B_{1}\}, \{a\}, \{b\},\{\mathbf{0}\}\}$. Combining the fact that $\simeq^{\Delta} \;\subseteq\; \approx^{\Delta}$, we obtain $\simeq^{\Delta} = \{\{A_{1}\}, \{B_{1}\}, \{a\}, \{b\}, \{\mathbf{0}\}\}$.
		Now we see that $A_{1}\simeq_{e}B_{1}$ yet $A_{1}\not\simeq^{\Delta}B_{1}$ and $A_{1}\approx_{e}B_{1}$ yet $A_{1}\not\approx^{\Delta}B_{1}$.
	\end{example}
	\begin{example}
		Let $C_2 = \mathbf{0}$, $B_2 = \mu X.(\tau. X + \tau. C_2)$ and $A_2 = \frac{1}{2}\tau.B_2 \oplus \frac{1}{2}\tau.C_2.$
		The induced transition graph of $A_2$ is depicted in Figure \ref{fig-inclusion-b}.
		It is not hard to see that the coarsest relation $\{\{A_{2},B_{2},C_{2}\}\}$ is a branching bisimulation, thus $\approx \;=\; \simeq \;= \{\{A_{2},B_{2},C_{2}\}\}$.
		Since both $A_{2}$ and $B_{2}$ can reach the only $\tau$-EC $\mathsf{ec}_{B_{2}} = (B_{2}, \{B_{2} \xrightarrow{\tau} B_{2}\})$ while $C_{2}$ cannot, any equivalence $\mathcal{E}$ satisfying $(A_{2},C_{2}) \in \mathcal{E}$ or $(B_{2},C_{2}) \in \mathcal{E}$ cannot be an exhaustive weak bisimulation.
		Now let $\mathcal{E}' = \{\{A_{2},B_{2}\}, \{C_{2}\}\}$, then we have $\approx_{e} ~\subseteq \mathcal{E}'$. We can further verify the following facts.
		\begin{enumerate}
			\item $\mathcal{E}'$ is an exhaustive weak bisimulation. It is not hard to see that $\mathcal{E}'$ satisfies the divergence condition with respect to $\tau$-EC. Then we only need to show that $\mathcal{E}'$ is a weak bisimulation. Now consider the following two transitions $\mathsf{tr}_{A_{2}}$ and $\mathsf{tr}_{B_{2}}$ for pair $(A_{2},B_{2}) \in \mathcal{E}'$.
			\begin{itemize}
				\item $\mathsf{tr}_{A_{2}} = A_{2} \xrightarrow{\tau} \{(B_{2}: \frac{1}{2}), (C_{2}: \frac{1}{2})\}$. Then the matched weak combined transition $B_{2} \stackrel{\tau}{\Longrightarrow}_{c} \{(B_{2}: \frac{1}{2}), (C_{2}: \frac{1}{2})\}$ for $B_{2}$ is induced by the following scheduler:\\
				$
				\sigma_{B_{2}} = \begin{cases}
					\{(B_{2}\xrightarrow{\tau}\delta_{B_{2}}: \frac{1}{2}), (B_{2}\xrightarrow{\tau}\delta_{C_{2}}:\frac{1}{2})\},  & \text{if $\omega = B_{2}$,} \\
					\delta_{\bot}, & \text{otherwise.}
				\end{cases}
				$
				\item $\mathsf{tr}_{B_{2}} = B_{2} \xrightarrow{\tau} \delta_{C_{2}}$. Then the matched weak combined transition $A_{2} \stackrel{\tau}{\Longrightarrow}_{c}  \delta_{C_{2}}$ for $A_{2}$ is induced by the following scheduler:\\
				$
				\sigma_{A_{2}} = \begin{cases}
					\delta_{\mathsf{tr}_{A_{2}} },  & \text{if $\omega = A_{2}$,} \\
					\delta_{\mathsf{tr}_{B_{2}} },  & \text{if $\omega = A_{2}\tau B_{2}$,} \\
					\delta_{\bot}, & \text{otherwise.}
				\end{cases}
				$
			\end{itemize}
			\item $\mathcal{E}'$ is not an exhaustive branching bisimulation. In fact, we will show that it is not a branching bisimulation. Since the $\ell$-transition $B_{2} \xrightarrow{\tau} [C_{2}]_{\mathcal{E}'}$ for $B_{2}$ cannot be bisimulated by $A_{2}$ (for $A_{2}$ cannot perform any $\ell$-transition), we see that the pair $(A_{2},B_{2}) \in \mathcal{E}'$ violates the branching bisimulation conditions.
		\end{enumerate}
		Since $\approx_{e} \;\subseteq \mathcal{E}'$ and $\mathcal{E}'$ is an exhaustive weak bisimulation, we have $\approx_{e} \;= \mathcal{E}'$.
		Since $\simeq_{e} \;\subseteq\; \approx_{e} $ and $\mathcal{E}'$ is not an exhaustive branching bisimulation, we have $\simeq_{e} \;= \{\{A_{2}\},\{B_{2}\},\{C_{2}\}\}$.
		Since $\simeq^{\Delta} \;\subseteq\; \simeq_{e}$, we have $\simeq^{\Delta}  \;= \{\{A_{2}\},\{B_{2}\},\{C_{2}\}\}$.
		Now we see that although $A_{2} \approx B_{2}$ and $A_{2} \simeq B_{2}$ hold, $A_{2}\approx_{e}B_{2}$ but $A_{2}\not\simeq_{e} B_{2}$.
	\end{example}
	
	\begin{figure}[htb]
		\captionsetup{justification=centering}
		\begin{subfigure}[t]{0.48\textwidth}
			\centering
			\begin{tikzpicture}[on grid,scale=0.8]
				
				\begin{scope}[shift={(0,-5.5)}, every node/.style={draw,circle,transform shape,inner sep=3pt}]
					\foreach \pos/\name/\text in {
						{(0-0.5,1.5)/B/{$B_1$}},
						{(1,3)/A/{$A_1$}},
						{(1,1.5)/a/{$a$}},
						{(2+0.5,1.5)/b/{$b$}},
						{(1.75,0.75)/z/{$\mathbf{0}$}}}
					\node (\name) at \pos {\text};
				\end{scope}
				
				\begin{scope}[every node/.style={transform shape, auto=left,inner sep=1pt}, >={Latex[width'=0pt .5, length=5pt]}]
					\foreach \source/ \dest / \text in {
						b/z/{$b$}, A/a/{$\frac{1}{3}\tau$}, A/b/{$\frac{1}{3}\tau$}}
					\path[->] (\source) edge node {\text} (\dest);
					
					\path[->] (B) edge [out=-135, in=-45, looseness=6] node {$\tau$} (B);
					\path[->] (B) edge [bend left=45] node {$\tau$} (A);
				\end{scope}	
				
				\begin{scope}[every node/.style={transform shape, auto=right,inner sep=1pt}, >={Latex[width'=0pt .5, length=5pt]}]
					\foreach \source/ \dest / \text in {
						a/z/{$a$}, A/B/{$\frac{1}{3}\tau$}}
					\path[->] (\source) edge node {\text} (\dest);
					
				\end{scope}	
				
			\end{tikzpicture}
			\vspace{-5mm}
			\caption{
				\setlength{\tabcolsep}{1pt}
				\begin{tabular}{rl}
					&\\
					$B_1$ & $= \mu X. (\tau.X + \tau. (\frac{1}{3}\tau.X \oplus \frac{1}{3}\tau.a \oplus \frac{1}{3}\tau.b)),$\\
					$A_1$ & $= \frac{1}{3}\tau.B_1 \oplus \frac{1}{3}\tau.a \oplus \frac{1}{3}\tau.b.$
			\end{tabular}}\label{fig-inclusion-a}
		\end{subfigure}
		\hspace{15mm}
		\begin{subfigure}[t]{0.38\textwidth}
			\centering
			\begin{tikzpicture}[on grid,scale=0.8]
				
				\begin{scope}[shift={(5+2,-5.5)}, every node/.style={draw,circle,transform shape,inner sep=3pt}]
					\foreach \pos/\name/\text in {
						{(0,1.5)/B/{$B_2$}},
						{(1,3)/A/{$A_2$}},
						{(2,1.5)/C/{$C_2$}}}
					\node (\name) at \pos {\text};
				\end{scope}
				
				\begin{scope}[every node/.style={transform shape, auto=left,inner sep=1pt}, >={Latex[width'=0pt .5, length=5pt]}]
					\foreach \source/ \dest / \text in {
						A/C/{$\frac{1}{2}\tau$}, B/C/{$\tau$}}
					\path[->] (\source) edge node {\text} (\dest);
					
				\end{scope}	
				
				\begin{scope}[every node/.style={transform shape, auto=right,inner sep=1pt}, >={Latex[width'=0pt .5, length=5pt]}]
					\foreach \source/ \dest / \text in {
						A/B/{$\frac{1}{2}\tau$}}
					\path[->] (\source) edge node {\text} (\dest);
					
					\path[->] (B) edge [out=-135, in=135, looseness=6] node {$\tau$} (B);
				\end{scope}	
			\end{tikzpicture}
			\vspace{-5mm}
			\caption{
				\setlength{\tabcolsep}{1pt}
				\begin{tabular}{rl}
					&\\
					$ C_2$ & $= \mathbf{0}\;, B_2 = \mu X.(\tau. X + \tau. C_2),$\\
					$ A_2$ & $= \frac{1}{2}\tau.B_2 \oplus \frac{1}{2}\tau.C_2.$
			\end{tabular}}\label{fig-inclusion-b}
		\end{subfigure}%
		\caption{Counterexamples of the inclusion relationship.}\label{fig-inclusion}	
	\end{figure}

	Lattice among variants of branching and weak bisimilarities for $\mathrm{RCCS}_{fs}$ model can be summarized by the following theorem. A more visual presentation of the theorem is given in Figure~\ref{Fig-spectrum}.

	\begin{theorem}\label{thm-spectrum}
		The relationship between $\simeq^{\Delta}, \simeq_{e}, \simeq,  \approx_{e}$ and $\approx$ is summarized as follows.
		\begin{enumerate}
			\item $\simeq^{\Delta} \;\subsetneq\; \simeq_{e} \;\subsetneq\; \simeq$ and $ \approx_{e}\;\subsetneq\; \approx$;\label{thm-inclusion-item1}
			\item $\simeq \;\subsetneq\; \approx$, $\simeq_{e} \;\subsetneq\; \approx_{e}$ and $\approx_{e} \;\not\subseteq\;  \simeq \;\not\subseteq\; \approx_{e}$. \label{thm-inclusion-item2}
		\end{enumerate}
	\end{theorem}

	\begin{proof}
		(\ref{thm-inclusion-item1}) $\simeq^{\Delta} \;\subseteq\; \simeq_{e}$ follows from Theorem \ref{thm-dsbbVSebb} while $\simeq_{e} \;\subseteq\; \simeq$ and $\approx_{e} \;\subseteq\; \approx$ are by definition.
		The strictness is witnessed by the pair $(A_1,B_1)$ given in Figure~\ref{fig-inclusion-a} and $(B_2, C_2)$ given in Figure~\ref{fig-inclusion-b}.
		
		(\ref{thm-inclusion-item2}) By Theorem \ref{thm-BbInWb}, we have $\simeq \;\subseteq\; \approx$. We conclude that $\simeq_{e} \;\subseteq\; \approx_{e}$ by noticing that the definition of $\tau$-EC invariant is independent of the requirement of bisimulation. 
		The processes $A_3 = \tau.a + a + b$ and $B_3 = \tau.a + b$ witness the strictness of the subset relations.
		For one thing, we can show that $\mathcal{E} = \{\{A_3, B_3\}, \{a\}, \{b\}, \{\mathbf{0}\}\}$ is an exhaustive weak bisimulation, which implies that $A_3 \approx_e B_3$. For another thing, since $A_3 \not\simeq \mathbf{0}$, the $\ell$-transition $A_3  \rightsquigarrow_{\simeq}\xrightarrow{a} [\mathbf{0}]_{\simeq}$ cannot be bisimulated by any $\ell$-transition of $B_3$.  Thus $A_3 \not\simeq B_3$. The inclusions $\simeq_{e} \;\subseteq\; \approx_{e}$ and $\simeq \;\subseteq\; \approx$ are strict since $A_3 \approx_e B_3$ yet $A_3  \not\simeq _{e} B_3$, and $A_3 \approx B_3$ yet $A_3  \not\simeq B_3$.
		The pair $(A_2, B_2)$ in Figure \ref{fig-inclusion-b} is also a non-trivial example for $\simeq_{e} \;\subsetneq\; \approx_{e}$, and $(A_3, B_3)$ is also an evidence for the strictness of $\approx_{e} \;\not\subseteq\;  \simeq$.  Finally the pair $(B_2, C_2)$ in Figure \ref{fig-inclusion-b} shows that $\simeq \;\not\subseteq\; \approx_e$.
	\end{proof}

	We end this part by summarizing the results in Figure \ref{Fig-spectrum}, where the arrow from one bisimilarity to the other means that the former bisimilarity is strictly finer than the latter one.
	Solid arrows are new results of this paper while the dotted arrow is a result from \cite{he_DivergencesensitiveWeakProbabilistic_2023}.
	
	\begin{figure}[htb]
		\vspace{-3mm}	
		\[
		\xymatrix @R=0.618in @C=1in{
			\simeq^{\Delta}\ar @{->}[r]^{\subsetneq}_{\text{Theorem \ref{thm-dsbbVSebb}}}
			& {\simeq_e} \ar @{->}[r]^{\subsetneq} \ar @{->}[d]^{\subsetneq}
			&\simeq\ar @{->}[d]^{\subsetneq}_{\text{Theorem \ref{thm-BbInWb}}}\\
			& \approx_e \ar @{.>}[r]^{\subsetneq}
			&\approx\\
		}
		\]
		\vspace{-3mm}
		\caption{Divergence-sensitive bisimulation lattice (Theorem \ref{thm-spectrum}).}\label{Fig-spectrum}  \vspace{-2mm}
	\end{figure}	

	\section{Efficient equivalence checking algorithms}\label{sec-algorithm}
	In this section, we provide polynomial time verification algorithms for all notions of bisimilarity presented in this work. Particularly for $\approx_{e}$, we improve known results for divergence-sensitive weak bisimilarity in \cite{he_DivergencesensitiveWeakProbabilistic_2023} by giving a more direct algorithm  based on maximal end components. An overview of the algorithmic results is given in Table \ref{table-allAlg}.

	\subsection{Algorithm for deciding branching bisimilarity with explicit divergence}\label{sec-alg4dsbb}
	We recall an algorithmic result in \cite{zhang_UniformRandomProcess_2019}, which says that the largest (divergence-insensitive) branching bisimulation is  efficiently computable.
	\begin{theorem}[\cite{zhang_UniformRandomProcess_2019}]
		\label{thm-bbAlg}
		Given two processes $A,B \in \PRCCS$. Let $S$ be the set of processes reachable from $A$ and $B$, and $N = |S|$ be the size of $S$.
		For any equivalence $\mathcal{E}$ on $S$, the largest branching bisimulation $\mathcal{E}'$ contained in $\mathcal{E}$ can be computed by a procedure $\mathsf{Quotient}(\mathcal{E})$ in polynomial time of $N$.
	\end{theorem}
	
	Given a process $A$ and an equivalence $\mathcal{E}$, the number of maximal $\epsilon$-trees of $A$ with regard to $\mathcal{E}$ can be exponentially many.
	However, the existence of a divergent $\epsilon$-tree can be checked in polynomial time by the procedure $\mathsf{DetDivTree}$ given in Algorithm~\ref{DetDivTree}.
	In what follows, we will use $\backslash$ for set difference, and $\slash$ for relation quotient.
	The \emph{$\epsilon$-graph of $A$ with regard to $\mathcal{E}$}, denoted by $G_A^{\mathcal{E}} = (V_A^{\mathcal{E}}, E_A^{\mathcal{E}})$, is a subgraph of $G_A$ (where $G_A$ is the induced transition graph of $A$) satisfying that $V_A^{\mathcal{E}}$ contains all processes reachable from $A$ by state-preserving immediate silent transitions and $E_A^{\mathcal{E}}$ contains all the corresponding transition edges.
	A node in $G_A^{\mathcal{E}}$ is called a \emph{sink node} if its out degree is $0$.
	
	Intuitively speaking, the procedure $\mathsf{DetDivTree}(A, \mathcal{E})$ starts with the set of sink nodes in $\epsilon$-graph $G_A^{\mathcal{E}} = (V_A^{\mathcal{E}}, E_A^{\mathcal{E}})$, then iteratively constructs the set $\{A' \in V_A^{\mathcal{E}}\mid A'\not\Uparrow_{\mathcal{E}}\}$.
	For a better understanding of  $\mathsf{DetDivTree}$, we use an example to explain how it works.
	
		\SetKwRepeat{Do}{do}{while}
	\begin{algorithm2e}[H]
		\footnotesize
		\DontPrintSemicolon
		\caption{$\mathsf{DetDivTree}$
			{\footnotesize \CommentSty{/* checking the existence of a divergent $\epsilon$-tree of $A$ with regard to $\mathcal{E}$  */}}
		}\label{DetDivTree}
		\SetKwInOut{KwIn}{Input}
		\SetKwInOut{KwOut}{Output}
		
		\KwIn{$A, \mathcal{E}$}
		\KwOut{$isDiv\in\{\mathbf{T},\mathbf{F}\}$}
		
		$(V, E) \gets \mathsf{CompEpsGraph}(A, \mathcal{E})$
		{\footnotesize \CommentSty{/* Computes the $\epsilon$-graph $ G_A^{\mathcal{E}} = (V_A^{\mathcal{E}}, E_A^{\mathcal{E}})$ */}}
		
		$\mathcal{L}_{ndiv} \gets \mathsf{Sink}((V, E))$,
		$\mathcal{L}_{und} \gets V \backslash \mathcal{L}_{ndiv}$
		{\footnotesize \CommentSty{/* $\mathsf{Sink}$ returns the sink nodes in $(V, E)$ */}}
		
		\Do{$toCon = \mathbf{T}$}{
			$toCon \gets \mathbf{F}$
			
			\For{$B \in \mathcal{L}_{und}$}{
				
				$nonDiv \gets \mathbf{T}$
				
				\For{$(B,B') \in E$ with label $\tau$}{
					\If{$B' \in \mathcal{L}_{und}$}{
						$nonDiv \gets \mathbf{F}$
					}
				}
				
				\For{$(B,B') \in E$ with label $p\tau$}{
					\If{all $B \xrightarrow{q\tau} B''$ satisfying that $B'' \in \mathcal{L}_{und}$}{
						$nonDiv \gets \mathbf{F}$
					}
				}
				
				\If{$nonDiv = \mathbf{T}$}{
					$toCon \gets \mathbf{T}$,
					$\mathcal{L}_{ndiv} \gets \mathcal{L}_{ndiv} \cup \{B\}$,
					$\mathcal{L}_{und} \gets \mathcal{L}_{und} \backslash \{B\}$
				}
			}
		}
		
		\eIf{$A \in \mathcal{L}_{ndiv}$}{$isDiv \gets \mathbf{F}$}{
			$isDiv \gets \mathbf{T}$
		}
		\KwRet{$isDiv$}
	\end{algorithm2e}
	
	\begin{example}
		Let $A = \mu X.(\frac{1}{2}\tau.(\tau.X+\tau.\mathbf{0}) \oplus \frac{1}{2}\tau.(\frac{1}{2}\tau.\mathbf{0} \oplus \frac{1}{2}\tau.(\mu Y.\tau.Y)))$, $B = \tau.A + \tau.\mathbf{0}$ and $\simeq$ be branching bisimilarity.
		The $\epsilon$-graph $G_B^{\simeq} = (S_B^{\simeq}, T_B^{\simeq})$ of $B$ with regard to $\simeq$ is shown in Figure \ref{fig-CompEpsGraph-a}, where $s_0 = B$, $s_1 = A$, $s_2 = \frac{1}{2}\tau.\mathbf{0} \oplus \frac{1}{2}\tau.(\mu Y.\tau.Y)$, $s_3 = \mathbf{0}$ and $s_4 = \mu Y.\tau.Y$.
		
	Procedure $\mathsf{DetDivTree}(B, \simeq)$ works as follows.
		\begin{enumerate}
			\item Procedure $\mathsf{CompEpsGraph}(B, \simeq)$ computes the $\epsilon$-graph $G_B^{\simeq} = (S_B^{\simeq}, T_B^{\simeq})$ and $\mathsf{Sink}((S_B^{\simeq}, T_B^{\simeq}))$ returns the set of sink nodes in $G_B^{\simeq}$.
			Thus $\mathcal{L}_{ndiv}^0 = \{s_3\}$ and $\mathcal{L}_{und}^0 = \{s_0, s_1, s_2, s_4\}$.
			\item In the first iteration of the $\textbf{do--while}$ loop, we add all processes $B' \in \mathcal{L}_{und}^0$ satisfying that $tgt(\mathsf{itr}) \cap \mathcal{L}_{ndiv}^0 \ne \emptyset$ for all immediate silent transitions $\mathsf{itr}$ of $B'$ into the set $\mathcal{L}_{ndiv}^1$. Then we have
			\begin{itemize}
				\item $\mathcal{L}_{ndiv}^1 = \mathcal{L}_{ndiv}^0 \cup \{s_2\}, \mathcal{L}_{und}^1 = \{s_0, s_1, s_4\}$ and $toCon = \mathbf{T}$.
			\end{itemize}
			\item Similarly, in the second and third iterations of the $\textbf{do--while}$ loop, we have
			\begin{itemize}
				\item $\mathcal{L}_{ndiv}^2 = \mathcal{L}_{ndiv}^1 \cup \{s_1\}, \mathcal{L}_{und}^2 = \{s_0, s_4\}$ and $toCon = \mathbf{T}$.
				\item $\mathcal{L}_{ndiv}^3 = \mathcal{L}_{ndiv}^2 \cup \{s_0\}, \mathcal{L}_{und}^3 = \{s_4\}$ and $toCon = \mathbf{T}$.
			\end{itemize}
			\item In the fourth iteration of the $\textbf{do--while}$ loop, there does not exist any process $B' \in \mathcal{L}_{und}^3$ satisfying that $tgt(\mathsf{itr}) \cap \mathcal{L}_{ndiv}^3 \ne \emptyset$ for all collective transitions $\mathsf{itr}$ of $B'$. Then the loop terminates, and we have
			\begin{itemize}
				\item $\mathcal{L}_{ndiv}^4 = \mathcal{L}_{ndiv}^3, \mathcal{L}_{und}^4 = \{s_4\}$ and $toCon = \mathbf{F}$.
			\end{itemize}
			\item For the final $\mathcal{L}_{ndiv}^4 = \{s_3, s_2, s_1, s_0\}$, we depict the result in Figure \ref{fig-CompEpsGraph-b}, where the four nodes in red are  non-divergent, only the one in blue is divergent as it has a divergent $\epsilon$-tree $t^{s_4}_{\simeq}$ with respect to $\simeq$. As $B = s_0 \in \mathcal{L}_{ndiv}^4$, we have $B \not\Uparrow_{\simeq}$.
		\end{enumerate}
	\end{example}


				\begin{figure}[htb]	
					\begin{subfigure}[t]{0.45\textwidth}
						\centering
						\begin{tikzpicture}[on grid,scale=0.9]
							
							\begin{scope} [every node/.style={draw,circle,transform shape,inner sep=3pt}]
								\foreach \pos/\name/\text in {
									{(2.5,2.5)/s0/{$s_0$}},
									{(4,4)/s1/{$s_1$}},
									{(5.5,2.5)/s2/{$s_2$}},
									{(4,1)/s3/{$s_3$}},
									{(7,1)/s4/{$s_4$}}}
								\node (\name) at \pos {\text};
							\end{scope}
							
							\begin{scope}[every node/.style={transform shape,auto=right,inner sep=0pt}, >={Latex[width'=0pt .5, length=5pt]}]
								\foreach \source/ \dest / \text in {
									s1/s0/{$\frac{1}{2}\tau$},
									s2/s3/{$\frac{1}{2}\tau$}}
								\path[->] (\source) edge node {\text} (\dest);			
							\end{scope}		
							\begin{scope}[every node/.style={transform shape,auto=left,inner sep=0pt}, >={Latex[width'=0pt .5, length=5pt]}]
								\foreach \source/ \dest / \text in {
									s1/s2/{$\frac{1}{2}\tau$}, s0/s3/{$\tau$},
									s2/s4/{$\frac{1}{2}\tau$}}
								\path[->] (\source) edge node {\text} (\dest);
								
								\path[->] (s0) edge [bend left=45] node {$\tau$} (s1);
								
								\path[->] (s4) edge [out=45, in=-45, looseness=8] node {$\tau$} (s4);
							\end{scope}			
						\end{tikzpicture}
						\vspace{-5mm}
						\caption{Before $\mathsf{DetDivTree}(B, \simeq)$.}\label{fig-CompEpsGraph-a}
					\end{subfigure}
					\begin{subfigure}[t]{0.45\textwidth}
						\centering
						\begin{tikzpicture}[on grid,scale=0.9]
							\begin{scope} [shift={(0,-6)}, every node/.style={draw,circle,fill = red!20,transform shape,inner sep=3pt}]
								\foreach \pos/\name/\text in {
									{(2.5,2.5)/s0/{$s_0$}},
									{(4,4)/s1/{$s_1$}},
									{(5.5,2.5)/s2/{$s_2$}},
									{(4,1)/s3/{$s_3$}}}
								\node (\name) at \pos {\text};
							\end{scope}
							
							\begin{scope} [shift={(0,-6)}, every node/.style={draw, fill = blue!20, circle,transform shape,inner sep=3pt}]
								\foreach \pos/\name/\text in {
									{(7,1)/s4/{$s_4$}}}
								\node (\name) at \pos {\text};
							\end{scope}
							
							\begin{scope}[every node/.style={font=\scriptsize, text = red, align=center}]
								\node at ($(s3.0) + (1/2,0)$) {
									$\mathcal{L}_{ndiv}^0$};
								\node at ($(s2.0) + (1/2,0)$) {
									$\mathcal{L}_{ndiv}^1$};
								\node at ($(s1.0) + (1/2,0)$) {
									$\mathcal{L}_{ndiv}^2$};
								\node at ($(s0.0) + (1/2,0)$) {
									$\mathcal{L}_{ndiv}^3$};
							\end{scope}	
							
							\begin{scope}[every node/.style={font=\scriptsize, text = blue, align=center}]
								\node at ($(s4.180) + (-1/2,0)$) {
									$\mathcal{L}_{und}^4$};
							\end{scope}	
							
							\begin{scope}[every node/.style={transform shape,auto=right,inner sep=0pt}, >={Latex[width'=0pt .5, length=5pt]}]
								\foreach \source/ \dest / \text in {
									s1/s0/{$\frac{1}{2}\tau$},
									s2/s3/{$\frac{1}{2}\tau$}}
								\path[->] (\source) edge node {\text} (\dest);			
							\end{scope}		
							\begin{scope}[every node/.style={transform shape,auto=left,inner sep=0pt}, >={Latex[width'=0pt .5, length=5pt]}]
								\foreach \source/ \dest / \text in {
									s1/s2/{$\frac{1}{2}\tau$}, s0/s3/{$\tau$},
									s2/s4/{$\frac{1}{2}\tau$}}
								\path[->] (\source) edge node {\text} (\dest);
								
								\path[->] (s0) edge [bend left=45] node {$\tau$} (s1);
								
								\path[->] (s4) edge [out=45, in=-45, looseness=8] node {$\tau$} (s4);
							\end{scope}		
						\end{tikzpicture}
						\vspace{-5mm}
						\caption{After $\mathsf{DetDivTree}(B, \simeq)$.}\label{fig-CompEpsGraph-b}
					\end{subfigure}
					\caption{The procedure of $\mathsf{DetDivTree}(B, \simeq)$.}
					\label{fig-CompEpsGraph}
				\end{figure}
				
				The correctness of $\mathsf{DetDivTree}$ is proven in the following proposition.
				\begin{proposition}
					Given an equivalence $\mathcal{E}$ on $\PRCCS$ and a process $A \in \PRCCS$.
					Then $A \not\Uparrow_{\mathcal{E}}$ if and only if the procedure $\mathsf{DetDivTree}(A,\mathcal{E})$ returns $\mathbf{F}$.
				\end{proposition}
				
				\begin{proof}
					Let $\mathcal{L}_{ndiv}^i$ be the set $\mathcal{L}_{ndiv}$ at the end of $i$-th iteration of the $\textbf{do--while}$ loop.
					The $\textbf{do--while}$ loop always terminates, as the $(i+1)$-th iteration proceeds iff in $i$-th iteration the set $\mathcal{L}_{ndiv}^i$ gets strictly larger, while it is always true that $\mathcal{L}_{ndiv}^i \subseteq V_A^{\mathcal{E}}$. Let $n$ be the number of iterations of the $\textbf{do--while}$ loop.
					For correctness, it will be sufficient to show that $A \not\Uparrow_{\mathcal{E}}$ iff $A \in \mathcal{L}_{ndiv}^n$.
					
					$(\impliedby)$ We prove that $\forall B\in \mathcal{L}_{ndiv}^i: B\not\Uparrow_{\mathcal{E}}$ holds for all $0 \le i \le n$ by induction on $i$.
					\begin{itemize}
						\item  \emph{(Base case).}
						$\forall B\in \mathcal{L}_{ndiv}^0: B\not\Uparrow_{\mathcal{E}}$ holds trivially for $\mathcal{L}_{ndiv}^0 = \mathsf{Sink}((V_A^{\mathcal{E}}, E_A^{\mathcal{E}}))$ is just the set of nodes that cannot perform any silent action.
						
						\item \emph{(Induction step).} Assume that $\forall B\in \mathcal{L}_{ndiv}^i: B\not\Uparrow_{\mathcal{E}}$. We need to show that $\forall B\in \mathcal{L}_{ndiv}^{i+1}: B\not\Uparrow_{\mathcal{E}}$.		
						
						To $\mathcal{L}_{ndiv}^{i+1}$,  the case  $B\in \mathcal{L}_{ndiv}^{i}$ holds by induction. For any
						$B\in \mathcal{L}_{ndiv}^{i+1}\backslash \mathcal{L}_{ndiv}^{i}$,
						according to our algorithm, there could be two cases:
						if $(B, B'') \in E_A^{\mathcal{E}}$ with label $\tau$ for some $B''$, then $B'' \in \mathcal{L}_{ndiv}^i$;
						if $(B, B') \in E_A^{\mathcal{E}}$ with label $p\tau$ for some $B'$, then there exists $B \xrightarrow{q\tau} B''$ with $B'' \in \mathcal{L}_{ndiv}^i$.
						Now any $\epsilon$-tree $t_B^{\mathcal{E}}$ of $B$ with regard to $\mathcal{E}$ will go through a process $B'' \in \mathcal{L}_{ndiv}^i$.
						By inductive hypothesis, $B''\not\Uparrow_{\mathcal{E}}$ holds, from which it follows that there does not exist any divergent $\epsilon$ tree of $B$. Thus $B\not\Uparrow_{\mathcal{E}}$ for all $B \in \mathcal{L}_{ndiv}^{i+1}$.
					\end{itemize}
					
					$(\implies)$  We prove this direction by contradiction. Suppose there exists some $A$ such that
					$A \not\Uparrow_{\mathcal{E}}$ and $A \notin \mathcal{L}_{ndiv}^n$.
					Since $A \not\Uparrow_{\mathcal{E}}$, any maximal $\epsilon$-tree $t$ of $A$ must have some intermediate nodes $A' \notin \mathcal{L}_{ndiv}^n$ satisfying that some children $A''$ of $A'$ in the tree belong to the set $\mathcal{L}_{ndiv}^n$, for otherwise the tree would be divergent.
					Since $A' \notin \mathcal{L}_{ndiv}^n$, there exists some $\mathsf{itr}_{A'}$ of $A'$ such that $tgt(\mathsf{itr}_{A'}) \subseteq V_A^{\mathcal{E}} \backslash \mathcal{L}_{ndiv}^{n}$.
					Replacing all such transitions from $A'$ to $A''$ by the immediate silent transition $\mathsf{itr}_{A'}$ in tree $t$, we can obtain a divergent $\epsilon$-tree $t'$ of $A$ (since all processes in $V_A^{\mathcal{E}} \backslash \mathcal{L}_{ndiv}^{n}$ can perform state-preserving internal actions), which contradicts the assumption that $A \not\Uparrow_{\mathcal{E}}$.
				\end{proof}
				
				Algorithm \ref{algo-DivBranBisim}  gives the main algorithm for deciding whether two processes are branching bisimilar with explicit divergence.
				Here we follow the classical \emph{partition-refinement} framework \cite{paige_ThreePartitionRefinement_1987,kanellakis_CCSExpressionsFinite_1990}.
				The procedure  $\mathsf{DivBranBisim}(A,B)$ initializes set $R$ as the disjoint union of processes reachable from $A$ and $B$.
				Then it iteratively constructs the set  $\mathcal{E} = R\slash\simeq^{\Delta}$ (i.e., the set of equivalence classes of $R$ under $\simeq^{\Delta}$), starting with the  coarsest partition $\mathcal{E}_{ini} = \{R\}$ and refining it
				until the refined partition satisfies the definition of branching bisimulation with explicit divergence.
				
				At the beginning of each iteration, the procedure $\mathsf{Quotient}(\mathcal{E}_{ini})$ in Theorem \ref{thm-bbAlg} is invoked to extract the largest branching bisimulation $\mathcal{E}$ contained in $\mathcal{E}_{ini}$.
				Procedure $\mathsf{FindDivSplit}(\mathcal{E})$ (given as Algorithm \ref{algo-FindDivSplit}) then checks whether there is a pair of processes $(P,Q) \in \mathcal{E}$ that violates the divergent $\epsilon$-tree preserving condition, i.e., $P \Uparrow_{\mathcal{E}}$ and $Q \not\Uparrow_{\mathcal{E}}$, or $P \not\Uparrow_{\mathcal{E}}$ and $Q \Uparrow_{\mathcal{E}}$. If there is, the discriminating information, i.e., $P$ (also called \emph{divergence splitter}), is returned.
				Procedure $\mathsf{DivRefine}$ (given as Algorithm \ref{algo-DivRefine}) then splits the equivalence class $[P]_{\mathcal{E}}$ into two new equivalence classes $\mathcal{C}_{div}$ and $\mathcal{C}_{ndiv}$ according to the splitter $P$ identified by $\mathsf{FindDivSplit}$. More specifically,  $\mathcal{C}_{div}$ contains all processes $P' \in [P]_{\mathcal{E}}$ satisfying $P' \Uparrow_{\mathcal{E}}$, while $\mathcal{C}_{ndiv}$ contains all processes $P'' \in [P]_{\mathcal{E}}$ satisfying $P'' \not\Uparrow_{\mathcal{E}}$.
				When the iteration terminates, the resulting partition $\mathcal{E}$ is $R\backslash\simeq^{\Delta}$.
				Then checking whether $A \simeq^{\Delta}  B$ is equivalent to checking whether $(A, B) \in \mathcal{E}$.
				
				\begin{algorithm2e}[htb]
					\footnotesize
					\DontPrintSemicolon
					\caption{$\mathsf{DivBranBisim}$
						{\footnotesize \CommentSty{/* checking whether $A \simeq^{\Delta} B$ */}}
					}\label{algo-DivBranBisim}
					\SetKwInOut{KwIn}{Input}
					\SetKwInOut{KwOut}{Output}
					
					\KwIn{$A,B$}
					\KwOut{$b\in\{\mathbf{T},\mathbf{F}\}$}
					
					$R \gets \mathsf{Reach}(A) \uplus  \mathsf{Reach}(B)$
					{\footnotesize \CommentSty{/* $\mathsf{Reach}(P)$ returns the set of processes reachable from $P$ */}}
					
					$\mathcal{E}_{ini} \gets \{R\}, toCon \gets \mathbf{T}$
					
					\Do{$toCon = \mathbf{T}$}{
						$\mathcal{E} \gets \mathsf{Quotient}(\mathcal{E}_{ini})$
						
						{\footnotesize \CommentSty{/* $\mathsf{Quotient}(\mathcal{E}_{ini})$ computes the largest branching bisimulation contained in $\mathcal{E}_{ini}$ */}}
						
						$(divSen, P)\gets \mathsf{FindDivSplit}(\mathcal{E})$
						
						{\footnotesize \CommentSty{/* $\mathsf{FindDivSplit}(\mathcal{E})$ checks whether there is a divergence splitter $P$ of $\mathcal{E}$ */}}
						
						\eIf{$divSen = \mathbf{T}$}{
							$toCon \gets \mathbf{F}$
						}{
							$\mathcal{E}_{ini} \gets \mathsf{DivRefine}(\mathcal{E}, P)$
							
							{\footnotesize \CommentSty{/* $\mathsf{DivRefine}(\mathcal{E}, P)$ refines $\mathcal{E}$ according to the splitter $P$ identified by $\mathsf{FindDivSplit}(\mathcal{E})$ */}}
						}
					}
					{\footnotesize \CommentSty{/* when the \textbf{do-while} loop terminates, $\mathcal{E} = R\slash\simeq^{\Delta}$ */}}
					
					\eIf{$(A,B) \in \mathcal{E}$}{
						\KwRet{$\mathbf{T}$}
					}{\KwRet{$\mathbf{F}$}}
				\end{algorithm2e}
				
				\begin{figure*}[htb]
					\noindent\begin{minipage}[t]{0.5\textwidth}
						\vspace{0pt}
						\begin{algorithm2e}[H]
							\footnotesize
							\DontPrintSemicolon
							\caption{$\mathsf{FindDivSplit}$
							}
							\label{algo-FindDivSplit}
							\SetKwInOut{KwIn}{Input}
							\SetKwInOut{KwOut}{Output}
							
							\KwIn{$\mathcal{E}$}
							\KwOut{$(divSen,P)\in\{(\mathbf{T},\bot), (\mathbf{F}, P)\}$}
							
							$divSen \gets \mathbf{T}$
							
							\For{$(P,Q) \in \mathcal{E}$}{
								$isDivP \gets \mathsf{DetDivTree}(P, \mathcal{E})$
								
								$isDivQ \gets \mathsf{DetDivTree}(Q, \mathcal{E})$
								
								\If{$isDivP \ne isDivQ$}{
									$divSen \gets \mathbf{F}$
									
									\KwRet{$(divSen, P)$}
								}
							}
							
							\KwRet{$(divSen,\bot)$}
							\vspace*{-0.6pt}
							\BlankLine\BlankLine\BlankLine
						\end{algorithm2e}
					\end{minipage}%
					\begin{minipage}[t]{0.5\textwidth}
						\vspace{0pt}
						\begin{algorithm2e}[H]
							\footnotesize
							\DontPrintSemicolon
							\caption{$\mathsf{DivRefine}$
							}
							\label{algo-DivRefine}
							\SetKwInOut{KwIn}{Input}
							\SetKwInOut{KwOut}{Output}
							
							\KwIn{$\mathcal{E}, P$}
							\KwOut{$\mathcal{E}_{ref}$}
							
							$\mathcal{C}_{div} \gets \emptyset, \mathcal{C}_{ndiv} \gets \emptyset$
							
							\For{$Q \in [P]_{\mathcal{E}}$}{
								$isDiv \gets \mathsf{DetDivTree}(Q, \mathcal{E})$
								
								\eIf{$isDiv = \mathbf{T}$}{$\mathcal{C}_{div} \gets \mathcal{C}_{div} \cup \{Q\}$}{
									$\mathcal{C}_{ndiv} \gets \mathcal{C}_{ndiv} \cup \{Q\}$
								}
							}
							
							$\mathcal{E}_{ref} \gets \mathcal{E} ~\backslash~ \{[P]_{\mathcal{E}}\} \cup \{\mathcal{C}_{div}, \mathcal{C}_{ndiv}\}$
							
							\KwRet{$\mathcal{E}_{ref}$}
						\end{algorithm2e}
					\end{minipage}
				\end{figure*}

				The following lemma shows that if two processes have different divergence properties with respect to an equivalence coarser than $\simeq^{\Delta}$, then they will keep such distinction for $\simeq^{\Delta}$.
				\begin{lemma}\label{lem-divsen}
					Given an equivalence $\mathcal{E}$ on $\PRCCS$ satisfying that $\simeq^{\Delta} \;\subseteq \mathcal{E}$ and two processes $A,B \in \PRCCS$.
					If $A \Uparrow_{\mathcal{E}}$ and $B \not\Uparrow_{\mathcal{E}}$, then $(A,B) \notin \;\simeq^{\Delta}$.
				\end{lemma}
				\begin{proof}
					We prove this lemma by contradiction.
					Assume that $(A,B) \in \;\simeq^{\Delta}$.
					Since $A \Uparrow_{\mathcal{E}}$, there exists a divergent $\epsilon$-tree $t_{\mathcal{E}}^A$ of $A$ with regard to $\mathcal{E}$.
					Since $\simeq^{\Delta} \;\subseteq \mathcal{E}$ and $\simeq^{\Delta}$ is a branching bisimulation with explicit divergence, by a similar argument as in the proof of Lemma \ref{lem-dsbb-equivalence}, we can construct a divergent $\epsilon$-tree $t_{\mathcal{E}}^B$ of $B$ with regard to $\mathcal{E}$ by induction on the structure of $t_{\mathcal{E}}^A$.
					Thus we have $B \Uparrow_{\mathcal{E}}$, which leads to a contradiction.
				\end{proof}
				
				The real challenge in designing an efficient algorithm for the branching bisimilarity with explicit divergence is to do with correctness. Here Lemma \ref{lem-divsen} plays a key role in the correctness proof (Theorem \ref{thm-DivBranBisim-Correctness}) of the partition-refinement algorithm (Algorithm \ref{algo-DivBranBisim}). Lemma \ref{lem-divsen} is a new result highly related to the notion of \emph{divergent $\epsilon$-tree preserving}, which we have not seen mentioned in the literature.  More importantly, the proof of Lemma \ref{lem-divsen} heavily relies on the new technique developed in the proof of Lemma \ref{lem-dsbb-equivalence}.
				
				\begin{theorem}[Correctness]\label{thm-DivBranBisim-Correctness}
					Given two processes $A,B \in\PRCCS$, $\mathsf{DivBranBisim}(A,B)$ returns $\mathbf{T}$ if and only if $A \simeq^{\Delta} B$.
				\end{theorem}
				
				\begin{proof}
					To  the procedure $\mathsf{DivBranBisim}(A,B)$,  let $\mathcal{E}_i$ (resp. $\mathcal{I}(\mathcal{E}_i)$, $(\mathcal{E}_{ini})_i$) be the current value of $\mathcal{E}$ (resp. $\mathcal{I}(\mathcal{E})$, $\mathcal{E}_{ini}$) at the end of the $i$-th iteration of the $\textbf{do--while}$ loop.
					It is not hard to show that all $\mathcal{E}_i$ are equivalence relations by induction.
					We then prove that $\simeq^{\Delta} \;\subseteq\; \mathcal{E}_i \;\subseteq\; (\mathcal{E}_{ini})_{i-1}$ holds for all $i \ge 1$ by induction on $i$.
					\begin{itemize}
						\item  \emph{(Base case).} We need to show that $\simeq^{\Delta} \;\subseteq\mathcal{E}_1  \;\subseteq\; (\mathcal{E}_{ini})_{0}$.
						
						$\mathcal{E}_1 \;\subseteq (\mathcal{E}_{ini})_0$ holds trivially for $(\mathcal{E}_{ini})_0 = \{R\}$. As $\mathcal{E}_1 = \mathsf{Quotient}((\mathcal{E}_{ini})_0) = (\mathcal{E}_{ini})_0 \slash \simeq$ is the set of equivalence classes of $(\mathcal{E}_{ini})_{0}$ under $\simeq$ (the branching bisimilarity), $\simeq^{\Delta} \;\subseteq\; \simeq\;=\mathcal{E}_1$.
						
						\item \emph{(Induction step).} Assume that $\simeq^{\Delta} \;\subseteq\; \mathcal{E}_i \;\subseteq\; (\mathcal{E}_{ini})_{i-1}$, we need to show that $\simeq^{\Delta} \;\subseteq\; \mathcal{E}_{i+1} \;\subseteq\; (\mathcal{E}_{ini})_{i}$.
						
						We consider the $i$-th iteration of the $\textbf{do--while}$ loop first.
						Function $\mathsf{FindDivSplit}(\mathcal{E}_i)$ returns $(divSen, A)$, where $divSen$ is the flag that indicates whether $\mathcal{E}_i$ is a branching bisimulation with explicit divergence  and $A$ is the found splitter.
						If $divSen = \mathbf{T}$, then $(\mathcal{E}_{ini})_{i} = (\mathcal{E}_{ini})_{i-1}$ holds.
						By inductive hypothesis we have $\simeq^{\Delta} \;\subseteq\; (\mathcal{E}_{ini})_{i-1}$, which implies that $\simeq^{\Delta} \;\subseteq (\mathcal{E}_{ini})_{i}$.
						If $divSen = \mathbf{F}$, then $(\mathcal{E}_{ini})_{i} \subsetneq \mathcal{E}_{i}$. Then consider any pair $(A,B)$ deleted by $\mathsf{DivRefine}$, i.e., $(A,B) \in \mathcal{E}_{i}\backslash (\mathcal{E}_{ini})_{i}$.
						According to the definition of $\mathsf{DivRefine}$, such pair $(A,B)$ must violate the divergence condition, and we may assume that $A \Uparrow_{\mathcal{E}_i}$ and $B \not\Uparrow_{\mathcal{E}_i}$.
						Since $\simeq^{\Delta} \;\subseteq \mathcal{E}_i$, $(A,B) \notin \;\simeq^{\Delta}$ follows from Lemma \ref{lem-divsen}.
						Now we see that none of the pairs deleted by $\mathsf{DivRefine}$ belongs to $\simeq^{\Delta}$, which leads to $\simeq^{\Delta} \;\subseteq (\mathcal{E}_{ini})_{i}$.
						We then consider the $(i+1)$-th iteration of the $\textbf{do--while}$ loop.
						Since the result of $\mathsf{Quotient}((\mathcal{E}_{ini})_i)$ is a refinement of $(\mathcal{E}_{ini})_i$, we have $\mathcal{E}_{i+1} = \mathsf{Quotient}((\mathcal{E}_{ini})_i) \subseteq (\mathcal{E}_{ini})_{i}$.
						Since any pair $(A, B)$ deleted by $\mathsf{Quotient}$ must violate the branching bisimulation conditions, we have $(A, B) \notin\; \simeq^{\Delta}$.
						Therefore $\simeq^{\Delta} \;\subseteq\mathcal{E}_{i+1}$.
					\end{itemize}
					
					The $\textbf{do--while}$ loop in procedure $\mathsf{DivBranBisim}(A,B)$ proceeds to $(i+1)$-th iteration iff the flag $divSen = \mathbf{F}$ after $i$-th iteration, or equivalently iff $\mathcal{E}_{i+1} \subsetneq \mathcal{E}_{i}$.
					Now we have that $\mathcal{E}_{0} \supsetneq \mathcal{E}_{1} \supsetneq \cdots \supsetneq \mathcal{E}_{i} \supsetneq\cdots$.
					In the light of the facts that $\simeq^{\Delta} \;\subseteq\; \mathcal{E}_i $ holds for all $i \ge 0$ and that all $\mathcal{E}_i$ are finite sets, the chain $\{\mathcal{E}_{i}\}_{i\in \mathsf{N}}$ must end up with some $\mathcal{E}_{n}$ satisfying $\simeq^{\Delta} \;\subseteq\; \mathcal{E}_n$, which assures the termination of $\mathsf{DivBranBisim}(A,B)$.
					Now since the $\textbf{do--while}$ loop terminates in $n$-th iteration, it must be the case that any pair $(A,B) \in \mathcal{E}_{n}$ satisfies both branching bisimulation  and divergence-sensitive conditions.
					By definition, $\mathcal{E}_{n}$ is a branching bisimulation with explicit divergence  and $\mathcal{E}_{n} \subseteq\; \simeq^{\Delta}$.
					Combining the fact $\simeq^{\Delta} \;\subseteq\; \mathcal{E}_n$ and $\mathcal{E}_{n} \subseteq\; \simeq^{\Delta}$, we conclude that $\simeq^{\Delta} \;=\; \mathcal{E}_n$.
					Now it should be clear that $A \simeq^{\Delta} B$ iff $(A,B) \in \mathcal{E}_n$ iff the procedure $\mathsf{DivBranBisim}(A,B)$ returns $\mathbf{T}$.
				\end{proof}
				
				\begin{proposition}
					[Complexity]\label{thm-DivBranBisim-Complexity}
					Let $N$ be the number of processes reachable from $A$ and $B$.
					The algorithm $\mathsf{DivBranBisim}(A,B)$ runs in polynomial time with respect to $N$.
				\end{proposition}
				\begin{proof}
					As is shown in the proof of Theorem \ref{thm-DivBranBisim-Correctness}, $\mathcal{E}_{i+1} \subsetneq \mathcal{E}_{i}$ holds for all $i < n$, where $n$ is the number of iterations of the $\textbf{do--while}$ loop in procedure $\mathsf{DivBranBisim}(A,B)$.
					Now it is easy to see that $n \le |\mathcal{E}_0| \le N^2$.
					Let $Q(N)$ be the complexity of $\mathsf{Quotient}$, which is shown to be polynomial in $N$ in \cite{zhang_UniformRandomProcess_2019}.
					For procedure $ \mathsf{FindDivSplit}(\mathcal{E})$, the $\textbf{for}$ loop can run for no more than $|\mathcal{E}| = \mathcal{O}(N^2)$ times.
					For procedure $\mathsf{DetDivTree}(A,\mathcal{E})$, the outer $\textbf{do--while}$ loop can repeat for no more than $|V_A^{\mathcal{E}}| \le N$ times; the loop body detects all the state-preserving transitions in the $\epsilon$-graph $(V_A^{\mathcal{E}}, E_A^{\mathcal{E}})$, which leads to $\mathcal{O}(N^2)$ complexity; thus the time complexity for $\mathsf{DetDivTree}$ is $\mathcal{O}(N^3)$.
					Therefore the time complexity for $\mathsf{FindDivSplit}(\mathcal{E})$ is $\mathcal{O}(N^5)$.
					Similarly, we can show that the time complexity for $\mathsf{DivRefine}$ is $\mathcal{O}(N^4)$.
					Thus the overall complexity of the algorithm $\mathsf{DivBranBisim}(A,B)$ is $\mathcal{O}(N^2(Q(N) + N^5 + N^4)) = \mathcal{O}(N^2\cdot Q(N)+N^7)$, i.e., polynomial in $N$.
				\end{proof}

	\subsection{Algorithm for deciding exhaustive branching bisimilarity}\label{sec-alg4ebb}
	In this part, we focus on the decision algorithm for $\simeq_{e}$.
	We start with the following definition.
	
	\begin{definition}[Maximal $\tau$-EC]\label{def-MtEC}
		Suppose $B\in \PRCCS$, and let $G_B = (V_B, E_B)$ be the induced transition graph of $B$, where $V_B$ is the set of all processes  reachable from $B$. A $\tau$-EC $\mathsf{ec}=(V,E)$ of $B$ is called \emph{maximal} if there is no other $\tau$-EC $\mathsf{ec}'=(V',E')$ such that $(V,E) \subsetneq (V',E')$.
		We usually use $\mathsf{mec}=(V,E)$ to denote a maximal $\tau$-EC.
	\end{definition}

	\begin{definition}[Maximal $\tau$-EC invariant]
		Let $\mathcal{E}$ be an equivalence on $\PRCCS$. $\mathcal{E}$ is \emph{maximal $\tau$-EC invariant} if for all $(A,B)\in\mathcal{E}$ the following holds: whenever $A \xRightarrow{} \circlearrowleft_{\mathsf{mec}_1}$ for a  maximal $\tau$-EC $\mathsf{mec}_1$, then $B \xRightarrow{} \circlearrowleft_{\mathsf{mec}_2}$ for some maximal $\tau$-EC $\mathsf{mec}_2$ such that  $\mathsf{mec}_1 \;\mathcal{E}^{\ddagger}\; \mathsf{mec}_2$.
	\end{definition}
	
	The connection between $\tau$-EC invariant and maximal $\tau$-EC invariant can be stated in the following lemma. Its proof relies on the simple observation:  each maximal $\tau$-EC is itself a $\tau$-EC and each $\tau$-EC is contained in some maximal $\tau$-EC.
	
				\begin{lemma}\label{lem-EcAndTauEc}
		Let $\mathcal{E}$ be an equivalence on $\PRCCS$. $\mathcal{E}$ is $\tau$-EC invariant iff it is maximal $\tau$-EC invariant.
	\end{lemma}
	
	With the help of Lemma \ref{lem-EcAndTauEc}, the correctness of the following proposition should be clear.
	
	\begin{proposition}\label{prop-mec}
		An equivalence $\mathcal{E}$ on $\PRCCS$ is an exhaustive weak bisimulation iff $\mathcal{E}$ is a weak bisimulation and maximal $\tau$-EC invariant.
	\end{proposition}
	
	Proposition \ref{prop-mec} allows us to focus on maximal $\tau$-ECs (rather than all $\tau$-ECs).
	Although the number of $\tau$-ECs reachable from a process $A$ could be exponentially many, the number of \emph{maximal} $\tau$-ECs is upper bounded by $|\mathsf{Reach}_{\tau}(A)|$, where $\mathsf{Reach}_{\tau}(A)$ is the set of processes reachable from $A$ through internal actions.
	
	\SetKwRepeat{Do}{do}{while}
	\begin{algorithm2e}[H]
		\footnotesize
		\DontPrintSemicolon
		\caption{$\mathsf{CompMec}$
			{\footnotesize \CommentSty{/* compute the set of maximal $\tau$-ECs $\mathsf{MEC}_A$ of $A$ */}}
		}\label{algo-CompMec}
		\SetKwInOut{KwIn}{Input}
		\SetKwInOut{KwOut}{Output}
		
		\KwIn{$A$}
		\KwOut{$\mathsf{MEC}_A$}
		
		$(V, E) \gets \mathsf{CompTauGraph}(A)$
		{\footnotesize \CommentSty{/* compute the induced $\tau$-graph $G_A^{\tau} = (V_A^{\tau}, E_A^{\tau})$ */}}
		
		$\mathsf{MEC}_A \gets \emptyset, \mathcal{L}_{und} \gets \{(V, E)\}$
		
		\Do{$toCon = \mathbf{T}$}{
			$toCon \gets \mathbf{F}$
			
			\For{$(V',E') \in \mathcal{L}_{und}$}{			
				$scc \gets \mathsf{CompScc}((V',E'))$
				
				{\footnotesize \CommentSty{/* compute the set of strongly connected components $scc$ for graph $(V',E')$ */}}
				
				\For{$(V'',E'') \in scc$}{
					$isChange \gets  \mathbf{F}, E_{new} \gets E''$
					
					\For{$B \in V''$}{
							\For{$(B,C) \in E''$ with label $p\tau$}{
									\If{there exists some $D$ such that $(B,D) \notin E''$ with label $q\tau$}{
										{\footnotesize \CommentSty{/* if $B \xrightarrow{q\tau} D$ violates $\tau$-EC condition, then updates $E_{new}$ */}}
										
										$isChange \gets \mathbf{T}$, $toCon \gets \mathbf{T}$,
										$E_{new} = E_{new}\backslash \{(B,C)\}$
									}
								}
							}
							\eIf{$isChange =  \mathbf{F}$}{
								$\mathsf{MEC}_A \gets \mathsf{MEC}_A \cup \{(V'',E'')\}$
								{\footnotesize \CommentSty{/* add $\tau$-EC $(V'',E'')$ to $\mathsf{MEC}_A$ */}}
							}{
								$\mathcal{L}_{und} \gets \mathcal{L}_{und} \cup \{(V'',E_{new})\}$
								{\footnotesize \CommentSty{/* update the set of undecided graphs */}}			
							}
						}
					}
					$\mathcal{L}_{und} \gets \mathcal{L}_{und}\backslash\{(V',E')\}$
				}
				\KwRet{$\mathsf{MEC}_A$}
			\end{algorithm2e}
			\vspace*{5pt}
			
			The \emph{induced $\tau$-graph of $A$}, denoted by $G_A^{\tau} = (V_A^{\tau}, E_A^{\tau})$, is a subgraph of $G_A$ (where $G_A$ is the induced transition graph of $A$) satisfying that $V_A^{\tau}$ contains all processes reachable from $A$ through internal actions and $E_A^{\tau}$ contains all the corresponding transition edges.
			Now the set of maximal $\tau$-ECs of $A$, denoted by $\mathsf{MEC}_A$, can be computed by Algorithm \ref{algo-CompMec}, which is an adaption of Algorithm 3.1 of \cite{dealfaro_FormalVerificationProbabilistic_1998} in our setting and runs in polynomial time. Intuitively speaking, in each iteration of
			$\mathsf{CompMec}$, it first computes the strongly connected components of the graph and then removes those probabilistic transitions that do not satisfy the requirement of $\tau$-EC.

			The main algorithm for deciding $\simeq_{e}$ is given in Algorithm \ref{algo-ExhBranBisim}.
			$\mathsf{ExhBranBisim}(A,B)$ is similar to the one in Algorithm \ref{algo-DivBranBisim} for $\simeq^{\Delta}$, we only explain the difference here. $\mathsf{FindMecSplit}(\mathcal{E})$ (given as Algorithm \ref{algo-FindMecSplit}) checks whether there is a pair of processes $(P,Q) \in \mathcal{E}$ that violates the (maximal) $\tau$-EC invariant condition, i.e., $P \xRightarrow{} \circlearrowleft_{mec}$ and there does not exist any $mec'$  such that $Q \xRightarrow{} \circlearrowleft_{mec'}$ and $mec \;\mathcal{E}^{\ddagger}\; mec'$ (or vice versa). If there is, then the discriminating evidence $(P,mec)$ (also called \emph{mec splitter}) is returned.
			Procedure $\mathsf{MecRefine}$ (given as Algorithm \ref{algo-MecRefine}) then splits the equivalence class $[P]_{\mathcal{E}}$ into two new equivalence classes $\mathcal{C}_{T}$ and $\mathcal{C}_{F}$ according to the splitter $(P,mec)$ returned by $\mathsf{FindMecSplit}$. More specifically, $\mathcal{C}_{T}$ contains all processes $P' \in [P]_{\mathcal{E}}$ that can arrive at a related maximal $\tau$-EC of $mec$, while $\mathcal{C}_{F}$ contains all processes $P'' \in [P]_{\mathcal{E}}$ that cannot.
			
			\begin{algorithm2e}[htb]
				\DontPrintSemicolon
				\footnotesize
				\caption{$\mathsf{ExhBranBisim}$
					{\footnotesize \CommentSty{/* decide whether $A \simeq_e B$ */}}
				}\label{algo-ExhBranBisim}
				\SetKwInOut{KwIn}{Input}
				\SetKwInOut{KwOut}{Output}
				
				\KwIn{$A,B$}
				\KwOut{$b\in\{\mathbf{T},\mathbf{F}\}$}
				
				$R \gets \mathsf{Reach}(A) \uplus  \mathsf{Reach}(B)$
				{\footnotesize \CommentSty{/* $\mathsf{Reach}(P)$ returns the set of processes reachable from $P$ */}}
				
				$\mathcal{E}_{ini} \gets \{R\}, toCon \gets \mathbf{T}$
				
				\Do{$toCon = \mathbf{T}$}{
					$\mathcal{E} \gets \mathsf{Quotient}(\mathcal{E}_{ini})$
					
					{\footnotesize \CommentSty{/* $\mathsf{Quotient}(\mathcal{E}_{ini})$ computes the largest branching bisimulation contained in $\mathcal{E}_{ini}$ */}}
					
					$(divSen, (P, mec))\gets \mathsf{FindMecSplit}(\mathcal{E})$
					
					{\footnotesize \CommentSty{/* $\mathsf{FindMecSplit}(\mathcal{E})$ checks whether there is a mec splitter $(P, mec)$ of $\mathcal{E}$ */}}
					
					\eIf{$divSen = \mathbf{T}$}{
						$toCon \gets \mathbf{F}$
					}{
						$\mathcal{E}_{ini} \gets \mathsf{MecRefine}(\mathcal{E}, (P, mec))$
						
						{\footnotesize \CommentSty{/* $\mathsf{MecRefine}(\mathcal{E}, (P, mec))$ refines $\mathcal{E}$ according to the splitter $(P, mec)$ identified by $\mathsf{FindMecSplit}(\mathcal{E})$ */}}
					}
				}
				{\footnotesize \CommentSty{/* when the \textbf{do-while} loop terminates, $\mathcal{E} = R\slash\simeq_{e}$ */}}
				
				\eIf{$(A,B) \in \mathcal{E}$}{
					\KwRet{$\mathbf{T}$}
				}{\KwRet{$\mathbf{F}$}}
			\end{algorithm2e}
			
			\begin{figure*}[htb]
				\noindent\begin{minipage}[t]{0.62\textwidth}
					\vspace{0pt}
					\begin{algorithm2e}[H]
						\footnotesize
						\DontPrintSemicolon
						\SetNoFillComment
						\caption{$\mathsf{FindMecSplit}$
						}
						\label{algo-FindMecSplit}
						\SetKwInOut{KwIn}{Input}
						\SetKwInOut{KwOut}{Output}
						
						\KwIn{$\mathcal{E}$}
						\KwOut{$(divSen,(P,ec))\in\{(\mathbf{T},(\bot, \bot)), (\mathbf{F},(P, mec))\}$}
						
						$divSen \gets \mathbf{T}$
						
						\For{$(P,Q) \in \mathcal{E}$}{ \label{OuterForStart}
							$\mathsf{MEC}_P \gets \mathsf{CompMec}(P)$
							
							\For{$mec \in \mathsf{MEC}_P$}{ \label{MediumForStart}
								\If{$P \xRightarrow{} \circlearrowleft_{mec}$}{ \label{OuterIfStart}
									$\mathsf{MEC}_Q \gets \mathsf{CompMec}(Q)$
									
									\For{$mec' \in \mathsf{MEC}_Q$}{ \label{InnerForStart}
										\eIf{$Q \xRightarrow{} \circlearrowleft_{mec'}$ \textbf{and} $mec \;\mathcal{E}^{\ddagger}\; mec'$}{
											\tcc{nothing changes}
										}{
											$divSen \gets \mathbf{F}$
											
											\KwRet{$(divSen,(P, mec))$}
										}
									}\label{InnerForEnd}
								}\label{OuterIfEnd}
							}\label{MediumForEnd}
							$\mathsf{MEC}_Q \gets \mathsf{CompMec}(Q)$
							
							\For{$mec \in \mathsf{MEC}_Q$}{
								\{the symmetric statements as from line \ref{OuterIfStart} to line \ref{OuterIfEnd}\}
							}
						}\label{OuterForEnd}
						
						\KwRet{$(divSen,(\bot, \bot))$}
					\end{algorithm2e}
				\end{minipage}%
				\begin{minipage}[t]{0.38\textwidth}
					\vspace{0pt}
					\begin{algorithm2e}[H]
						\footnotesize
						\DontPrintSemicolon
						\caption{$\mathsf{MecRefine}$
						}
						\label{algo-MecRefine}
						\SetKwInOut{KwIn}{Input}
						\SetKwInOut{KwOut}{Output}
						
						\KwIn{$\mathcal{E}, (P, mec)$}
						\KwOut{$\mathcal{E}_{ref}$}
						
						$\mathcal{C}_{T} \gets \emptyset, \mathcal{C}_{F} \gets \emptyset$
						
						\For{$Q \in [P]_{\mathcal{E}}$}{
							$\mathsf{MEC}_Q \gets \mathsf{CompMec}(Q)$
							
							\For{$mec \in \mathsf{MEC}_Q$}{
								\eIf{$Q \xRightarrow{} \circlearrowleft_{mec'}$ \textbf{and} $mec \;\mathcal{E}^{\ddagger}\; mec'$}{
									$\mathcal{C}_{T} \gets \mathcal{C}_{T} \cup \{Q\}$
								}{
									$\mathcal{C}_{F} \gets \mathcal{C}_{F} \cup \{Q\}$
								}
							}
						}
						
						$\mathcal{E}_{ref} \gets \mathcal{E} ~\backslash~\{[P]_{\mathcal{E}}\} \cup \{\mathcal{C}_{T}, \mathcal{C}_{F}\}$
						
						\KwRet{$\mathcal{E}_{ref}$}
						\vspace*{3.5cm}
						\vspace*{6pt}
					\end{algorithm2e}
				\end{minipage}
			\end{figure*}
			
			\begin{theorem}[Correctness]\label{thm-ExhBranBisim-Correctness}
				Given two processes $A,B \in\PRCCS$, $\mathsf{ExhBranBisim}(A,B)$ returns $\mathbf{T}$ if and only if $A \simeq_e B$.
			\end{theorem}

			\begin{proof}
				The proof is similar to the one for Theorem \ref{thm-DivBranBisim-Correctness} and is also carried out by induction. Here we only give a sketch to show the correctness of the procedure $\mathsf{MecRefine}$.  For any pair $(A,B)$ deleted by $\mathsf{MecRefine}$ in $i$-th iteration, i.e., $(A,B) \in \mathcal{E}_{i}\backslash (\mathcal{E}_{ini})_{i}$,
				according to the construction of $\mathsf{MecRefine}$, $(A,B)$ violates the divergence condition. Suppose without loss of generality that $A \xRightarrow{} \circlearrowleft_{mec}$ and there does not exist any $mec'$  such that $B \xRightarrow{} \circlearrowleft_{mec'}$ and $mec \;(\mathcal{E}_i)^{\ddagger}\; mec'$.	Meanwhile, by induction hypothesis we have $\simeq_e \;\subseteq\; \mathcal{E}_{i}$, which implies that there does not exist any $mec''$  such that $B \xRightarrow{} \circlearrowleft_{mec''}$ and $mec \;(\simeq_e)^{\ddagger}\; mec''$.	Thus $(A,B) \notin \;\simeq_e$.	This shows that no pairs deleted by $\mathsf{MecRefine}$ belong to $\simeq_e$. It can also be verified easily that all such pair $(A,B)$ are removed by the algorithm.
			\end{proof}
			
			\begin{figure}[htb]
				\captionsetup{justification=centering}
				\begin{subfigure}[t]{0.42\textwidth}
					\centering
					\begin{tikzpicture}[on grid,scale=0.8]
						\begin{scope}[shift={(-1.5,-5)}, every node/.style={draw,circle,transform shape,inner sep=3pt}]
							\foreach \pos/\name/\text in {
								{(0-0.5,1.5)/B/{$B_1$}},
								{(1,3)/A/{$A_1$}},
								{(1,1.5)/a/{$a$}},
								{(2+0.5,1.5)/b/{$b$}}}
							\node (\name) at \pos {\text};
						\end{scope}
						
						\begin{scope}[every node/.style={transform shape, auto=left,inner sep=1pt}, >={Latex[width'=0pt .5, length=5pt]}]
							\foreach \source/ \dest / \text in {
								A/a/{$\frac{1}{3}\tau$}, A/b/{$\frac{1}{3}\tau$}}
							\path[->] (\source) edge node {\text} (\dest);
							
							\path[->] (B) edge [out=-135, in=-45, looseness=6] node {$\tau$} (B);
							\path[->] (B) edge [bend left=45] node {$\tau$} (A);
						\end{scope}	
						
						\begin{scope}[every node/.style={transform shape, auto=right,inner sep=1pt}, >={Latex[width'=0pt .5, length=5pt]}]
							\foreach \source/ \dest / \text in {
								A/B/{$\frac{1}{3}\tau$}}
							\path[->] (\source) edge node {\text} (\dest);
							
						\end{scope}	
						
						\begin{scope}[on background layer]
							\fill[blue!15] plot[smooth cycle] coordinates{($(A.90)+(0,1/3)$)
								($(A.45)+(1/4,1/4)$) ($(A.0)+(1/3,0)$) ($(A.-45)+(1/4,-1/4)$) ($(A.-90)+(0,-1/3)$) ($(B.0)+(2/3,0)$) ($(B.0)+(1/3,-1)$) ($(B.-90)+(0,-1)$) ($(B.180)+(-1/3,-1)$) ($(B.180)+(-2/3,0)$) ($(B.180)+(-1/3,1)$) ($(B.90)+(0,5/4)$)};
							\fill[red!20] plot[smooth cycle,tension=1] coordinates{($(B.90)+(0,1/8)$)
								($(B.90)+(-3/4,-6/8)$) ($(B.90)+(0,-13/8)$)
								($(B.90)+(3/4,-6/8)$)};	
						\end{scope}	
					\end{tikzpicture}
					\vspace{-5mm}
					\caption{
						\scriptsize
						\setlength{\tabcolsep}{0pt}
						\begin{tabular}{rl}
							&\\
							&\\
							$B_1$ & $= \mu X. (\tau.X + \tau. (\frac{1}{3}\tau.X \oplus \frac{1}{3}\tau.a \oplus \frac{1}{3}\tau.b)),$\\
							$A_1$ & $= \frac{1}{3}\tau.B_1 \oplus \frac{1}{3}\tau.a \oplus \frac{1}{3}\tau.b.$\\
							&
					\end{tabular}}\label{fig-ExhBranBisim-a}
				\end{subfigure}
				\hspace{1mm}
				\begin{subfigure}[t]{0.55\textwidth}
					\centering
					\begin{tikzpicture}[on grid,scale=0.8]
						\begin{scope}[shift={(8,-5)}, every node/.style={draw,circle,transform shape,inner sep=3pt}]
							\foreach \pos/\name/\text in {
								{(0-0.5,1.5)/B/{$B_3$}},
								{(1,3)/A/{$A_3$}},
								{(-0.5,3)/C/{$C_3$}},
								{(1,1.5)/a/{$a$}},
								{(2+0.5,1.5)/b/{$b$}}}
							\node (\name) at \pos {\text};
						\end{scope}
						
						\begin{scope}[every node/.style={transform shape, auto=left,inner sep=1pt}, >={Latex[width'=0pt .5, length=5pt]}]
							\foreach \source/ \dest / \text in {
								A/a/{$\frac{1}{4}\tau$}, A/b/{$\frac{1}{4}\tau$}, C/A/{$\tau$}, B/C/{$\tau$}}
							\path[->] (\source) edge node {\text} (\dest);
							
							\path[->] (B) edge [out=-135, in=-45, looseness=6] node {$\tau$} (B);
						\end{scope}	
						
						\begin{scope}[every node/.style={transform shape, auto=right,inner sep=1pt}, >={Latex[width'=0pt .5, length=5pt]}]
							\foreach \source/ \dest / \text in {
								A/B/{$\frac{1}{2}\tau$}}
							\path[->] (\source) edge node {\text} (\dest);
							
							\path[->] (C) edge [out=135, in=45, looseness=6] node {$\tau$} (C);
						\end{scope}	
						
						\begin{scope}[on background layer]
							\fill[blue!15] plot[smooth cycle] coordinates{($(C.90)+(0,1)$) ($(A.90)+(0,1/2)$)
								($(A.45)+(1/4,1/4)$) ($(A.0)+(1/3,0)$) ($(A.-45)+(1/4,-1/4)$) ($(A.-90)+(0,-1/3)$) ($(B.0)+(2/3,0)$) ($(B.0)+(1/3,-1)$) ($(B.-90)+(0,-1)$) ($(B.180)+(-1/3,-1)$) ($(B.180)+(-2/3,0)$)  ($(C.180)+(-2/3,0)$) ($(C.180)+(-1/2,3/4)$) ($(C.90)+(-1/2,1-1/8)$)};
							\fill[red!20] plot[smooth cycle,tension=1] coordinates{($(B.90)+(0,1/8)$)
								($(B.90)+(-3/4,-6/8)$) ($(B.90)+(0,-13/8)$)
								($(B.90)+(3/4,-6/8)$)};	
							\fill[red!20] plot[smooth cycle,tension=1] coordinates{($(C.-90)+(0,-1/8)$)
								($(C.-90)+(-3/4,6/8)$) ($(C.-90)+(0,13/8)$)
								($(C.-90)+(3/4,6/8)$)};	
						\end{scope}	
					\end{tikzpicture}
					\vspace{-5mm}
					\caption{
						\scriptsize
						\setlength{\tabcolsep}{0pt}
						\begin{tabular}{rl}
							&\\
							&\\
							& $B_3 = \mu X.(\tau.X + \tau. (\mu Y. (\tau.Y + \tau. (\frac{1}{2}\tau.X\oplus \frac{1}{4}\tau.a \oplus \frac{1}{4}\tau.b)))),$\\
							& $C_3 = \mu Z.(\tau.Z + \tau.(\frac{1}{2}\tau.B_3\oplus \frac{1}{4}\tau.a \oplus \frac{1}{4}\tau.b)),$\\
							& $A_3 = \frac{1}{2}\tau.B_3\oplus \frac{1}{4}\tau.a \oplus \frac{1}{4}\tau.b.$
					\end{tabular}}\label{fig-ExhBranBisim-b}
				\end{subfigure}%
				\caption{Example to illustrate Algorithm $\mathsf{ExhBranBisim}$.}\label{fig-ExhBranBisim}	
			\end{figure}
			
			\begin{example} \label{exm-algebb}
				Figure \ref{fig-ExhBranBisim} depicts two probabilistic systems with initial states $A_1$ and $A_3$, respectively.
				Now consider the execution process of $\mathsf{ExhBranBisim}(A_1, A_3)$.
				\begin{enumerate}
					\item $(\mathcal{E}_{ini})_0 \gets R = \{A_1, B_1, A_3, B_3, C_3, a, b, \mathbf{0}\}$.
					\item In the first iteration of the $\textbf{do--while}$ loop:
					\begin{enumerate}
						\item $\mathcal{E}_1 \gets \mathsf{Quotient}((\mathcal{E}_{ini})_0) = (\mathcal{E}_{ini})_0 \slash \!\simeq \;\;= \{\{A_1, B_1, A_3, B_3, C_3\}, \{a\}, \{b\}, \{\mathbf{0}\}\}$.
						\item $(divSen, (D, mec))\gets \mathsf{FindMecSplit}(\mathcal{E}_1) = (\mathbf{T},(\bot, \bot))$. Here procedure $\mathsf{FindMecSplit}$ will invoke the subroutine $\mathsf{CompMec}$ to obtain the set of maximal $\tau$-ECs. We take $\mathsf{CompMec}(A_3)$ as an example. The procedure starts by computing the set of strongly connected components,  which is the set marked in blue in Figure \ref{fig-ExhBranBisim-b}. Then it removes those probabilistic transitions which do not satisfy the requirement of $\tau$-EC and repeat the process until the final set of maximal $\tau$-ECs (marked in red in Figure \ref{fig-ExhBranBisim-b}) is obtained. It is not hard to see that $A_1$ and $A_2$ can reach equivalent (maximal) $\tau$-ECs.
						\item $toCon \gets \mathbf{F}$.
					\end{enumerate}
					\item The final partition $\mathcal{E}_1 = \{\{A_1, B_1, A_3, B_3, C_3\}, \{a\}, \{b\}, \{\mathbf{0}\}\}$ computes the relation $R \slash \simeq_{e}$. Since $(A_1, A_3) \in \mathcal{E}_1$, we conclude that these two systems are exhaustive branching bisimilar.
				\end{enumerate}
			\end{example}

			\begin{proposition}[Complexity]\label{thm-ExhBranBisim-Complexity}
				Let $N$ be the number of processes reachable from $A$ and $B$. The algorithm $\mathsf{ExhBranBisim}(A,B)$ runs in polynomial time with respect to $N$.
			\end{proposition}
			
			\begin{proof}
				As is shown in the proof of Theorem \ref{thm-ExhBranBisim-Correctness}, $\mathcal{E}_{i+1} \subsetneq \mathcal{E}_{i}$ holds for all $i < n$, where $n$ is the number of iterations of the $\textbf{do--while}$ loop in procedure $\mathsf{ExhBranBisim}(A,B)$.
				Now it is easy to see that $n \le |\mathcal{E}_0| \le N^2$.
				Let $Q(N)$ be the complexity of the procedure $\mathsf{Quotient}$, which is shown to be polynomial in $N$ in \cite{zhang_UniformRandomProcess_2019}.
				For procedure $ \mathsf{FindMecSplit}(\mathcal{E})$, the $\textbf{for}$ loop at lines \ref{OuterForStart}-\ref{OuterForEnd} can run no more than $|\mathcal{E}| = \mathcal{O}(N^2)$ times;
				since $|\mathsf{MEC}_A|, |\mathsf{MEC}_B| \le N$, both the $\textbf{for}$ loop at lines \ref{MediumForStart}-\ref{MediumForEnd} and line \ref{InnerForStart}-\ref{InnerForEnd} can repeat for no more than $\mathcal{O}(N)$ times.
				Let $S(N)$ be the complexity of the procedure $\mathsf{CompMec}$, which is shown to be polynomial in $N$ in \cite{dealfaro_FormalVerificationProbabilistic_1998}.
				Therefore, the time complexity for $\mathsf{FindMecSplit}(\mathcal{E})$ is $\mathcal{O}(N^3\cdot S(N))$.
				Similarly, we can show that the time complexity for $\mathsf{MecRefine}$ is $\mathcal{O}(N\cdot S(N))$.
				Thus the overall complexity of the algorithm $\mathsf{ExhBranBisim}(A,B)$ is $\mathcal{O}(N^2(Q(N)+N^3 \cdot S(N) +N \cdot S(N))) = \mathcal{O}(N^2\cdot Q(N)+N^5 \cdot S(N))$, which is polynomial in $N$.
			\end{proof} 

	\subsection{Algorithm for deciding exhaustive weak bisimilarity}
	In this part, we extend the results for checking exhaustive branching bisimilarity  to the weak case. The readers will notice an advantage of our way in handling divergence: as the concept of $\tau$-EC is actually independent of bisimilarities, it brings extra convenience for algorithmic re-usability.  We first recall a classical result.
		\begin{theorem}[\cite{turrini_PolynomialTimeDecision_2015}]
			\label{thm-wbAlg}
			Given two processes $A,B \in \PRCCS$. Let $S$ be the set of processes reachable from $A$ and $B$, and $N = |S|$ be the size of $S$.
			For any equivalence $\mathcal{E}$ on $S$, the largest weak bisimulation $\mathcal{E}''$ contained in $\mathcal{E}$ can be computed by a procedure $\mathsf{WeakQuotient}(\mathcal{E})$ in polynomial time of $N$.
		\end{theorem}
	
	As mentioned in Section \ref{sec-introduction}, He et al. \cite{he_DivergencesensitiveWeakProbabilistic_2023} take the  \emph{inductive  verification method} for algorithm design. More specifically, instead of directly verifying exhaustive weak bisimilarity (by using $\tau$-EC), they prove the coincidence of $\approx_e$ and the so-called \emph{inductive weak probabilistic bisimilarity} and give an algorithm for the latter equivalence. 
	The reason for such algorithm design, as mentioned in \cite{he_DivergencesensitiveWeakProbabilistic_2023}, is that there could be an exponential number of $\tau$-ECs in the transition graph. 
	However, as we use \emph{maximal} $\tau$-EC in Definition \ref{def-MtEC}, there could be only a polynomial number of maximal $\tau$-ECs, because two different  maximal $\tau$-EC must be disjoint from each other. Compared with the inductive verification approach,  maximal $\tau$-EC is a concept for graphs and thus independent of the bisimilarities. Therefore, we can reuse Algorithm~\ref{algo-CompMec} directly. 
	All we need to do is to  replace the $\mathsf{Quotient}$ procedure with the analogue $\mathsf{WeakQuotient}$ procedure for weak bisimulation (cf. Theorem \ref{thm-wbAlg}) in Algorithm \ref{algo-ExhBranBisim}. 
	Then we will obtain a polynomial algorithm $\mathsf{ExhWeakBisim}$ for  exhaustive weak bisimilarity.

	\begin{proposition}
		[Complexity]\label{prop-ExhWeakBisim-Complexity}
		Let $N$ be the number of processes reachable from $A$ and $B$.
		The algorithm $\mathsf{ExhWeakBisim}(A,B)$ runs in polynomial time with respect to $N$.
	\end{proposition}
	
	We end this section by summarizing the algorithmic results in Table \ref{table-allAlg}.
	\begin{table}[htb]
		\centering
		\begin{tabular}{|c|c|c|c|}
			\hline
			Bisimilarity &$\simeq^{\Delta}$ & $\simeq_e$ & $\approx_e$ \\
			\hline
			Algorithm & Proposition \ref{thm-DivBranBisim-Complexity}     & Proposition \ref{thm-ExhBranBisim-Complexity} &
			\cite{he_DivergencesensitiveWeakProbabilistic_2023}, Proposition \ref{prop-ExhWeakBisim-Complexity}\\
			\hline
		\end{tabular}
		\caption{Polynomial algorithms for bisimilarities.}
		\label{table-allAlg}
	\end{table}

	\section{Conclusion and future work}\label{sec-conclusion}
	The probabilistic process theory has been studied for over three decades.
	From early on it has been realized that the key issue is to reconcile the imcompatibility between the probabilistic choice and the nondeterministic choice.
	Models, equivalence relations and investigating tools have been proposed to address the issue.
	A rich theory of distribution-based equivalence is now available \cite{segala_ModelingVerificationRandomized_1995,deng_SemanticsProbabilisticProcesses_2014,turrini_PolynomialTimeDecision_2015}, and a model independent theory of probabilistic process theory has been shown to enjoy the congruence property \cite{fu_ModelIndependentApproach_2021}.
	
	A difficult topic in the classical process theory is about dealing with divergence.
	Intensive studies on this issue have revealed that a comprehensive understanding of divergence is crucial if one hopes to place the classical process theory on a firmer foundation \cite{vanglabbeek_BranchingBisimilarityExplicit_2009,liu_AnalyzingDivergenceBisimulation_2017}.
	In the probabilistic scenario, the issue of divergence becomes urgent once the basic observational theory of the probabilistic processes has been settled.
	It is the opinion of the present authors that studies on the divergence issue in the probabilistic models are still on early stage, and further research can definitely improve our understanding of the probabilistic models.
	Based upon the previous work \cite{liu_AnalyzingDivergenceBisimulation_2017,fu_ModelIndependentApproach_2021,he_DivergencesensitiveWeakProbabilistic_2023}, we have conducted in this paper a systematic study on the (divergence-sensitive) branching and weak bisimilarities for the $\mathrm{RCCS}_{fs}$ model.
	We have explored two distinct methods to handle divergence, i.e., by the existence of divergent $\epsilon$-trees (roughly, divergent with probability $1$) or by the reachability of related $\tau$-ECs (roughly, divergent with any non-zero probability).
	We have established a lattice over these bisimilarities (see Figure  \ref{Fig-spectrum}) and showed that divergent $\epsilon$-tree preserving property is stronger than $\tau$-EC invariant property.
	And finally, we have provided efficient checking algorithms for all the divergence-sensitive bisimilarities in the lattice, as summarized in Table \ref{table-allAlg}.
	
	Having done the work reported in this paper, we feel that the role of divergence needs be further clarified in several accounts.
	Here are two possible directions for future investigation.
	Firstly, similar to van Glabbeek's famous linear time-branching time spectrum, it would be valuable to  give a comprehensively comparative study on other process semantics for probabilistic models with divergence.
	Notice that when divergence is defined independent of bisimulations (such as by $\tau$-EC), the algorithms of this paper can be reused.
	It would also be interesting to generalize our approach to other popular nondeterministic probabilistic models such as MDP \cite{baier_PrinciplesModelChecking_2008,etessami_RecursiveMarkovDecision_2015,brazdil_ReachabilityRecursiveMarkov_2008}, PA \cite{segala_ModelingVerificationRandomized_1995, cattani_DecisionAlgorithmsProbabilistic_2002, turrini_PolynomialTimeDecision_2015}, and the like.
	Notice that the technique for relating the $\epsilon$-trees and the distributions is actually independent of models.
	Secondly complete axiomatization systems are available for the divergence-sensitive branching bisimulations of $\mathrm{CCS}_{fs}$~\cite{liu_CompleteAxiomatisationDivergence_2021} in the absence of probability, and also for the branching bisimulations of $\mathrm{RCCS}_{fs}$ \cite{zhang_UniformRandomProcess_2019} in the absence of divergence.
	A challenging issue is about sound and complete axiomatizations for the divergence-sensitive branching (or weak) bisimulations for $\mathrm{RCCS}_{fs}$.

%
	
	\section*{Acknowledgment}
	We thank BASICS members for their instructive discussions and feedbacks. The support from the National Natural Science Foundation of China (62072299, 62102243) is acknowledged.

\bibliographystyle{alphaurl}
\bibliography{DivBisim}

\end{document}